\documentclass[lettersize,onecolumn]{IEEEtran}
\usepackage{amsmath,amsfonts}
\usepackage{algorithmic}
\usepackage{algorithm}
\usepackage{array}
\usepackage[caption=false,font=normalsize,labelfont=sf,textfont=sf]{subfig}
\usepackage{textcomp}
\usepackage{stfloats}
\usepackage{url}
\usepackage{verbatim}
\usepackage{graphicx}
\usepackage{cite}
\usepackage{booktabs}
\usepackage{arydshln}
\usepackage{amssymb}
\usepackage{float}
\usepackage{color}
\usepackage{mathrsfs}
\newcommand{\tr}{{\mathrm{Tr}}}
\newcommand{\wt}{{\mathtt{wt}}}

\newcommand{\gf}{{\mathbb{F}}}
\newcommand{\C}{{\mathcal{C}}}
\newcommand{\cC}{{\mathcal{C}}}

\newcommand{\ord}{{\mathrm{ord}}}

\newtheorem{remark}{Remark}
\newtheorem{theorem}{Theorem}
\newtheorem{lemma}[theorem]{Lemma}

\newtheorem{proposition}[theorem]{Proposition}
\newtheorem{corollary}[theorem]{Corollary}

\newtheorem{definition}[theorem]{Definition}
\newtheorem{example}[theorem]{Example}

\allowdisplaybreaks[4]

\begin{document}

\title{New Criteria and Constructions for\\ Self-Orthogonal Codes

}

\author{Peng Wang and Ziling Heng
        % <-this % stops a space
\thanks{P. Wang and Z. Heng are with the School of Science, Chang'an University, Xi'an 710064, China,  and also with the State Key Laboratory of Integrated Services Networks, Xidian University, Xi'an 710071, China (email: wp20201115@163.com, zilingheng@chd.edu.cn).}% <-this % stops a space
\thanks{This research was supported by the National Natural Science Foundation of China under Grant 12271059.}}

\markboth{Journal of \LaTeX\ Class Files,~Vol.~, No.~, ~}%
{Shell \MakeLowercase{\textit{et al.}}: A Sample Article Using IEEEtran.cls for IEEE Journals}

\IEEEpubid{0000--0000~\copyright~2023 IEEE}

\maketitle

\begin{abstract}
Self-orthogonal codes have attracted considerable attention owing to their applications in quantum error-correcting codes, linear complementary dual codes, and a variety of other fields.
 In this paper, we construct new families of self-orthogonal codes and self-orthogonal minimal codes by establishing criteria that characterize the self-orthogonality of certain linear codes. To this end, we first establish several new criteria for linear codes arising from the defining-set construction to be self-orthogonal. More specifically, for $q > 3$, we show that the code $\mathcal{C}_D$ is self-orthogonal whenever the defining set $D$ is $G$-invariant, where $G \subseteq \mathbb{F}_q^*$ and $|G| > 2$. For $q = 2, 3$, we characterize the self-orthogonality of $\mathcal{C}_D$ using certain character sums. By combining these criteria with partial difference sets, we construct several new families of self-orthogonal codes, which lead to optimal or almost optimal quantum codes.
Secondly, via the action of a multiplicative subgroup $G \subseteq \mathbb{F}_q^*$ with $|G| > 2$ on the columns of a projective linear code, we construct self-orthogonal codes with flexible parameters. From their augmented codes, we derive quantum codes. By choosing suitable projective codes, we obtain optimal quantum codes with high parametric flexibility.
Thirdly, we employ the characteristic function method to construct linear codes and establish criterion for their self-orthogonality. Based on this, we construct several classes of self-orthogonal minimal codes that violate the Ashikhmin–Barg condition, using vectorial dual-bent functions and $s$-plateaued functions.
\end{abstract}

\begin{IEEEkeywords}
Self-orthogonal codes, minimal codes, linear codes.
\end{IEEEkeywords}

\section{Introduction}
In this section, we introduce linear codes, self-orthogonal codes, minimal codes, and then state the motivations and main contributions of this paper.

\subsection{Linear codes}
Let $q=p^e$ for a prime $p$ and a positive integer $e$, $\gf_q$ the finite field with $q$ elements, $\mathfrak{n}$ a positive integer and $\gf_{q}^\mathfrak{n}$ the vector space of the $\mathfrak{n}$-tuples over $\gf_q$. The Hamming weight of a vector $\mathbf{a}=(a_1, a_2, \cdots, a_{\mathfrak{n}})\in \gf_{q}^{\mathfrak{n}}$, denoted by $\wt(\mathbf{a})$, is the  cardinality of its support defined by supp$(\mathbf{a})=\{1\leq i \leq \mathfrak{n}: a_i\neq 0\}$. If a  nonempty set $\mathcal{C}$ is a $k$-dimensional linear subspace of $\gf_{q}^\mathfrak{n}$, then it is called an $[\mathfrak{n}, k, d]$ linear code over $\gf_{q}$, where $d$ denotes the minimum Hamming distance of $\mathcal{C}$. It is known that $d$ equals the minimum weight of nonzero codewords in $\mathcal{C}$. 
 Let $A_{i}$ denote the number of codewords with  weight $i$ in $\mathcal{C}$, where $1\leq i\leq \mathfrak{n}$. The sequence $(1, A_1, A_2,\cdots, A_\mathfrak{n})$ is called the weight distribution of $\mathcal{C}$. The weight enumerator of $\mathcal{C}$ is defined as
$$A(z)=1+A_{1}z+A_{2}z^2+ \cdots +A_{\mathfrak{n}}z^\mathfrak{n}.$$
One fundamental problem in coding theory is the construction of $[\mathfrak{n},k,d]$ linear codes, where we strive to maximize both the code rate $k/\mathfrak{n}$ and the minimum distance $d$.
However, there  exists a tradeoff among the parameters of linear codes. In the following, we present the well-known  sphere-packing  bound on linear codes over $\gf_{q}$. 
\begin{lemma}\label{sphere-packing}
Let  $\mathcal{C}$ be an $[\mathfrak{n}, k, d]$ linear code over $\gf_{q}$. Then
\begin{eqnarray*}
q^\mathfrak{n} \geq q^k \sum_{i=0}^{\lfloor\frac{d-1}{2}\rfloor}\binom{\mathfrak{n}}{i}(q-1)^i,
\end{eqnarray*}
where $\lfloor \cdot \rfloor$ denotes the floor function.
\end{lemma}

An $[\mathfrak{n}, k, d]$ linear code $\mathcal{C}$ over $\gf_q$ is said to be (distance) optimal if no $[\mathfrak{n}, k, d']$ linear code over $\gf_q$  with $d'>d$ exists, and it is called almost optimal if there exists an $[\mathfrak{n}, k, d+1]$ optimal  code over $\gf_q$. A wide range of approaches for constructing codes with good parameters can be found in existing literature \cite{Huff, Lidl, Ding}.

Now we survey two well-known generic constructions of linear codes over finite fields, referred to as the first and the second generic constructions.
Let $n$ and $e$ be positive integers such that $e\mid n$, $V_n^{(p)}$ denote an $n$-dimensional vector space over $\mathbb{F}_p$ and  $f(x)$ be a function from $V_{n}^{(p)}$ to  $\gf_{q}$ for $q=p^e$.
$V_n^{(p)}$ is an $n$-dimensional vector space over $\mathbb{F}_p$. In order to define a convenient inner product, we view $V_n^{(p)}$ as a vector space over $\mathbb{F}_{q}$ of dimension $\frac{n}{e}$. We then identify $x\in V_{n}^{(p)}$ with the tuple $(x_1,\dots,x_{\frac{n}{e}}) \in \mathbb{F}_{q}^{\frac{n}{e}}$.
For $x = (x_1,\dots,x_{\frac{n}{e}})$ and $y = (y_1,\dots,y_{\frac{n}{e}})$ in $\mathbb{F}_{q}^{\frac{n}{e}}$, we define the inner product
$$
  \langle x, y \rangle_{n/e} := x \cdot y = \sum_{i=1}^{\frac{n}{e}} x_i y_i \in \mathbb{F}_{q}.
$$
The   first   generic   construction is given by 
 \begin{eqnarray}\label{eqn-01}
 \mathcal{C}_{f}=\left\{\left(a f(x)+\langle b,x\rangle_{n/e}\right)_{x\in V_{n}^{(p)}}: (a,b)\in \gf_{q}\times V_{n}^{(p)}\right\}.
 \end{eqnarray}
 Note that $\mathcal{C}_{f}$ is a linear code over $\gf_{q}$ with length $p^n$ and dimension at most $n/e+1$. This construction is generic in the sense that
 any function $f(x)$ from $V_{n}^{(p)}$ to  $\gf_{q}$ yields a $q$-ary linear code $\mathcal{C}_{f}$. This construction was used to produce minimal codes and self-orthogonal codes in \cite{Ding1, Heng2, Jin, Tao, XuG}. 
 The second generic construction, also called the deﬁning-set method, is defined as
 \begin{eqnarray}\label{eq-1}
\mathcal{C}_{D}=\left\{\left(\langle b,d_1\rangle_{n/e}, \langle b,d_2\rangle_{n/e},\cdots, \langle b,d_\mathfrak{n}\rangle_{n/e}\right): b \in V_{n}^{(p)}\right\},
\end{eqnarray}
where $D = \left\{d_1, d_2, \cdots, d_\mathfrak{n}\right\} \subseteq V_{n}^{(p)}$ is called  deﬁning set of $\mathcal{C}_{D}$.
The ordering of the elements in $D$ does not effect the parameters and weight distribution of  $\mathcal{C}_{D}$. 
Obviously, $\mathcal{C}_{D}$ is a $q$-ary linear code of length $\mathfrak{n}$ and dimension at most $n/e$. 
This defining-set construction is generic in the sense that any set $D \subseteq V_{n}^{(p)}$ produce a linear code $\mathcal{C}_{D}$. 
In 2007, Ding et al. used this construction to obtain cyclotomic linear codes \cite{DingN}. 
Motivated by this work, a large number of good linear codes have been reported in the literature.
For instance, several families of self-orthogonal codes and minimal codes were constructed based on this construction \cite{DingK, Heng3, LiX, Tang}.

Another approach to constructing linear codes is to modify old linear codes. The augmentation technique is useful for constructing linear codes that achieve a higher code rate than the original ones. For an $[\mathfrak{n},k]$ linear code $\mathcal{C}$ over $\gf_{q}$ satisfying $\mathbf{1}\notin \mathcal{C}$, its augmented code is defined as
\begin{eqnarray}\label{eq-01}
\overline{\mathcal{C}}=\left\{\mathbf{c}+c\mathbf{1}: c\in \gf_q \right\},
\end{eqnarray}
where $\mathbf{1}=(1,1,\cdots,1) \in \gf_q$. 
It is obvious that $\overline{\mathcal{C}}$ is an $[\mathfrak{n},k+1]$ linear code which has a higher code rate than that of  $\mathcal{C}$. Generally speaking, it is difficult to determine the minimum distance of $\overline{\mathcal{C}}$, even when the parameters of the original code $\mathcal{C}$ are known. 
It is easy to see that determining the minimum distance of $\overline{\mathcal{C}}$ is equivalent to determining the complete weight distribution of  $\mathcal{C}$. 
The augmentation technique has been employed in \cite{Heng, Heng3, LiX} and many other works to construct linear codes with excellent properties.

 \subsection{Self-orthogonal codes}

For an $[\mathfrak{n}, k]$ linear code $\mathcal{C}$ over $\gf_q$, its dual code is defined as 
$$\mathcal{C}^\perp=\left\{\mathbf{c}^{\perp}\in \gf_q^n:\langle\mathbf{c}^{\perp},\mathbf{c}\rangle=0\text{ for all }\mathbf{c}\in \mathcal{C}\right\},$$
where $\langle \cdot,\cdot  \rangle$ denotes the standard inner product of two codewords in a code. It is easy to know that $\mathcal{C}^\perp$ is an $[\mathfrak{n}, \mathfrak{n}-k]$ linear code over $\mathbb{F}_q$. If $\mathcal{C}\subset\mathcal{C}^{\perp}$, then $\mathcal{C}$ is called a \emph{self-orthogonal} code. Specially, if $\mathcal{C}$ satisfies  $\mathcal{C} =\mathcal{C}^{\perp}$, then it is said to be  \emph{self-dual}. Due to their nice applications in quantum codes, lattices, and linear complementary dual codes (LCD codes for short) \cite{LingS, M, Wan1}, self-orthogonal codes have become an interesting research topic, and several infinite families have been constructed in the literature \cite{Jin,Heng3,WangJ2,YWei}.

Given a linear code, a natural question is how to verify whether it is self-orthogonal or not. By definition, a linear code with generator matrix $G$ is self-orthogonal if and only if $GG^T$ is a zero matrix, with $G^T$ the transpose of $G$. However, this criterion loses its practical utility when $G$ is of a complicated form. When $q$ is an odd prime power, the following lemma simplifies the task of proving a $q$-ary linear code to be self-orthogonal.

 \begin{lemma}\cite{Wan1}\label{orthogonal}
Let $q$ be a power of an odd prime $p$ and $\mathcal{C}$ a $q$-ary
linear code. Then $\mathcal{C}$ is self-orthogonal if and only if $\mathbf{c} \cdot \mathbf{c} = 0$ for all $\mathbf{c} \in \mathcal{C}$.
\end{lemma}

If all codewords in $\mathcal{C}$ have weights divisible by an integer $\Delta$, then $\mathcal{C}$ is said to be $\Delta$-divisible. In coding theory, divisible codes and self-orthogonal codes are closely related. For the binary and ternary cases, this relationship is characterized as follows.
 
\begin{lemma}\cite[Theorems 1.4.8 and 1.4.10]{Huff}\label{lem-binary}
Let $\mathcal{C}$ be a linear code over $\gf_p$. For $p=2$, if every codeword of $\mathcal{C}$ has weight divisible by four, then $\mathcal{C}$ is self-orthogonal. For $p=3$, $\mathcal{C}$ is self-orthogonal if and only if every codeword of $\mathcal{C}$ has weight divisible by three.
\end{lemma}

A binary linear code is said to be \emph{doubly-even} if the weights of all its codewords are divisible by four, and \emph{singly-even} if they are all divisible by two but not by four.
By Lemma \ref{lem-binary}, a doubly-even binary linear code must be self-orthogonal. However, the same does not hold for singly-even binary linear codes.

If the all-1 codeword $\mathbf{1}$ is contained in $\mathcal{C}$ and  $q$ is a power of an odd prime, Li and Heng established a sufﬁcient condition for a $q$-ary linear code to be self-orthogonal as follows.

\begin{lemma}\cite{LiX}\label{lem-self}
Let $q$ be a power of an odd prime $p$. Then any $p$-divisible linear code containing the all-1 vector over the finite field $\gf_q$ is self-orthogonal. More generally, any $p$-divisible linear code of length $n$ over $\gf_q$ containing a codeword of
weight $n$ is monomially equivalent to a self-orthogonal
code of length $n$ over $\gf_q$.
\end{lemma}

 Thanks to Lemma \ref{lem-self}, several families of  $p$-divisible self-orthogonal codes were constructed in \cite{LiX}. 
 Using Walsh spectral values, Jin, Li, and Qu recently established necessary and sufficient conditions for the self-orthogonality of codes arising from the first generic construction \cite{Jin}. Recently, Li, Chen, Jin and Qu derived a necessary and sufficient condition for a binary linear code to be self-orthogonal, yielding a complete characterization of the singly-even self-orthogonal codes produced by the first generic construction \cite{Li}. 

 \subsection{Minimal linear codes }
For any two  codewords $\mathbf{a}$ and $\mathbf{b} \in \mathcal{C}$, we say that  $\mathbf{b}$  covers $\mathbf{a}$ if $\mbox{supp}(\mathbf{a})\subseteq \mbox{supp}(\mathbf{b})$. Obviously, if $\mathbf{b}$ covers $\mathbf{a}$, then $\mathbf{b}$ covers $c\mathbf{a}$ for all $c \in \gf_{q}$. 
A nonzero codeword \(\mathbf{a} \in \mathcal{C}\) is defined to be minimal if it covers no codewords in \(\mathcal{C}\) other than scalar multiples \(c\mathbf{a}\) for \(c \in \gf_{q}\). A linear code \(\mathcal{C}\) is said to be minimal when all of its nonzero codewords possess this property.
Minimal linear codes attract considerable interest due to their important applications in secret sharing schemes \cite{DingK}.

A natural problem is to determine whether a given linear code is minimal.
Ashikhmin and Barg proposed a well-known sufficient condition for linear codes to be minimal \cite{Ashikhmin}, which is stated below.

\begin{lemma} \cite{Ashikhmin}\label{minimal}
Let $\mathcal{C}$ be a linear code  over $\gf_{q}$. Then it is minimal if 
$$\frac{\wt{_{\min}}}{\wt{_{\max}}}>\frac{q-1}{q},$$
where \(\wt_{\min}\) and \(\wt_{\max}\) denote the minimum and maximum nonzero Hamming weights of the code \(\mathcal{C}\), respectively.
\end{lemma}

The above Ashikhmin–Barg sufficient condition is highly useful for constructing minimal linear codes. However, it is merely sufficient but not necessary.
Subsequently, Ding, Heng and Zhou established necessary and sufficient conditions for binary and general $q$-ary linear codes to be minimal \cite{Ding1,Heng2}.

\begin{lemma} \cite{Ding1,Heng2}\label{minimal1}
A linear code $\mathcal{C}$ over $\gf_{q}$ is minimal if and only if 
$$\sum_{z\in \gf_{q}^*}\wt(\mathbf{a}+z\mathbf{b})\neq (q-1)\wt(\mathbf{a})-\wt(\mathbf{b}).$$
for any $\gf_q$-linearly independent codewords $\mathbf{a}, \mathbf{b} \in \mathcal{C}.$
\end{lemma}

Since minimal linear codes violating the Ashikhmin–Barg condition are rarer than those satisfying it, their construction has become a worthwhile research direction.
A number of related constructions have been presented in the literature; see, e.g., \cite{Jin, Ding1, Heng2, Mesnager, Tao, XuG}.

\subsection{Motivations and contributions of this paper}
Since self-orthogonal codes and minimal codes have important applications, many researchers have devoted significant effort to constructing such codes.
Constructing self-orthogonal codes hinges on establishing useful criteria that characterize the self-orthogonality of a linear code.
Although some criteria have been proposed in \cite{Huff, Li, LiX, Jin}, they rely heavily on either the divisibility of codes or the Walsh spectra of functions.
This motivates us to seek new criteria for self-orthogonal codes.
Furthermore, as pointed out by Jin et al. \cite{Jin}, constructing linear codes that possess both minimality and self-orthogonality is a very interesting research topic.
To the best of our knowledge, only a few constructions of linear codes that violate the Ashikhmin–Barg condition and simultaneously satisfy self-orthogonality have been reported \cite{Jin, Li}. This motivates us to provide new constructions of linear codes that possess both minimality and self-orthogonality.
These two motivations lead us to make the following contributions:
\begin{enumerate}
\item \textbf{Criteria and constructions for self-orthogonal codes via the defining-set method:}
For linear codes based on the defining-set construction, we establish new criteria for self-orthogonality. 
More specifically, for $q>3$, we show that the code $\mathcal{C}_D$ is self-orthogonal whenever the defining set $D$ is $G$-invariant, where $G \subseteq \mathbb{F}_q^*$ and $|G|>2$.
For $q=2,3$, we characterize the self-orthogonality of $\mathcal{C}_D$ using certain character sums.
Compared with those in \cite{Huff, Li, LiX, Jin}, these new criteria have the advantage of not relying on code divisibility or the Walsh spectra of functions.
Using $\mathbb{F}_q^*$-invariant partial difference sets, we construct two families of linear codes and determine their weight distributions. We further prove that these codes are self-orthogonal under additional conditions. Consequently, they yield pure quantum codes that are at least almost optimal with respect to the quantum Hamming bound.

\item \textbf{Criteria and constructions for self-orthogonal codes via projective codes:}
We present a new construction of self-orthogonal codes by the action of a multiplicative subgroup $G \subset \mathbb{F}_q^*$ $(|G|>2)$ on the columns of an arbitrary projective linear code. This construction yields self-orthogonal codes with flexible parameters. We then derive quantum codes from their augmented codes. In particular, choosing appropriate projective codes leads to optimal quantum codes with a high degree of parametric flexibility.

\item \textbf{Criteria and constructions for self-orthogonal minimal codes via the characteristic function method:}
We employ the characteristic function method to construct linear codes and establish criteria for their self-orthogonality. Based on these criteria, we construct several new classes of self-orthogonal minimal codes that violate the Ashikhmin–Barg condition, using vectorial dual-bent functions and $s$-plateaued functions.
\end{enumerate}

\section{Preliminaries}\label{sec1}
This section establishes the notation employed in this paper and provides the necessary preliminary knowledge for the subsequent discussion.

\subsection{Notation}

We adopt the following notation throughout this paper:
\begin{itemize}
  \item $p$ is a prime, and $q = p^e$ with $e$ a positive integer.
  \item $n$ and $e$ are positive integers such that $e \mid n$. Let $r := \frac{n}{e}$.
  \item $\epsilon := \sqrt{(-1)^{\frac{p-1}{2}}}$ for an odd prime $p$.
  \item $\zeta_p := e^{\frac{2\pi\sqrt{-1}}{p}}$ is a complex primitive $p$-th root of unity.
  \item $\mathbb{F}_q$ denotes the finite field with $q$ elements, and $\mathbb{F}_q^* := \mathbb{F}_q \setminus \{0\}$.
  \item $\mathbb{F}_q^n$ is the vector space of $n$-tuples over $\mathbb{F}_q$.
  \item $V_n^{(p)}$ is an $n$-dimensional vector space over $\mathbb{F}_p$. In order to define a convenient inner product, we view $V_n^{(p)}$ as a vector space over $\mathbb{F}_{p^e}$ of dimension $\frac{n}{e}$. Concretely, we fix an $\mathbb{F}_{p^e}$-basis $\{\gamma_1,\dots,\gamma_{\frac{n}{e}}\}$ of $V_n^{(p)}$, so that every $x \in V_n^{(p)}$ can be uniquely written as
$$
  x = \sum_{i=1}^{\frac{n}{e}} x_i \gamma_i, \qquad x_i \in \mathbb{F}_{p^e}.
$$
  We then identify $x$ with the tuple $(x_1,\dots,x_{\frac{n}{e}}) \in \mathbb{F}_{p^e}^{\frac{n}{e}}$.
  \item For $x = (x_1,\dots,x_{\frac{n}{e}})$ and $y = (y_1,\dots,y_{\frac{n}{e}})$ in $\mathbb{F}_{p^e}^{\frac{n}{e}}$, we define the inner product
$$
  \langle x, y \rangle_{n/e} := x \cdot y = \sum_{i=1}^{\frac{n}{e}} x_i y_i \in \mathbb{F}_{p^e},
$$
  and then the $\mathbb{F}_p$-valued inner product
$$
  \langle x, y \rangle_n := \operatorname{Tr}_{p^e/p}\bigl( \langle x, y \rangle_{n/e} \bigr),
$$
  where $\operatorname{Tr}_{p^e/p} : \mathbb{F}_{p^e} \to \mathbb{F}_p$ is the trace function defined by $\operatorname{Tr}_{p^e/p}(x)=\sum_{i=0}^{e-1}x^{p^i}$ for $x\in \mathbb{F}_{p^e} $.
  It is easy to verify that $\langle \cdot,\cdot \rangle_n$ is $\mathbb{F}_p$-bilinear and satisfies $\langle x, a y \rangle_n = \langle a x, y \rangle_n$ for all $a \in \mathbb{F}_{p^e}$.
\item For a function \( F : V_n^{(p)} \to V_m^{(p)} \) and any \( A \subseteq V_m^{(p)} \), set \( D_{F,A} = \{ x \in V_n^{(p)} : F(x) \in A \} \). In the case \( A = \{a\} \), we write \( D_{F,\{a\}} \) for brevity.
  \item For any set $A \subseteq V_n^{(p)}$ and any $b \in V_n^{(p)}$, we denote $\chi_b(A) := \sum_{x \in A} \chi_b(x)$, where $\chi_b$ is the additive character of $V_n^{(p)}$ defined by
  \[
  \chi_b(x) = \zeta_p^{\langle b, x \rangle_n}, \qquad x \in V_n^{(p)}.
  \]
  \item For $y \in \mathbb{F}_p$, $\varphi_y$ denotes an additive character of $\mathbb{F}_p$.
  \item $\eta_e$ denotes the quadratic multiplicative character of $\mathbb{F}_{p^e}$.
  \item $SQ$: the set of all squares in $\gf_{p^m}^*$ for an integer $m\geq 1$.
  \item $NSQ$: the set of all  nonsquares in $\gf_{p^m}^*$ for an integer $m\geq 1$.
\end{itemize}

\subsection{Group actions and orbits}

\begin{definition}
Let $G$ be a group and $X$ be a non-empty set. An action of $G$ on $X$ is a map 
\begin{eqnarray*}
G\times X & \rightarrow & X\\
(g, x) & \longmapsto & g\cdot x
\end{eqnarray*}
satisfying the following two axioms:
\begin{itemize}
\item{For every $x\in X$, $\mathbf{e}\cdot x=x$, where $\mathbf{e}$ is the identity element of $G$.}
\item{ For all $g, h \in G$ and $x \in X$, $(g\cdot h)\cdot x=g\cdot(h\cdot x)$.}  
\end{itemize}
\end{definition}

\begin{definition}
For an element $x \in X$, the orbit of $x$ is defined as 
\begin{eqnarray*}
G\cdot x=\{g\cdot x| g\in G\} \subseteq X.
\end{eqnarray*}
Orbits form the equivalence classes of the relation $x\sim y$ if and only if $y=g\cdot  x$ for $g \in G$. 
\end{definition}

From the definition, the orbit containing an element $x$ is the subset 
\begin{eqnarray*}
\operatorname{Orb}(x)=\{g\cdot x| g \in G\}.
\end{eqnarray*}
Clearly, the set $X$ is the union of all distinct orbits, i.e.,
\begin{eqnarray*}
X=\bigcup_{x}\operatorname{Orb}(x),
\end{eqnarray*}
where $x$ runs over a set of representatives of the distinct orbits.

If a subset $X'\subseteq X$ satisfies $g\cdot x \in X'$ for each $g \in G$ and $x \in X'$, then we say that $X'$ is $G$-\emph{invariant}. In particular, if a nonempty subset $X'$ is $G$-invariant, then $X'$ can be partitioned into a union of orbits.

%\begin{definition}
%Let $G$ be a group, $X$ and let $X'$ be two non-empty sets. Suppose $G$ acts on $X$ and also acts on $X'$. If there exists a bijection $f: X\rightarrow X'$ such that 
%\begin{eqnarray*}
%f(g(x))=g(f(x))
%\end{eqnarray*}
%for all $g \in G$ and $x \in X$, then the two actions are said to equivalent.
%\end{definition}

\subsection{Cyclotomic field}
The following lemma gives some results on the cyclotomic field $\mathbb{Q}(\zeta_p)$.
\begin{lemma}\label{lem-cyclo}\cite{Ireland}
Let $K=\mathbb{Q(}\zeta_p)$ denote the $p$-th cyclotomic field over the rational number field $\mathbb{Q}$, where $p$ is an odd prime. Then the following hold.
\begin{itemize}
  \item The ring of integers in $K$ is $O_K=\mathbb{Z}(\zeta_p)$ and $\{\zeta_p^i: 1 \leq i \leq p-1\}$ is an integral basis of $O_K$, where $\zeta_p$ is the primitive $p$-th root of complex unity.
  \item The field extension $K/\mathbb{Q}$ is Galois of degree $p-1$ and Galois group $\text{Gal}(K/\mathbb{Q})=\{\sigma_a: a \in \gf_p^*\},$
  where the automorphism $\sigma_a$ of $K$ is defined by $\sigma_a(\zeta_p)=\zeta_p^a$.
  \item The field $K$ has a unique quadratic subfield $L=\mathbb{Q}(\sqrt{p^*})$. For $1 \leq a \leq p-1$, $\sigma_a(\sqrt{p^*})=\eta_0(a)\sqrt{p^*}$, where $p^*=(-1)^{\frac{p-1}{2}}p$ and $\eta_0$ is the quadratic multiplicative character of $\gf_p$. Hence, the Galois group $\text{Gal}(L/\mathbb{Q})=\{1, \sigma_\gamma\}$, where $\gamma$ is a nonsquare in $\gf_p^*$.
\end{itemize}
\end{lemma}
According to Lemma \ref{lem-cyclo}, we have
\begin{eqnarray}\label{eq-sigma}
\sigma_a(\zeta_p^b)=\zeta_p^{ab} \mbox{ and } \sigma_a(\sqrt{p^*}^e)=\eta_0^e(a)\sqrt{p^*}^e.
\end{eqnarray}

\subsection{Characters and Gaussian sums over finite fields}
An \emph{additive character} $\lambda$ of $\gf_{q}$ is defined as a homomorphism from $\gf_{q}$ onto the multiplicative group  of $p$-th roots of complex unity. It is known that all additive characters of $\gf_q$ are given by $\{\lambda_a:a\in \gf_q\}$, where
$\lambda_{a}(x)=\zeta_p^{\tr_{q/p}(ax)}$ for  $x\in \gf_{q}$. 
In particular, $\lambda_0$ is called the \emph{trivial} additive character of $\gf_q$, and $\lambda_1$ is referred to as the \emph{canonical} additive character of $\gf_q$. 
The conjugate $\overline{\lambda}$ of $\lambda$ is defined by $\overline{\lambda}(x)=\overline{\lambda(x)}=\lambda(-x)$ for all $x\in \gf_{q}$. 

Let $\gf_{q}^{*}=\langle\alpha\rangle$. A \emph{multiplicative character} of $\gf_{q}$ is defined as the homomorphic function $\psi_j$ from $\gf_{q}^*$ to the multiplicative group  of $(q-1)$-th roots of complex unity.
It is known that all multiplicative characters of $\gf_q$ are given by $\{\psi_j:j=0, 1, \cdots, q-2\}$, where $\psi_j(\alpha^s)=\zeta_{q-1}^{sj}$ for $0\leq s\leq q-2$. 
 If $q=p^e$ is odd, then $\eta_e:=\psi_{\frac{q-1}{2}}$ is called the \emph{quadratic multiplicative character} of $\gf_{q}$. Besides, $\psi_0$ is called the trivial multiplicative character of $\gf_q$. 
For convenience, we extend the definition of $\psi$ by $\psi(0)=0$ if $\psi$ is nontrivial and $\psi(0)=1$ if $\psi$ is trivial.

For an additive character $\lambda$ and a multiplicative character $\psi$ of $\gf_{q}$, the \emph{Gaussian sum} $G(\psi,\lambda)$ over $\gf_{q}$ is defined by 
\begin{eqnarray*}
G(\psi,\lambda)=\sum_{x\in \gf_{q}^{*}}\psi(x)\lambda(x).
\end{eqnarray*}
Specially, $G(\eta,\lambda)$ is called the quadratic Gaussian sum for $\lambda\neq\lambda_{0}$.

\begin{lemma}\cite{Lidl}\label{quadGuasssum1}
The properties of Gaussian sums are given as follows:
\begin{enumerate}
\item $G(\psi,\lambda_{ab})=\overline{\psi(a)}G(\psi,\lambda_{b})$ for $a\in\gf_{q}^{*}$ and $b\in\gf_{q}$;
\item $G(\psi,\overline{\lambda})=\psi(-1)G(\psi,\lambda)$;
\item $G(\overline{\psi},\lambda)=\psi(-1)\overline{G(\psi,\lambda)}$.
\end{enumerate}
\end{lemma} 

\begin{lemma}[The quadratic Gaussian sums]\cite{Lidl}\label{quadGuasssum3}
Let $p$ be an odd prime and $q=p^e$. Let $\lambda_1$ be the canonical additive character of $\gf_{q}$. Then 
\begin{eqnarray*}
G(\eta_e,\lambda_1)&=&(-1)^{e-1}(\sqrt{-1})^{(\frac{p-1}{2})^2e}\sqrt{q}\\
 &=&\left\{
\begin{array}{lll}
(-1)^{e-1}\sqrt{q}    &   \mbox{for }p\equiv 1\pmod{4},\\
(-1)^{e-1}(\sqrt{-1})^{e}\sqrt{q}    &   \mbox{for }p\equiv 3\pmod{4}.
\end{array} \right.
\end{eqnarray*}
\end{lemma}

\begin{lemma}\cite{Heng}\label{Fourier1}
Let $\psi$ be a multiplicative character of $\gf_{q}$. Then
\begin{eqnarray*}
\psi(x)=\frac{1}{q}\sum_{\lambda\in\widehat{\gf}_{q}}G(\psi,\overline{\lambda})\lambda(x)\text{ for }x\in \gf_{q}.
\end{eqnarray*}
\end{lemma}

%\subsection{The Pless power moments}\label{Plessmoments}
%\begin{lemma}[Pless power moments]\cite{Huff}\label{Plessmoments}
%For an $[n,k,d]$ linear code $\C$ over $\gf_q$, let $A_{i}^\perp$ denote the number of codewords with weight $i$ in $\C^\perp$, where $0\leq i \leq n$. Then the following equalities hold:
%\begin{eqnarray*}
%\sum_{j=0}^{n}A_j&=&q^{k},\\
%\sum_{j=0}^{n}jA_j&=&q^{k-1}(qn-n-A_1^\perp),\\
%\sum_{j=0}^{n}j^2A_j&=&q^{k-2}\big((q-1)n(qn-n+1)-  \\ && (2qn-q-2n+2)A_1^\perp+2A_2^\perp\big),\\
%\sum_{j=0}^{n}j^3A_j&=&q^{k-3}[(q-1)n(q^2n^2-2qn^2+ 3qn-q\\&&+n^2-3n+2)
%-(3q^2 n^2-3q^2n -6qn^2\\ & &+12qn+q^2-6q+3n^2-9n+6)A_1^\perp\\
%& &+6(qn-q-n+2)A_2^\perp-6A_3^\perp].
%\end{eqnarray*}
%\end{lemma}

%The Pless power moments are very useful in determining the weight distributions of linear codes.

\subsection{Vectorial dual-bent functions}
A function $F: V_{n}^{(p)}\rightarrow V_{m}^{(p)}$ is called a vectorial $p$-ary function.
If $m=1$, the vectorial $p$-ary function is the usual $p$-ary function.

For a $p$-ary function $f: V_{n}^{(p)}\rightarrow \gf_p$, its Walsh transform is defined by
\begin{eqnarray*}
W_{f}(b)=\sum_{x \in V_{n}^{(p)}}\zeta_{p}^{f(x)-\langle b,x\rangle_n},\ b\in V_{n}^{(p)}.
\end{eqnarray*}
The inverse Walsh transform of $f(x)$ is defined as
\begin{eqnarray*}
\zeta_{p}^{f(x)}=\frac{1}{p^n}\sum_{b \in V_{n}^{(p)}}W_{f}(b)\zeta_{p}^{\langle b, x\rangle_n},\ x \in V_{n}^{(p)}.
\end{eqnarray*}
If $| W_{f}(b)|=p^{\frac{n}{2}}$ for all $b \in V_{n}^{(p)}$, then $f$ is called a $p$-ary\emph{ bent function}. 
For a $p$-ary bent function $f$, if $p=2$, then $W_{f}(b)=2^{n/2}(-1)^{f^*(a)}$; if $p$ is an odd prime, then
\begin{eqnarray*}
W_{f}(a)=\left\{
\begin{array}{ll}
\pm p^{\frac{n}{2}}\zeta_{p}^{f^{*}(a)}  &   \mbox{if  $p \equiv 1 \pmod 4$ or $n$ is even},\\
\pm \sqrt{-1}p^{\frac{n}{2}}\zeta_{p}^{f^{*}(a)}   &   \mbox{if $p \equiv 3 \pmod 4 $ and $n$ is odd},\\
\end{array} \right.
\end{eqnarray*}
where $f^*: V_{n}^{(p)}\rightarrow \gf_p$  is the \emph{dual} of $f$ \cite{Tang}. A $p$-ary function $f$ is said to be \emph{regular bent} if $W_{f}(a)=p^{\frac{n}{2}}\zeta_{p}^{f^{*}(a)}$ for any $a \in V_n^{(p)}$, and \emph{weakly regular bent} if $W_{f}(a)=\varepsilon_{f}p^{\frac{n}{2}}\zeta_{p}^{f^{*}(a)}$ for any $a \in V_n^{(p)}$, where $\varepsilon_{f} \in \{\pm1,\pm\sqrt{-1}\}$ is referred to as the \emph{sign of the walsh transform of} $f$. Otherwise, the $p$-ary bent function $f$ is said to be \emph{non-weakly regular}. 
It is known from \cite{Tang} that the dual $f^{*}$ of a $p$-ary weakly regular bent function $f$ is also a weakly regular bent function such that
\begin{eqnarray*}
 (f^*)^*(x)=f(-x),\ \varepsilon_{f^*}=\varepsilon_{f}^{-1}.
\end{eqnarray*}

A vectorial $p$-ary function $F:V_{n}^{(p)}\rightarrow V_{m}^{(p)}$ is said to be \emph{vectorial bent} if the component function $F_c$ defined as $F_{c}(x)=\langle c, F(x)\rangle_{m}$ is a $p$-ary bent function \cite{Coesmelioglu1}. It is clear that every $p$-ary bent function is vectorial bent. For a vectorial $p$-ary bent function $F: V_{n}^{(p)}\rightarrow V_{m}^{(p)}$ and any $c \in V_{m}^{(p)}\setminus\{0\}$, if there exists a vectorial bent function $G:V_{n}^{(p)}\rightarrow V_{m}^{(p)}$ satisfying $(F_{c})^{*}=G_{\sigma(c)}$,  then $F$ is said to be \emph{vectorial dual-bent}, where $(F_{c})^{*}$ is the dual of $F_c$ and $\sigma$ is some permutation over  $V_{m}^{(p)}\setminus\{0\}$. The vectorial bent function $G$ is referred to as  a \emph{vectorial dual} of $F$ which is denoted by $F^*$.

A $p$-ary function $f:V_{n}^{(p)}\rightarrow \gf_p$ is said to be of $l$-form if $f(ax)=a^lf(x)$ for any $a \in \gf_p^*$, $x \in V_{n}^{(p)}$ and some integer $l$. In \cite{Tang}, 
 the set consisting of all weakly regular bent functions $f:V_{n}^{(p)}\rightarrow \gf_p$ of $l$-form such that $f(\textbf{0})=0$ and $\gcd(l-1,p-1)=1$ is denoted by $\mathcal{RF}$.
\begin{proposition}\cite{Coesmelioglu}
Let $f: V_{n}^{(p)}\rightarrow \gf_p$ be a bent function belonging to $\mathcal{RF}$. If we view $f$ as a vectorial bent function, then
 $f$ is a vectorial dual-bent function with $(cf)^*=c^{1-d}f^*$, $c \in \gf_p^*$, and $\varepsilon_{cf}=\varepsilon_{f}$ if $n$ is even, $\varepsilon_{cf}=\varepsilon_{f}\eta_{0}(c)$ if $n$ is odd, where $(l-1)(d-1)\equiv 1\mod {(p-1)}$.
\end{proposition}

In this paper,  we  consider vectorial  dual-bent  functions  satisfying the following \textbf{Condition A} \cite{WangJ2}.

\textbf{Condition A}: Let $p$ be an odd prime and  $n, e, m, d, l$ be positive integers with $e \mid n$, $e \mid m$ such that $\frac{n}{m}\geq3$ is odd. Let $F: V_{n}^{(p)}\rightarrow \gf_{p^m}$ be  a vectorial dual bent function satisfying
 \begin{itemize}
 \item{there is a vectorial dual $F^*$ such that ${F_c}^{*}=(F^*)_{c^{1-d}}$, $c \in \gf_{p^m}^{*}$, where $\gcd(d-1,p^m-1)=1$};
 \item{$F(ax)=a^{l}F(x)$, $a \in \gf_{p^e}^{*}$, $x \in V_{n}^{(p)}$ and $F(0)=0$;}    
 \item{all component functions $F_c, c \in \gf_{p^m}^{*}$, are weakly regular with $\varepsilon_{F_{c}}=\upsilon\eta_m(c)$, $c \in \gf_{p^m}^{*}$, where $\upsilon \in \{\pm\epsilon^{m}\}$ is a constant for $\epsilon:=\sqrt{(-1)^{\frac{p-1}{2}}}$.}  
 \end{itemize}

\begin{lemma}\label{lemm-21}\cite{WangJ2}
Let $F: V_{n}^{(p)}\rightarrow \gf_{p^m}$ be  a vectorial dual bent function satisfying \textbf{Condition A}. Then the  value distributions of $F$ and $F^*$ are given by
$$|D_{F,0}|=|D_{F^*,0}|=p^{n-m},$$
$$|D_{F,i}|=p^{n-m}+\upsilon(-1)^{m-1}\epsilon^{m}\eta_m(-i)p^{\frac{n-m}{2}}, i \in \gf_{p^m}^*,$$
$$|D_{F^*,i}|=p^{n-m}+\upsilon^{-1}(-1)^{m-1}\epsilon^{m}\eta_m(-i)p^{\frac{n-m}{2}}, i \in \gf_{p^m}^*.$$
\end{lemma}

In what follows, we introduce some known classes
of vectorial dual-bent functions $F: V_{n}^{(p)}\rightarrow \gf_{p^m}$ satisfying \textbf{Condition A} for general $m$. 
Observe that if $\frac{n}{m}$ is odd, then it follows that $\epsilon^n \in \{\pm\epsilon^m\}$.
\begin{itemize}
\item{Let $p$ be an odd prime and $m, s$ be positive integers. Let $s\geq3$ be odd. All non-degenerate quadratic forms $F: \gf_{p^m}^s\rightarrow \gf_{p^m}$ or $ F: \gf_{p^{ms}}\rightarrow \gf_{p^m}$ are vectorial dual-bent functions satisfying \textbf{Condition A} with $l=d=2$ \cite{WangJ}, \cite{WangJ1}. 
The following are two families of such functions.
    \begin{enumerate}
    \item{Let $m, n$ be positive integers such that  $m\mid n, \frac{n}{m}\geq3$ is odd, $a \in \gf_{p^n}^*$. Define $F: \gf_{p^n}\rightarrow \gf_{p^m}$ by
 \begin{eqnarray}\label{Tr(x^2)1}
 F(x)=\tr_{p^n/p^m}(a x^2).
\end{eqnarray}
Then $F$ is a vectorial dual-bent function satisfying  \textbf{Condition A} with $$l=d=2,\upsilon=(-1)^{n-1}\epsilon^n\eta_n(a).$$
        }
     \item{Let $m, s$ be positive integers such that $ s\geq3$ is odd, $\alpha_i \in \gf_{p^m}^*, 1\leq i \leq s$. Define $F: \gf_{p^m}^s\rightarrow \gf_{p^m}$ by
      \begin{eqnarray*}
      F(x_1,\ldots, x_s)=\sum_{i=1}^{s}\alpha_ix_i^2.
      \end{eqnarray*} 
      Then $F$ is a vectorial dual-bent function satisfying  \textbf{Condition A} with $$l=d=2,\upsilon=(-1)^{m-1}\epsilon^{ms}\eta_m(\alpha_1\cdots\alpha_s).$$
     }       
    \end{enumerate}
    }
    \item{Let $p$ be an odd prime, $m, n', n''$ be positive integers such that $m\mid n', m\mid n'', \frac{n'}{m}$ is odd. For $i \in \gf_{p^m}$, let $H(i,x): \gf_{p^{n'}}\rightarrow \gf_{p^m}$ be defined by $H(0,x)=\tr_{p^{n'}/p^m}(\alpha_1x^2), H(i,x)=\tr_{p^{n'}/p^m}(\alpha_2x^2)$ if $i$ is a square in $\gf_{p^m}^*$, $H(i,x)=\tr_{p^{n'}/p^m}(\alpha_3x^2)$ if $i$ is a non-square in $\gf_{p^m}^*$, where $\alpha_1, \alpha_2, \alpha_3$ are all square elements  or all non-square elements in $\gf_{p^{n'}}^*$. Let $G: \gf_{p^{n''}}\times\gf_{p^{n''}}\rightarrow \gf_{p^m}$ be defined by $G(y_1,y_2)= \tr_{p^{n''}/p^m}(\beta y_1L(y_2))$, where $\beta \in \gf_{p^{n''}}^*$, $L(x)=\sum a_ix^{p^{mi}}$ is a $p^m$- polynomial over $\gf_{p^{n''}}$ inducing a permutation of $\gf_{p^{n''}}$. Let $\gamma \in \gf_{p^{n''}}^*$ and $F: \gf_{p^{n'}}\times \gf_{p^{n''}}\times \gf_{p^{n''}}\rightarrow \gf_{p^m}$ be defined by
    \begin{eqnarray*}
    F(x,y_1,y_2)=H(\tr_{p^{n''}/p^m}(\gamma y_2^2),x)+G(y_1,y_2).
    \end{eqnarray*}
    Then $F$ is a vectorial dual-bent function satisfying \textbf{Condition A} with $l=d=2$ and $\upsilon=(-1)^{n'-1}\epsilon^{n'}\eta_{n'}(\alpha_1)$ \cite{WangJ1}.
    }
\end{itemize}

\subsection{Plateaued Functions} \label{subsectionF}
If a $p$-ary function $f: V_{n}^{(p)}\rightarrow \gf_p$ takes every value of $\gf_{p}$ exactly $p^{n-1}$ times, then $f$ is said to be \emph{balanced} over $\gf_{p}$. Otherwise, $f$ is said to be \emph{unbalanced}. Moreover, $f$ is \emph{balanced} if and only if $ W_{f}(0)=0$.  If $| W_{f}(b)|\in \{0, p^{\frac{n+s}{2}}\}$ for all $b \in V_{n}^{(p)}$, where $0\leq s\leq n$, then $f$ is called $s$-\emph{plateaued}. Specially, when $s=0$, an $s$-plateaued function is a bent function. For an $s$-plateaued function $f$, its \emph{Walsh support} is defined by the set $S_f=\{b \in V_{n}^{(p)}:{\mid W_f(b) \mid}^2 = p^{n+s}$\}.
The Walsh transform of an $s$-plateaued function $f$ over $b \in V_{n}^{(p)}$ is defined as
\begin{eqnarray*}
W_{f}(a)=\left\{
\begin{array}{ll}
\pm p^{\frac{n+s}{2}}\zeta_{p}^{f^{*}(b)},0  &   \mbox{if  $p^{n+s} \equiv 1 \pmod 4$},\\
\pm \sqrt{-1}p^{\frac{n+s}{2}}\zeta_{p}^{f^{*}(b)},0   &   \mbox{if $p^{n+s} \equiv 3 \pmod 4 $},\\
\end{array} \right.
\end{eqnarray*}
where $f^*$ is a function from $S_f$ to $\gf_{p}$ called the \emph{dual} of $f$.If $W_f(b) = \epsilon p^{\frac{n+s}{2}} \zeta_p^{f^*(b)}$ for all $b \in S_f$, where $\epsilon \in \{\pm1, \pm \sqrt{-1}\}$ is called the \emph{sign of the Walsh transform} of $f$ and is independent of $b$, then $f$ is called \emph{weakly regular}. Otherwise, $f$ is called \emph{non-weakly regular}, and the sign of the Walsh transform of $f$ depends on $b$. 

\begin{definition}\label{de-sp}
Let $S$ be a subset of $V_{n}^{(p)}$ with $\sharp S= N$ and $f$ be a function from $S$ to $\gf_{p}$. If $|W_{f}(b)|=N^{\frac{1}{2}}$ for all $b \in V_{n}^{(p)}$, then $f$ is called bent relative to $S$, where $W_{f}(b)= \sum_{x \in S}\zeta_{p}^{f(x)-\langle b, x\rangle_n}$.
\end{definition}

For an $s$-plateaued function $f: V_{n}^{(p)}\rightarrow \gf_{p}$, if its dual $f^*$ is bent relative to $S_f$, then for any $b \in V_{n}^{(p)}$, the value of $W_{f^*}(b)$ is defined as
\begin{eqnarray*}
W_{f^*}(a)=\left\{
\begin{array}{ll}
\pm p^{\frac{n-s}{2}}\zeta_{p}^{f^{**}(b)},0  &   \mbox{if  $p^{n-s} \equiv 1 \pmod 4$},\\
\pm \sqrt{-1}p^{\frac{n-s}{2}}\zeta_{p}^{f^{**}(b)},0   &   \mbox{if $p^{n-s} \equiv 3 \pmod 4 $},\\
\end{array} \right.
\end{eqnarray*}
where $f^{**}(x): V_{n}^{(p)}\rightarrow \gf_{p}$ is dual of $f^{*}(x)$. Similarly, the dual $f^*$ is called weakly regular bent relative to $S_{f}$, if for all $b\in V_{n}^{(p)}$, $W_{f^*}(a)=\epsilon p^{\frac{n-s}{2}}\zeta_{p}^{f^{**}(x)}$, where $\epsilon \in \{\pm1, \pm\sqrt{-1}\}$ is independent of $b$. Otherwise, it is called non-weakly regular bent relative to $S_f$. In \cite{Ozbudak}, \"{O}zbudak et al. proved that the dual $f^*(x)$ of
a weakly regular $s$-plateaued function $f(x)$ is weakly
regular bent relative to $S_f$ and $f^{**}(x)=f(-x)$. They also proved that if $f(x)$ is a non-weakly regular $s$-plateaued function such that its dual $f^*(x)$ is bent relative to $S_f$, then $f^*(x)$ is non-weakly regular bent relative to $S_f$ and satisfies $f^{**}(x)=f(-x)$.

Let $f(x): V_{n}^{(p)}\rightarrow \gf_{p}$ be an $s$-plateaued function, $\mu =1$ if $p^{n+s}\equiv 1 \pmod{4}$ and $\mu =\sqrt{-1}$ if $p^{n+s}\equiv 3  \pmod{4}$. We define $B_{+}(f)$ and $B_{-}(f)$ as follows:
$$B_{+}(f)=\left\{b \in S_{f}: W_{f}(b)=\mu p^{\frac{n+s}{2}}\zeta_{p}^{f^*(b)}\right\},$$
$$B_{-}(f)=\left\{b \in S_{f}: W_{f}(b)=-\mu p^{\frac{n+s}{2}}\zeta_{p}^{f^*(b)}\right\}.$$
For any $b \in S_{f}$, define $\epsilon_{b}=1$ (respectively $-1$) if $b \in B_{+}(f)$ ( $B_{-}(f)$). If $f(x)$ is unbalanced, we define the type of $f(x)$ to be \emph{type} ($+$) if $\epsilon_0=1$, and \emph{type} ($-$) if $\epsilon_0=-1$. Besides, if the dual $f^*(x)$ of $f(x)$ is bent relative to $S_{f}$, then we define $B_{+}(f^*)$ and 
$B_{-}(f^*)$ as follows:
$$B_{+}(f^*)=\left\{b \in S_{f}: W_{f^*}(b)=\mu p^{\frac{n-s}{2}}\zeta_{p}^{f(-b)}\right\},$$
$$B_{-}(f^*)=\left\{b \in S_{f}: W_{f^*}(b)=-\mu p^{\frac{n-s}{2}}\zeta_{p}^{f(-b)}\right\}.$$

Meanwhile, for any $b \in V_{n}^{(p)}$, define $\epsilon^{*}_{b}=1$ (respectively $-1$) if $b \in B_{+}(f^*)$ (respectively $B_{-}(f^*)$). The type of $f^*(x)$ is defined as \emph{type} ($+$) if $\epsilon^{*}_{0}=1$,  and \emph{type} ($-$) if $\epsilon^{*}_{0}=-1$.

\begin{remark}\label{rema-21}
Let $f(x):V_{n}^{(p)}\rightarrow \gf_{p}$ be an unbalanced weakly regular $s$-plateaued function. According to \cite{Ozbudak}, we know that the types of $f(x)$ and $f^*(x)$ are the same for $p^{n+s}\equiv1 \pmod 4$ and the types of $f(x)$ and $f^*(x)$ are different for $p^{n+s}\equiv3 \pmod 4$. 
\end{remark}
 
Define $D_{f,i}=\left\{x \in V_{n}^{(p)}: f(x)=i\right\}$, $D_{f,sq}=\left\{x \in V_{n}^{(p)}: f(x) \in SQ\right\}$ and $D_{f,nsq}=\left\{x \in V_{n}^{(p)}: f(x) \in NSQ\right\}$  for $i \in \gf_{p}$. Let $\sharp S$ denote the cardinality of a set $S$.
The following lemma gives the value distribution of unbalanced $p$-ary $s$-plateaued functions.

\begin{lemma}\label{lem-11}\cite{Ozbudak}
Let $f(x):V_{n}^{(p)}\rightarrow \gf_{p}$ be an unbalanced $s$-plateaued function with $f^*(0)=j_0$. For any $j \in \gf_{p}$, define $N_{i}(f)=\sharp D_{f,j}$. Then we have the following. 
\begin{itemize}
 \item When $n+s$ is even, we have $N_{j}(f)=p^{n-1}+\epsilon_{0}(\delta_{0}(j-j_0)p-1)p^{\frac{n+s}{2}-1}$ for any $j \in \gf_{p}$.
  \item When $n+s$ is odd, we have $N_{j}(f)=p^{n-1}+\epsilon_{0}\eta_{0}(j-j_0)p^{\frac{n+s-1}{2}}$ for any $j \in \gf_{p}$.    
\end{itemize}
\end{lemma}

\begin{lemma}\label{lem-12}\cite{YWei}
Let $f(x):V_{n}^{(p)}\rightarrow \gf_{p}$ be an $s$-plateaued function whose dual $f^*(x)$ is bent relative to $S_{f}$ and $f(0)=j_0$. For any $j \in \gf_{p}$, define $N_{i}(f^*)=\sharp \left\{ x \in S_{f}: f^*(x)=j\right\}$. Then we obtain the following.
\begin{itemize}
 \item When $n+s$ is even, we have $N_{j}(f^*)=p^{n-s-1}+\epsilon_{0}^*(\delta_{0}(j-j_0)p-1)p^{\frac{n-s}{2}-1}$ for any $j \in \gf_{p}$.
  \item When $n+s$ is odd, we have $N_{j}(f^*)=p^{n-s-1}+\epsilon_{0}^*\eta_{0}(j-j_0)p^{\frac{n-s-1}{2}}$ for any $j \in \gf_{p}$.    
\end{itemize}
\end{lemma}

\begin{lemma}\label{lem-13}\cite{YWei}
Let $f(x): V_{n}^{(p)}\rightarrow \gf_{p}$ be an $s$-plateaued function whose dual $f^*(x)$ is bent relative to $S_{f}$ and $f(0)=j_0$, $\sharp B_{+}(f)=k$ with $0\leq k \leq p^{n-s}$. For any $j \in \gf_{p}$, define $c_{j}(f^*)=\sharp \left\{x \in B_{+}(f): f^*(x)=j\right\}$ and $d_{j}(f^*)=\sharp \left\{x \in B_{-}{f}: f^*(x)=j\right\}$. Then we obtain the following.
\begin{itemize}
 \item If $n+s$ is even, then for any $j \in \gf_{p}$ we have $c_{j}(f^*)=\frac{k}{p}+\frac{\epsilon_{0}^*+1}{2}(\delta_{0}(j-j_0)p-1)p^{\frac{n-s}{2}-1}$ and $d_{j}(f^*)=p^{n-s-1}-\frac{k}{p}+\frac{\epsilon_{0}^*-1}{2}(\delta_{0}(j-j_0)p-1)p^{\frac{n-s}{2}-1}$.
\item If $n+s$ is odd, then for any $j \in \gf_{p}$ we have $c_{j}(f^*)=\frac{k}{p}+\frac{\epsilon_{0}^*+\eta_{0}(-1)}{2}\eta_{0}(j-j_0)p^{\frac{n-s-2}{2}}$ and $d_{j}(f^*)=p^{n-s-1}-\frac{k}{p}+\frac{\epsilon_{0}^*-\eta_{0}(-1)}{2}\eta_{0}(j-j_0)p^{\frac{n-s-1}{2}}$.
\end{itemize}
\end{lemma}

Let $\mathscr{F}$ be the set of $p$-ary $s$-plateaued functions satisfying the following conditions \cite{YWei}:
\begin{itemize}
\item [(1)] $f(0)=0$;
\item [(2)] $f^*(x)$ is bent relative to $S_{f}$;
\item [(3)] For any $x \in V_{n}^{(p)}$, $a \in \gf_{p}^{*}$, if $x \in B_{+}(f^*)$ (respectively $B_{-}(f^*)$), then $ax \in B_{+}(f^*)$ (respectively $B_{-}(f^*)$);
\item [(4)] There exists a positive integer $t$ with $2\leq t\leq p-1$ and $\gcd (t-1,p-1)=1$ such that $f(ax)=a^{t}f(x)$ for any $a \in \gf_{p}^*$ and $x \in B_{+}(f^*)$, and there exists a positive integer $t^{'}$ with $2 \leq t^{'}\leq p-1$ and $\gcd (t^{'}-1,p-1)=1$ such that $f(ax)=a^{t^{'}}f(x)$ for any $a \in \gf_{p}^*$ and $x \in B_{-}(f^*)$.
\end{itemize}

In the following, we introduce a construction of plateaued  functions that belong to $\mathscr{F}$. Let $f^{(z)}$ be a weakly regular bent function from $V_{n_1}^{(p)}$ to $\gf_{p}$ for any $z \in \gf_{p^{n_2}}$. We consider  the variant of generalized Maiorana-McFarland (GMMF, \cite{Coesmelioglu2}) function $F: V_{n_1}^{(p)}\times \gf_{p^{n_2}}\times \gf_{p^{n_2}}\rightarrow \gf_{p}$ deﬁned by
\begin{eqnarray}\label{eqnarr-11}
F(x,y,z)=f^{(z)}(x)+\tr_{p^{n_2}/p}(yz^{l-1}),
\end{eqnarray}
where $l$ is a positive integer and $\gcd(l-1, p^{n_2}-1)=1$. Define
$$W_{+}(F)=\{z\in \gf_{p^{n_2}}: f^{(z^e)} \mbox{is of type}(+)\},$$
$$W_{-}(F)=\{z\in \gf_{p^{n_2}}: f^{(z^e)} \mbox{is of type}(-)\},$$
$$W_{+}(F^*)=\{z\in \gf_{p^{n_2}}: f^{(-z)^*} \mbox{is of type}(+)\},$$
$$W_{-}(F^*)=\{z\in \gf_{p^{n_2}}: f^{(-z)^*} \mbox{is of type}(-)\},$$
where $e$ is a  positive integer and $e(l-1)\equiv 1 \pmod {p^{n_2}-1}$. 
By the deﬁnition of Walsh transform and Remark \ref{rema-21}, we  easily 
know that $F^*(x, y, z)=f^{(y^e)^*}(x)-\tr_{p^{n_2}/p}(y^ez)$, and 
$$B_{\pm}(F)=V_{n_1}^{(p)}\times W_{\pm}(F)\times \gf_{p^{n_2}},$$
$$B_{\pm}(F^*)=V_{n_1}^{(p)}\times \gf_{p^{n_2}}\times W_{\pm}(F^*).$$

\begin{remark}\label{rema-22}
Let $f^{(0)}=0$ and $f^{(az)}=f^{(z)}$ for any $a \in \gf_{p}^*$ and $z \in \gf_{p^{n_2}}^*$. Let the types of $f^{(z)}$ be the same for all $z \in \gf_{p^{n_2}}^*$. 
If the types of $f^{(0)}$ and $f^{(z)}$ are the same, then $F(x, y, z)$ is weakly
regular. If the types of $f^{(0)}$ and $f^{(z)}$ are diﬀerent,then $F(x, y, z)$ is non-weakly regular. Let $f^{(0)}$ be of $t$-form ($t'$-form) and $f^{(z)}$ of $t'$-form ($t$-form) for $z \in \gf_{p^{n_2}}^*$, where $t'\equiv l \pmod {p-1}$ $(t\equiv l \pmod {p-1})$. By using known weakly regular bent functions, we can construct inﬁnitely many weakly regular and non-weakly regular bent functions belonging to $\mathscr{F}$ by Equation (\ref{eqnarr-11}). According to\cite[Examples 3, 4]{Ozbudak}, we can also obtain inﬁnitely many weakly regular and non-weakly regular $s$-plateaued functions belonging to $\mathscr{F}$ with $s \neq0$. 
\end{remark}

\begin{lemma}\cite{YWei}
Let $f(x): V_{n}^{(p)}\rightarrow \gf_{p}$ be an $s$-plateaued function and $f^*(x)$ be its dual. If $f(0)=0$ and $f(x)=f(-x)$ for any $x \in V_{n}^{(p)}$. Then $f^*(0)=0$.
\end{lemma}

\begin{lemma} \cite{YWei}\label{YWei}
Let $f(x):V_{n}^{(p)}\rightarrow \gf_{p}$ be an $s$-plateaued function belonging to $\mathscr{F}$, $\sharp B_{+}(f)=k$ with $0\leq k\leq p^{n-s}$, and $f^*(x)$ be its dual. 
Then we have the following.
\begin{itemize}
\item [(1)] For any $a \in \gf_{p}^*$ and $b \in V_{n}^{(p)}$, if $b \notin S_{f}$, then $ab \notin S_f$, and if $b\in B_{+}(f)$ (respectively $B_{-}(f)$), then $ab \in B_{+}(f)$ (respectively $B_{-}(f)$).
\item [(2)]  There exists a positive integer $h$ with $\gcd(h-1,p-1)=1$ such that $f^*(ab)=a^hf^*(b)$ for any $a \in \gf_{p}^*$, $b \in B_{+}(f)$, and there exists a positive integer $h'$ with $\gcd(h'-1,p-1)=1$ such that $f^*(ab)=a^{h'}f^*(b)$ for any $a \in \gf_{p}^*$, $b \in B_{-}(f)$.
 \item [(3)]   The types of $f(x)$ and $f^*(x)$ are the same if $p^{n+s}\equiv 1\pmod{4}$ and different if $p^{n+s}\equiv 3\pmod{4}$.
\end{itemize}
\end{lemma}

\subsection{Partial Difference Set}
Let \(G\) be a finite  Abelian group of order \(v\) and let \(D\) be a subset of \(G\) with \(\kappa\) elements. Denote by $D^{-1}=\{x^{-1}:x\in D\}$.
Then \(D\) is called a \((v, \kappa, \lambda, \mu)\) \emph{partial difference set} (PDS for short) in \(G\) if the multiset of expressions \(d_1 d_2^{-1}\), for \(d_1, d_2 \in D\) with \(d_1 \neq d_2\), contains each nonidentity element in \(D\) exactly \(\lambda\) times and each nonidentity element not in \(D\) exactly \(\mu\) times. 
In particular, when \(\lambda = \mu\), then \(D\) is called a difference set.

Let \( D \) be a subset of an Abelian group \( G \) such that the identity element \( \mathbf{e} \notin D \) and \( D^{-1} = D \). The Cayley graph \( \Gamma = \operatorname{Cay}(G, D) \) with connection set \( D \) is defined as follows: the vertex set is \( G \), and two vertices \( x, y \in G \) are adjacent, denoted \( x \sim y \), if and only if \( xy^{-1} \in D \). Then \( \operatorname{Cay}(G, D) \) is a strongly regular graph with parameters \( (v, \kappa, \lambda, \mu) \) if and only if \( D \) is a partial difference set (PDS) with parameters \( (v, \kappa, \lambda, \mu) \) \cite{SLMa}.

There is an important tool to characterize partial difference sets in terms of character sums.
\begin{lemma} \label{lemma-21}\cite{SLMa}
Let $G$ be an Abelian group of order $v$. Suppose that $D$ is a subset of $G$ such that $D^{-1}=D$. Suppose $\kappa, \lambda$ and $\mu$ are positive integers satisfying $\kappa^2=\mu v+(\lambda-\mu)\kappa+\gamma$, where $\gamma=\kappa-\mu$ if $\mathbf{e} \notin D$ and $\gamma=\kappa-\lambda$ if $\mathbf{e} \in D$, where $\mathbf{e}$ represents the identity element in $G$. Then $D$ is a $(v,\kappa,\lambda,\mu)$-PDS in $G$ if and only if
\begin{eqnarray*}
\chi(D)=\left\{
\begin{array}{ll}
\kappa  &   \mbox{if $\chi$ is trivial},\\
\frac{\beta\pm\sqrt{\Delta}}{2}   &   \mbox{if $\chi$ is nontrivial},
\end{array} \right.
\end{eqnarray*}
where $\chi$ is a character of $G$, $\beta=\lambda-\mu$ and $\Delta=\beta^2+4\gamma$.
\end{lemma}

In the following, let $\Delta_1=\beta^2+4(\kappa-\mu)$ if $\mathbf{e}  \notin D$ and $\Delta_2=\beta^2+4(\kappa-\lambda)$ if $\mathbf{e}  \in D$.

\begin{lemma} \label{lem-21}
Let $G$ be an Abelian group of order $v$. Suppose that $D$ is a subset of $G$ such that $D^{-1}=D$ and  $\mathbf{e} \notin D$. Suppose $\kappa, \lambda$ and $\mu$ are positive integers satisfying $\kappa^2=\mu v+(\lambda-\mu)\kappa+\gamma$, where $\gamma=\kappa-\mu$. Then $D$ is a $(v,\kappa,\lambda,\mu)$-PDS in $G$ if and only if
\begin{eqnarray*}
\chi(D)=\left\{
\begin{array}{ll}
\kappa  &   \mbox{with frequency $1$,}\\
\frac{\beta+\sqrt{\Delta_1}}{2}   &   \mbox{with frequency $\frac{1}{2}\left((v-1)-\frac{2\kappa+\beta(v-1)}{\sqrt{\Delta_1}}\right)$,}\\
\frac{\beta-\sqrt{\Delta_1}}{2}   &   \mbox{with frequency $\frac{1}{2}\left((v-1)+\frac{2\kappa+\beta(v-1)}{\sqrt{\Delta_1}}\right)$,}\\
\end{array} \right.
\end{eqnarray*}
where $\chi$ runs over the characters of $G$, $\beta=\lambda-\mu$ and $\Delta_1=\beta^2+4\gamma$.
\end{lemma}
\begin{IEEEproof}
Let $m_1$ and $m_2$ respectively denote the frequencies of  $\frac{\beta+\sqrt{\Delta_1}}{2}$ and $\frac{\beta-\sqrt{\Delta_1}}{2}$. It is clear that $m_1+m_2=v-1$ and 
\begin{eqnarray*}
\frac{\beta+\sqrt{\Delta_1}}{2}m_1+\frac{\beta-\sqrt{\Delta_1}}{2}m_2=\sum_{x \in D} \sum_{\chi \in G^*}\chi(x).
\end{eqnarray*}
Since $\mathbf{e} \notin D$,  the orthogonal relation of characters yields
\begin{eqnarray*}
\frac{\beta+\sqrt{\Delta_1}}{2}m_1+\frac{\beta-\sqrt{\Delta_1}}{2}m_2=-\kappa.
\end{eqnarray*}
Then we have
\begin{eqnarray*}
\left\{
\begin{array}{ll}
m_1+m_2=v-1,\\
\frac{\beta+\sqrt{\Delta_1}}{2}m_1+\frac{\beta-\sqrt{\Delta_1}}{2}m_2=-\kappa.
\end{array} \right.
\end{eqnarray*}
This yields $m_1=\frac{1}{2}\left((v-1)-\frac{2\kappa+\beta(v-1)}{\sqrt{\Delta_1}}\right)$ and $m_2=\frac{1}{2}\left((v-1)+\frac{2\kappa+\beta(v-1)}{\sqrt{\Delta_1}}\right)$.
The desired conclusion follows from Lemma \ref{lemma-21}.
\end{IEEEproof}

\begin{lemma}\label{lem-22}
Let $G$ be an Abelian group of order $v$. Suppose that $D$ is a subset of $G$ such that $D^{-1}=D$ and $\mathbf{e} \in D$. Suppose $\kappa, \lambda$ and $\mu$ are positive integers satisfying $\kappa^2=\mu v+(\lambda-\mu)\kappa+\gamma$ where $\gamma=\kappa-\lambda$. Then $D$ is a $(v,\kappa,\lambda,\mu)$-PDS in $G$ if and only if
\begin{eqnarray*}
\chi(D)=\left\{
\begin{array}{ll}
\kappa  &   \mbox{with frequency $1$,}\\
\frac{\beta+\sqrt{\Delta_2}}{2}   &   \mbox{with frequency $\frac{1}{2}\left((v-1)+\frac{2(v-\kappa)-\beta(v-1)}{\sqrt{\Delta_2}}\right)$,}\\
\frac{\beta-\sqrt{\Delta_2}}{2}   &   \mbox{with frequency $\frac{1}{2}\left((v-1)-\frac{2(v-\kappa)-\beta(v-1)}{\sqrt{\Delta_2}}\right)$,}\\
\end{array} \right.
\end{eqnarray*}
where $\chi$ runs over the characters of $G$, $\beta=\lambda-\mu$ and $\Delta_2=\beta^2+4\gamma$.
\end{lemma}
\begin{IEEEproof}
Let $m_1$ and $m_2$ respectively denote the frequencies of  $\frac{\beta+\sqrt{\Delta_2}}{2}$ and $\frac{\beta-\sqrt{\Delta_2}}{2}$. Note that $m_1+m_2=v-1$ and
\begin{eqnarray*}
\frac{\beta+\sqrt{\Delta_2}}{2}m_1+\frac{\beta-\sqrt{\Delta_2}}{2}m_2=\sum_{x \in D} \sum_{\chi \in G^*}\chi(x).
\end{eqnarray*}
Since $\mathbf{e} \in D$, the orthogonal relation of characters yields
\begin{eqnarray*}
\frac{\beta+\sqrt{\Delta_2}}{2}m_1+\frac{\beta-\sqrt{\Delta_2}}{2}m_2=v-\kappa.
\end{eqnarray*}
Then we have
\begin{eqnarray*}
\left\{
\begin{array}{ll}
m_1+m_2=v-1\\
\frac{\beta+\sqrt{\Delta_2}}{2}m_1+\frac{\beta-\sqrt{\Delta_2}}{2}m_2=v-\kappa.
\end{array} \right.
\end{eqnarray*}
This yields $m_1=\frac{1}{2}\left((v-1)+\frac{2(v-\kappa)-\beta(v-1)}{\sqrt{\Delta_2}}\right)$ and $m_2=\frac{1}{2}\left((v-1)-\frac{2(v-\kappa)-\beta(v-1)}{\sqrt{\Delta_2}}\right)$.
The desired conclusion follows from Lemma \ref{lemma-21}.
\end{IEEEproof}

The following presents a method for constructing partial difference sets, as given in \cite{Brouwer}. 

Let $q = p^e$, $\mathbb{F}_{q^r} = \langle \alpha_1 \rangle$, $p$ be a prime, and $r$ be a positive integer such that $2 \mid er$. For a positive integer $N$ dividing $q^r - 1$, the $N$-th cyclotomic classes of $\mathbb{F}_{q^r}$ are defined by
\[
C_i = \left\{ \alpha_1^{jN + i} : 0 \le j \le \frac{q^r - 1}{N} - 1 \right\}, \quad 0 \le i \le N-1.
\]
The set $C_0$ is a subgroup of $\mathbb{F}_{q^r}^*$ of index $N$, and $C_i = \alpha_1^i C_0$ for $0 \le i \le N-1$.

\begin{lemma}\label{lem-23}\cite{Brouwer}
Retain the notation as above, and suppose that \( N \) is a proper divisor of \( q^r - 1 \) with \( N \neq 1 \), and that there exists a positive integer \( \ell_1 \) such that \( p^{\ell_1} \equiv -1 \pmod{N} \). Choose \( \ell_1 \) to be minimal and write \( er = 2\ell_1 t \). Let \( J \subset \mathbb{Z}_N \) be a proper subset of size \( u \). If \( q \) is odd, we further assume that \( N \mid \frac{q^r - 1}{2} \) and that \( J + \frac{q^r - 1}{2} \equiv J \pmod{N} \). Define
\[
D = D_J = \bigcup_{j \in J} C_j.
\] 
Then the graph $\text{Cay}(\gf_{q^r},D)$ is strongly regular with eigenvalues
$$\kappa=|D|=\frac{q^r-1}{N}u,\ \mbox{with multiplicity} ~ 1;$$
$$\theta_1=\frac{u}{N}(-1+(-1)^t)\sqrt{q^r},\ \mbox{with multiplicity}~ q^r-1-\kappa;$$
$$\theta_2=\theta_1+(-1)^{t+1}\sqrt{q^r},\ \mbox{with multiplicity}~ \kappa.$$
To be speciﬁc, for $i\in \{0,\cdots, N-1\}$, we have 

\begin{eqnarray*}
\chi_1(\alpha^iD)=\left\{
\begin{array}{ll}
\theta_2  &   \mbox{if $\varpi^t=1, i\in -J \pmod{N}$ or $\varpi^t=-1, i\in -J +\frac{N}{2} \pmod{N}$,}\\
\theta_1   &   \mbox{ otherwise}.\\
\end{array} \right.
\end{eqnarray*}
where \begin{eqnarray*}
\varpi=\left\{
\begin{array}{ll}
-1 & \mbox{if $N$ is even and $\frac{p^{\ell_1+1}}{N}$ is odd,}\\
1 & \mbox{otherwise}.
\end{array} \right.
\end{eqnarray*}
\end{lemma}

\begin{lemma}\label{lem-24}\cite{Tao}
Retain the notation as above, and let \( D = D_J = \bigcup_{j \in J} C_j \) be as defined in Lemma~\ref{lem-23}. Then \( D \) is \( \mathbb{F}_q^* \)-invariant if and only if \( J \) is invariant under the map \( \rho: j \mapsto j + \frac{q^r - 1}{q - 1} \pmod{N} \).
\end{lemma}

\subsection{The Pless power moments}\label{Plessmoments}
\begin{lemma}[Pless power moments]\cite{Huff}\label{Plessmoments}
For an $[\mathfrak{n},k,d]$ linear code $\C$ over $\gf_q$, let $A_j$ and $A_{j}^\perp$ respectively denote the frequencies of codewords with weight $j$ in $\C$ and $\C^\perp$, where $0\leq j\leq \mathfrak{n}$. Then the following equalities hold:
\begin{eqnarray*}
\sum_{j=0}^{\mathfrak{n}}A_j&=&q^{k},\\
\sum_{j=0}^{\mathfrak{n}}jA_j&=&q^{k-1}(q\mathfrak{n}-\mathfrak{n}-A_1^\perp),\\
\sum_{j=0}^{\mathfrak{n}}j^2A_j&=&q^{k-2}\big((q-1)\mathfrak{n}(q\mathfrak{n}-\mathfrak{n}+1)-   (2q\mathfrak{n}-q-2\mathfrak{n}+2)A_1^\perp+2A_2^\perp\big),\\
\sum_{j=0}^{\mathfrak{n}}j^3A_j&=&q^{k-3}[(q-1)\mathfrak{n}(q^2\mathfrak{n}^2-2q\mathfrak{n}^2+ 3q\mathfrak{n}-q+\mathfrak{n}^2-3\mathfrak{n}+2)
-(3q^2 \mathfrak{n}^2-3q^2\mathfrak{n} -6q\mathfrak{n}^2\\ & &+12q\mathfrak{n}+q^2-6q+3\mathfrak{n}^2-9\mathfrak{n}+6)A_1^\perp+6(q\mathfrak{n}-q-\mathfrak{n}+2)A_2^\perp-6A_3^\perp].
\end{eqnarray*}
\end{lemma}

The Pless power moments are very useful for determining the weight distributions of linear codes.

\subsection{Quantum codes}
An $[[\mathfrak{n}, k, d]]_q$ quantum code $Q$ of length $\mathfrak{n}$ is a $K$-dimensional subspace of $\mathbb{C}^{q^\mathfrak{n}}$ with $k = \log_q K$, and $d$ denotes its minimum distance.

\begin{definition}
A quantum code $Q$ of minimum distance $d$ is said to be \emph{pure} if for every pair of codewords $v, v' \in Q$ and every error $\varepsilon$ satisfying $1 \le w_Q(\varepsilon) \le d-1$, the Hermitian inner product satisfies $\langle v | \varepsilon | v' \rangle = 0$, where $w_Q(\varepsilon)$ denotes the quantum weight of $\varepsilon$.
\end{definition}

Since the foundational contributions of Shor \cite{Shor}, Calderbank \cite{CRS}, and Steane \cite{Steane}, quantum codes have attracted considerable attention. For given $\mathfrak{n}$ and $d$, maximizing the dimension $k$ is a central problem. The following bound is fundamental.

\begin{lemma}[Quantum Hamming bound \cite{CRSS, FM}]\label{lem-bound}
The existence of an $[[\mathfrak{n}, k, d]]_q$ pure quantum code implies
\[
\sum_{i=0}^{\lfloor (d-1)/2 \rfloor} (q^2-1)^i \binom{\mathfrak{n}}{i} \le q^{\mathfrak{n}-k}.
\]
\end{lemma}

An $[[\mathfrak{n}, k, d]]_q$ quantum code is called \emph{dimension optimal} if no $[[\mathfrak{n}, k+1, d]]_q$ code exists for fixed $\mathfrak{n}$ and $d$; \emph{distance optimal} if no $[[\mathfrak{n}, k, d+1]]_q$ code exists for fixed $\mathfrak{n}$ and $k$; and \emph{almost optimal} if an optimal $[[\mathfrak{n}, k, d+1]]_q$ code exists. The present work is devoted to constructing pure quantum codes that are at least almost optimal, or optimal, according to the quantum Hamming bound.

For a comprehensive treatment of quantum code constructions, the reader is referred to \cite{CRS, Steane01, Steane02}. We next recall a generalization of Steane's enlargement construction due to Ling et al. \cite{LingS}.

\begin{lemma}\cite[Generalization of Steane's Enlargement Construction]{LingS}\label{lem-Steane}
Let $\mathcal{C}_1$ and $\mathcal{C}_2$ be $[\mathfrak{n}, k_1, d_1]$ and $[\mathfrak{n}, k_2, d_2]$ linear codes over $\mathbb{F}_q$, respectively. If $\mathcal{C}_1^\perp \subset \mathcal{C}_1 \subset \mathcal{C}_2$ and $k_1+2 \le k_2$, then there exists a pure quantum code with parameters
\[
[[\mathfrak{n},\, k_1+k_2-\mathfrak{n},\, \min\{d_1,\, \lceil (q+1)d_2/q \rceil\}]]_q.
\]
\end{lemma}

Consequently, the construction of good quantum codes reduces to finding self-orthogonal codes and their subcodes with large minimum distance.

\section{Self-orthogonal codes based on the defining-set construction}
The objective of this section is to construct new families of self-orthogonal codes. To this end,  we first establish several new criteria for linear codes derived from the defining-set construction to be self-orthogonal. Using these criteria in conjunction with partial difference sets, we construct several new families of self-orthogonal codes, which yield optimal or
almost optimal quantum codes.

\subsection{New criteria for self-orthogonal codes derived from the defining set construction}
From now on, we view $V_n^{(p)}$ as a vector space over $\mathbb{F}_{p^e}$ of dimension $\frac{n}{e}$.
In what follows, for $q>3$, we show that the code $\mathcal{C}_D$ over $\gf_q$ is self-orthogonal whenever the defining set $D$ is $G$-invariant, where $G \subseteq \mathbb{F}_q^*$ and $|G|>2$.

\begin{theorem} \label{The-31}
Let $q = p^e$ be a prime power with $q > 3$. Let $G$ be a multiplicative subgroup of $\mathbb{F}_q^*$ with $|G|>2$, and let $D = \{d_1, d_2, \dots, d_\mathfrak{n}\} \subseteq V_n^{(p)}$ be $G$-invariant. Then the code $\mathcal{C}_D$ defined in Equation (\ref{eq-1}) is self-orthogonal.
\end{theorem}

\begin{IEEEproof}
Since $D$ is $G$-invariant, it can be partitioned into orbits. Assume that the number of orbits is $t$ and that 
$$D = \bigcup_{i=1}^{t} \operatorname{Orb}(x_i)=\bigcup_{i=1}^{t}x_iG,$$
 where $\{x_1,x_2,\dots,x_t\}$ is a set of representatives of the distinct orbits and $x_iG:=\{x_ig:g\in G\}$. Let
 $$\mathbf{c}_1=\left(\langle b_1,d_1\rangle_{n/e}, \langle b_1,d_2\rangle_{n/e},\cdots, \langle b_1,d_\mathfrak{n}\rangle_{n/e}\right)$$ and
  $$\mathbf{c}_2=\left(\langle b_2,d_1\rangle_{n/e}, \langle b_2,d_2\rangle_{n/e},\cdots, \langle b_2,d_\mathfrak{n}\rangle_{n/e}\right)$$
be any two codewords in $\mathcal{C}_{D}$ with  $b_1,b_2 \in V_{n}^{(p)}$. Then their inner product is given by
\begin{eqnarray*}
\nonumber &
&\mathbf{c}_1\cdot\mathbf{c}_2\\
&=&\sum_{x\in D}\langle b_1,x\rangle_{n/e}\langle b_2,x\rangle_{n/e}\\
&=&\sum_{g \in G}\sum_{i=1}^{t}\langle b_1,gx_i\rangle_{n/e}\langle b_2,gx_i\rangle_{n/e}\\
&=&\sum_{i=1}^{t}\langle b_1,x_i\rangle_{n/e}\langle b_2,x_i\rangle_{n/e}\sum_{g \in G}g^2.
\end{eqnarray*}
For $G = \langle \beta \rangle$, the condition $|G|>2$ implies $\beta^2\neq 1$. 
Assume $\ord(G)=\mathbf{s}$. We have  $$\sum_{g \in G}g^2=1+\beta^2+\cdots+\beta^{2(s-1)}=\frac{1-(\beta^2)^{\mathbf{s}}}{1-\beta^2}=0.$$ 
This yields $\mathbf{c}_1\cdot\mathbf{c}_2=0$. 
\end{IEEEproof}

\begin{remark}We make the following remarks on Theorem \ref{The-31}. 
\begin{enumerate}
\item In Theorem \ref{The-31}, if $G=\gf_{q}^*$, then $\mathcal{C}_{D}$ is self-orthogonal if  $D$ is $\gf_{q}^*$-invariant for $q>3$.  
It is easy to verify that the defining sets of the linear codes in \cite[Theorem 1-3]{Mesnager1}, \cite[Theorem 15-16, Theorem 19, 22, 25]{Tang} are all $\gf_{q}^*$-invariant for $q>3$. According to Theorem \ref{The-31}, we conclude that these known linear codes are all self-orthogonal. 
\item From the proof of Theorem \ref{The-31}, we can construct shorter self-orthogonal codes by suitably puncturing the self-orthogonal code $\mathcal{C}_{D}$.  
Indeed, define  
$$ D' = \bigcup_{i \in S} \operatorname{Orb}(x_i) \subseteq D, $$  
where $S \subseteq \{x_1, x_2, \dots, x_t\}$.  
Then the corresponding linear code $\mathcal{C}_{D'}$ is also self-orthogonal.
\end{enumerate}
\end{remark}

In the following, we characterize the self-orthogonality of $\mathcal{C}_D$ using character sums for the cases $q = 2$ and $q = 3$.

\begin{theorem} \label{The-32}
Let $q=p^e=3$ and $D = \{d_1, d_2, \cdots, d_\mathfrak{n}\} \subseteq V_{n}^{(p)}$ be $\gf_{q}^*$-invariant. Then $\mathcal{C}_{D}$ defined in Equation (\ref{eq-1}) is self-orthogonal if and only if $9\mid \left(|D|-\chi_b(D)\right)$ for all $b \in V_{n}^{(p)}\backslash\{0\}$.
\end{theorem}

\begin{IEEEproof}
According to Lemma \ref{orthogonal},  $\mathcal{C}_{D}$ is self-orthogonal if and only if $\mathbf{c} \cdot \mathbf{c} = 0$ for all $\mathbf{c} \in \mathcal{C}_{D}$.
Let  $$\mathbf{c}=\left(\langle b,d_1\rangle_{n/e}, \langle b,d_2\rangle_{n/e},\cdots, \langle b,d_\mathfrak{n}\rangle_{n/e}\right)$$ 
be any codeword in $\mathcal{C}_{D}$ for $b \in V_{n}^{(p)}\backslash\{0\}$. For $a\in \gf_{3}^*$, denote by
$$N_{a}(b)=\sharp\left\{x\in D: \langle b,x\rangle_{n/e}=a\right\}.$$ 
Note that $yD=D$ for $y \in \gf_{3}^*$. By the orthogonal relation of additive characters, for $q=3$, we have 
\begin{eqnarray*}
N_{a}(b)&=&\frac{1}{q}\sum_{y \in \gf_{q}}\sum_{x \in D}\zeta_{p}^{\tr_{q/p}\left(\langle b,yx\rangle_{n/e}-ya\right)}\\
&=&\frac{1}{q}|D|+\frac{1}{q}\sum_{y \in \gf_{q}^*}\zeta_{p}^{\tr_{q/p}(-ya)}\sum_{x \in D}\zeta_{p}^{\langle b, yx\rangle_n}\\
&=&\frac{1}{q}\left(|D|+\sum_{y \in \gf_{q}^*}\varphi_1(-ya)\sum_{x \in D}\chi_{b}(x)\right)\\
&=&\frac{1}{q}\left(|D|-\sum_{x \in D}\chi_{b}(x)\right).
\end{eqnarray*}
Then we obtain
\begin{eqnarray*}
\mathbf{c}\cdot\mathbf{c}=N_{1}(b)+N_{-1}(b)=\frac{2}{3}\left(|D|-\sum_{x \in D}\chi_{b}(x)\right).
\end{eqnarray*}
This implies that $\mathcal{C}_{D}$ is self-orthogonal if and only if $9| \left(|D|-\sum_{x \in D}\chi_{b}(x)\right)$ for all $b \in V_{n}^{(p)}\backslash\{0\}$.
\end{IEEEproof}

\begin{theorem} \label{The-33}
Let $q=p^e=2$ and $D = \{d_1, d_2, \cdots, d_\mathfrak{n}\} \subseteq V_{n}^{(p)}$. Then $\mathcal{C}_{D}$ defined in Equation (\ref{eq-1})  is self-orthogonal if and only if $$8\mid \left(|D|-\chi_{b_1}(D)-\chi_{b_2}(D)+\chi_{b_1+b_2}\left(D\right)\right)$$ for all $b_1, b_2 \in  V_{n}^{(p)}\backslash \{0\}$.
\end{theorem}
\begin{IEEEproof}
Let  $$\mathbf{c_1}=\left(\langle b_1,d_1\rangle_{n/e}, \langle b_1,d_2\rangle_{n/e},\cdots, \langle b_1,d_\mathfrak{n}\rangle_{n/e}\right)$$ 
and 
$$\mathbf{c_2}=\left(\langle b_2,d_1\rangle_{n/e}, \langle b_2,d_2\rangle_{n/e},\cdots, \langle b_2,d_\mathfrak{n}\rangle_{n/e}\right)$$ 
be any two codewords in $\mathcal{C}_{D}$ for $b_1,b_2 \in V_{n}^{(p)}\backslash\{0\}$. Denote by
$$N(b_1,b_2)=\sharp\{x\in D: \langle x, b_1\rangle_{n/e}=1 ~\mbox{and}~ \langle x,b_2\rangle_{n/e}=1\}.$$
By the orthogonal relation of additive characters, we have 
\begin{eqnarray*}
\mathbf{c_1}\cdot\mathbf{c_2}&=&N(b_1,b_2)\\
&=&\frac{1}{4}\sum_{y \in \gf_{2}}\sum_{z \in \gf_{2}}\sum_{x \in D}(-1)^{\langle b_1,yx\rangle_{n/e}-y}(-1)^{\langle b_2,zx\rangle_{n/e}-z}\\
&=&\frac{1}{4}\left(|D|-\sum_{x \in D}\chi_{b_1}(x)-\sum_{x \in D}\chi_{b_2}(x)+\sum_{x \in D}\chi_{b_1+b_2}\left(x\right)\right).
\end{eqnarray*}
Hence, $\mathcal{C}_{D}$ is self-orthogonal if and only if $8\mid \left(|D|-\sum_{x \in D}\chi_{b_1}(x)-\sum_{x \in D}\chi_{b_2}(x)+\sum_{x \in D}\chi_{b_1+b_2}\left(x\right)\right)$.
\end{IEEEproof}

A vector $\mathbf{x} = (x_1, x_2, \dots, x_\mathfrak{n}) \in \mathbb{F}_q^\mathfrak{n}$ is said to be even-like provided that $\sum_{i=1}^{\mathfrak{n}} x_i = 0$ in $\mathbb{F}_q$, and odd-like otherwise.
A code is called even-like if all its codewords are even-like. The following lemma establishes a connection between the self-orthogonality of a linear code and that of its augmented code.

\begin{lemma} \label{lem-31}
Let $\mathcal{C}$ be a linear code over $\mathbb{F}_q$ with length $\mathfrak{n}$ and $\mathbf{1} \notin \mathcal{C}$, and let $\overline{\mathcal{C}}$ be its augmented code as defined in Equation~(\ref{eq-01}). Then $\overline{\mathcal{C}}$ is self-orthogonal if and only if $\mathcal{C}$ is an even-like self-orthogonal code and $p \mid \mathfrak{n}$.
\end{lemma}

\begin{IEEEproof}
Assume that $\overline{\mathcal{C}}$ is a self-orthogonal code. As $\mathcal{C}\subseteq\overline{\mathcal{C}}$, it follows that $\mathcal{C}$ is also self-orthogonal.
Since $\mathbf{1}\in \overline{\mathcal{C}}$, we know that $\mathbf{1}\cdot\mathbf{c}=\sum_{i=1}^{\mathfrak{n}}c_{i}=0$ for $\mathbf{c}=\{c_1,c_2,\cdots,c_\mathfrak{n}\} \in \mathcal{C}$. Thus $\mathcal{C}$ is an even-like code. Besides, $\mathbf{1}\in \overline{\mathcal{C}}$ implies $\mathbf{1}\cdot\mathbf{1}=\mathfrak{n}\equiv 0\pmod{p}$.
This completes the proof of necessity.

The sufficiency part has been proven in \cite[Proposition 1]{HengL}, so the desired conclusion follows.
\end{IEEEproof}

The following provides a sufficient condition for the augmented code $\overline{\mathcal{C}_{D}}$ to be self-orthogonal.

\begin{corollary}\label{corol-31}
Let $q=p^e$ be a prime power and $q>3$. Let $G$ be a multiplicative subgroup of $\gf_{q}^*$ with $|G|>2$ and $D = \{d_1, d_2, \cdots, d_\mathfrak{n}\} \subseteq V_{n}^{(p)}$ be $G$-invariant with $p \mid |D|$. Then $\mathcal{C}_{D}$ is even-like and its augmented code $\overline{\mathcal{C}_{D}}$ is self-orthogonal.
\end{corollary}

\begin{IEEEproof}
By Theorem \ref{The-31}, the linear code $\mathcal{C}_{D}$ is self-orthogonal as $D$ is $G$-invariant. Since $D$ is $G$-invariant, it can be partitioned into orbits. Assume that the number of orbits is $t$ and that 
$$D = \bigcup_{i=1}^{t} \operatorname{Orb}(x_i)=\bigcup_{i=1}^{t}x_iG,$$
 where $\{x_1,x_2,\dots,x_t\}$ is a set of representatives of the distinct orbits and $x_iG:=\{x_ig:g\in G\}$.
Let  $$\mathbf{c}=\left(\langle b,d_1\rangle_{n/e}, \langle b,d_2\rangle_{n/e},\cdots, \langle b,d_\mathfrak{n}\rangle_{n/e}\right)$$ 
be any codeword in $\mathcal{C}_{D}$ for $b \in V_{n}^{(p)}$.
Then we have
\begin{eqnarray*}
\sum_{x\in D}\langle b,x\rangle_{n/e}&=&\sum_{g \in G}\sum_{i=1}^{t}\langle b,gx_i\rangle_{n/e}\\
&=&\sum_{i=1}^{t}\langle b,x_i\rangle_{n/e}\sum_{g \in G}g=0.
\end{eqnarray*}
Thus $\mathcal{C}_{D}$ is even-like. Then the augmented code $\overline{\mathcal{C}_{D}}$ is self-orthogonal by Lemma \ref{lem-31}.
\end{IEEEproof}

\subsection{Self-orthogonal codes from  $\gf_{q}^*$-invariant partial difference sets}
In this subsection, we construct self-orthogonal codes from $\gf_{q}^*$-invariant partial difference sets. 
To this end, we first introduce several lemmas that will be needed later.

\begin{lemma}\label{lem-32}
Let $q=p^e$, $n$ and $e$ be positive integers such that $e\mid n$ and $r := \frac{n}{e}$. Let $D = \{d_1, d_2, \cdots, d_\mathfrak{n}\} \subseteq V_{n}^{(p)}$ be $\langle D\rangle=V_{n}^{(p)}$, where $\langle D\rangle$ is the $\mathbb{F}_{q}$-linear subspace spanned by $D$. Then the dimension of $\mathcal{C}_D$  defined in Equation (\ref{eq-1}) is $r$.
\end{lemma}

\begin{IEEEproof}
Let $\mathcal{B}=\{b\in V_{n}^{(p)}: \langle b,d_i\rangle_{n/e}=0~ \mbox{for all} ~d_i ~ \mbox{in} ~D\}$.
Since $\langle D\rangle=V_{n}^{(p)}$, for any $x \in V_{n}^{(p)}$, it follows that $x$ is a linear combination of elements in $D$.
Assume that $x=a_1d_1+a_2d_2+\cdots+a_\mathfrak{n}d_\mathfrak{n}$, where $a_i \in \gf_{q}$ for $1\leq i\leq \mathfrak{n}$.
Then $\langle b,x\rangle_{n/e}=\sum_{i=1}^{\mathfrak{n}}a_i\langle b,d_i\rangle_{n/e}=0$ for $b \in \mathcal{B}$. Therefore, we have $\mathcal{B}=\{0\}$ and the desired conclusion follows. 
\end{IEEEproof}

\begin{lemma}\label{lem-33}
Let $q=p^e>2$. Let $G=\langle \beta\rangle$ be a multiplicative subgroup of $\gf_{q}^*$ with $\ord(G)=\mathbf{s}>1$, and $D = \{d_1, d_2, \cdots, d_\mathfrak{n}\} \subseteq V_{n}^{(p)}$ be $G$-invariant. Then $\mathbf{1} \notin \mathcal{C}_{D}$ .
\end{lemma}

\begin{IEEEproof}
Assume that $\mathbf{1} \in \mathcal{C}_{D}$. Then there exists an element $b \in V_{n}^{(p)}$ such that $\langle b,d\rangle_{n/e}=1$ for all $d \in D$. 
Since $D$ is $G$-invariant, for any $a \in G$ and $d \in D$, we have $ad \in D$. Thus $1=\langle b,ad\rangle_{n/e}=a\langle b,d\rangle_{n/e}=a$ for all $a \in G$. This is impossible for $\ord(G)=\mathbf{s}>1$.
Therefore, we have $\mathbf{1} \notin \mathcal{C}_{D}.$
\end{IEEEproof}

\begin{lemma} \label{lem-34}
Let $q=p^e>2$ with $e\mid n$ and $r = \frac{n}{e}$. Let $D = \{d_1, d_2, \cdots, d_\mathfrak{n}\} \subseteq V_{n}^{(p)}$ be $\gf_{q}^*$-invariant, where $\langle D\rangle=V_{n}^{(p)}$.  Then $\overline{\mathcal{C}_D}$ is an $[\mathfrak{n}, r+1]$ linear code over $\gf_q$. For a codeword $$\textbf{c}(b,c)=\left(\langle b,d_1\rangle_{n/e}, \langle b,d_2\rangle_{n/e},\cdots, \langle b,d_\mathfrak{n}\rangle_{n/e}\right)+c\textbf{1}\in \overline{\mathcal{C}_D},$$ where $b \in V_{n}^{(p)}$ and $c\in \gf_q$, its Hamming weight is given by
\begin{eqnarray*}
\wt(\textbf{c}(b,c))
=\left\{
\begin{array}{ll}
0 & \mbox{if $b=c=0$,}\\
\mathfrak{n} & \mbox{if $b=0, c \neq 0$,}\\
\frac{q-1}{q}\mathfrak{n}-\frac{q-1}{q}\chi_{b}(D) & \mbox{if $b\neq0, c=0$,}\\
\frac{q-1}{q}\mathfrak{n}+\frac{1}{q}\chi_{b}(D) & \mbox{if $b\neq0, c\neq0$.}
\end{array}\right.
 \end{eqnarray*}
\end{lemma}

\begin{IEEEproof}
According to Lemma \ref{lem-32} and \ref{lem-33} ,we deduce that the dimension of $\overline{\mathcal{C}_D}$ is $r+1$.
For $(b, c) \in V_{n}^{(p)}\times \gf_{q}$, we have 
\begin{eqnarray*}
\wt(\textbf{c}(b,c))&=&|D|-\frac{1}{q}\sum_{y \in \gf_{q}}\sum_{x\in D}\zeta_{p}^{\tr_{q/p}\left(y\left(\langle b,x\rangle_{n/e}+c\right)\right)}\\
&=&|D|-\frac{1}{q}|D|-\frac{1}{q}\sum_{y \in \gf_{q}^*}\zeta_{p}^{\tr_{q/p}(yc)}\sum_{x\in D}\zeta_{p}^{\langle b,yx\rangle_{n}}\\
&=&|D|-\frac{1}{q}|D|-\frac{1}{q}\sum_{y \in \gf_{q}^*}\varphi_{1}(yc)\sum_{x \in D}\chi_{b}(yx)\\
&=&\frac{q-1}{q}|D|-\frac{1}{q}\sum_{y \in \gf_{q}^*}\varphi_{1}(yc)\sum_{x \in D}\chi_{b}(x),
\end{eqnarray*}
where the Fourth equation holds as $D$ is $\gf_{q}^*$-invariant. According to the orthogonal
relation of additive characters, we further have 
\begin{eqnarray*}
\wt(\textbf{c}(b,c))
=\left\{
\begin{array}{ll}
0 & \mbox{if $b=c=0$,}\\
|D| & \mbox{if $b=0, c \neq 0$,}\\
\frac{q-1}{q}|D|-\frac{q-1}{q}\sum_{x \in D}\chi_{b}(x) & \mbox{if $b\neq0, c=0$,}\\
\frac{q-1}{q}|D|+\frac{1}{q}\sum_{x \in D}\chi_{b}(x) & \mbox{if $b\neq0, c\neq0$.}
\end{array}\right.
\end{eqnarray*}
Then the desired conclusion follows. 
\end{IEEEproof}

% \begin{remark}
% If $0\in D$, then the Lemma \ref{lem-34} still holds for $p=2$. If $0 \notin D$ and $q=2$, when $D$ satisfies $\{b \in V_{n}^{(p)}: \langle b,x\rangle_{n/e}=1 ~\mbox{for all} ~ x \in D\}=\emptyset$ ,  then Lemma \ref{lem-34} is also still holds. In the following, when $0 \notin D$ and $q=2$, we say that $D$ satisfying $\{b \in V_{n}^{(p)}: \langle b,x\rangle_{n/e}=1 ~\mbox{for all} ~ x \in D\}=\emptyset$. 
% \end{remark}

Now we use $\gf_{q}^*$-invariant partial difference sets to construct self-orthogonal codes and determine their weight distributions.

\begin{theorem} \label{theo-31}
Let $p$ be a prime and $q=p^e$, where $n$ and $e$ are positive integers such that $e\mid n$ and $r = \frac{n}{e}$. Let $D\subseteq V_{n}^{(p)}$ be an $\gf_{q}^*$-invariant partial difference set with parameters $(p^n, \kappa, \lambda, \mu)$ such that $-D=D$, $\mathbf{e} \notin D$ and $\langle D\rangle=V_{n}^{(p)}$. Then we have the following.
\begin{enumerate}
\item $\overline{\mathcal{C}_D}$ is a $[\kappa, r+1]$ linear code with its weight distribution as shown in Table \ref{tab-31}. Besides,  $\overline{\mathcal{C}_D}^{\perp}$  has parameters $[\kappa, \kappa-r-1,\geq3]$.  In particular, if $q>3$, then $\overline{\mathcal{C}_D}^{\perp}$ has parameters $[\kappa, \kappa-r-1, 3]$ and, for $\lambda \geq 1, \mu \geq 1$ and $r \geq 2$, $\overline{\mathcal{C}_D}^{\perp}$ is at least almost optimal according to the sphere-packing bound.
\item If $q>3$ and $p\mid \kappa$, then  $\overline{\mathcal{C}_D}$ is a self-orthogonal code. When $q=3$, $\overline{\mathcal{C}_D}$ is a self-orthogonal code if $9 \mid \left(|D|-\chi_{b}(D)\right)$ for all $b \in V_{n}^{(p)} \backslash \{0\}$ and $p\mid \kappa$. When $q=2$,  $\overline{\mathcal{C}_D}$ is a self-orthogonal code if $8\mid(|D|-\chi_{b_1}(D)-\chi_{b_2}(D)+\chi_{b_1+b_2}\left(D\right))$ for all $b_1, b_2 \in V_{n}^{(p)}\backslash \{0\}$ and $p\mid \kappa$.
\end{enumerate}

\begin{table}[ht]                                                                                                                                                                                                                                                                                         
\begin{center}
\caption{The weight distribution of $\overline{\cC_{D}}$ in Theorem \ref{theo-31}}\label{tab-31}
\begin{tabular}{cc} \hline
Weight  &  Frequency   \\ \hline
$0$          &  $1$ \\
$\kappa$ & $q-1$ \\
$\kappa-\frac{1}{q}(\kappa-\frac{\beta+\sqrt{\Delta_1}}{2})$  &  $\frac{q-1}{2}\left((v-1)-\frac{2\kappa+\beta(v-1)}{\sqrt{\Delta_1}}\right)$ \\
$\kappa-\frac{1}{q}(\kappa-\frac{\beta-\sqrt{\Delta_1}}{2})$  & $\frac{q-1}{2}\left((v-1)+\frac{2\kappa+\beta(v-1)}{\sqrt{\Delta_1}}\right)$ \\
$\kappa-\frac{1}{q}\left(\kappa+(q-1)\frac{\beta+\sqrt{\Delta_1}}{2}\right)$& $\frac{1}{2}\left((v-1)-\frac{2\kappa+\beta(v-1)}{\sqrt{\Delta_1}}\right)$\\
$\kappa-\frac{1}{q}\left(\kappa+(q-1)\frac{\beta-\sqrt{\Delta_1}}{2}\right)$ & $\frac{1}{2}\left((v-1)+\frac{2\kappa+\beta(v-1)}{\sqrt{\Delta_1}}\right)$\\
\hline
\end{tabular}
\end{center}
\end{table}
\end{theorem}

\begin{IEEEproof}
  Let $$\textbf{c}(b,c)=\left(\langle b,d_1\rangle_{n/e}, \langle b,d_2\rangle_{n/e},\cdots, \langle b,d_n\rangle_{n/e}\right)+c\textbf{1}$$
   be a codeword in $\overline{\mathcal{C}_D}$, where $b \in V_{n}^{(p)}$ and $c\in \gf_q$.
According to Lemmas  \ref{lem-21} and \ref{lem-34}, $\overline{\mathcal{C}_D}$ is an $[\kappa, r+1]$ linear code over $\gf_q$ and
\begin{eqnarray*}
\wt(\textbf{c}(b,c))
=\left\{
\begin{array}{ll}
0 & \mbox{with frequency $1$,}\\
\kappa & \mbox{with frequency $q-1$,}\\
\kappa-\frac{1}{q}(\kappa-\frac{\beta+\sqrt{\Delta_1}}{2}) & \mbox{with frequency $\frac{q-1}{2}\left((v-1)-\frac{2\kappa+\beta(v-1)}{\sqrt{\Delta_1}}\right)$,}\\
\kappa-\frac{1}{q}(\kappa-\frac{\beta-\sqrt{\Delta_1}}{2}) & \mbox{with frequency $\frac{q-1}{2}\left((v-1)+\frac{2\kappa+\beta(v-1)}{\sqrt{\Delta_1}}\right)$,}\\
\kappa-\frac{1}{q}\left(\kappa+(q-1)\frac{\beta+\sqrt{\Delta_1}}{2}\right) & \mbox{with frequency $\frac{1}{2}\left((v-1)-\frac{2\kappa+\beta(v-1)}{\sqrt{\Delta_1}}\right)$,}\\
\kappa-\frac{1}{q}\left(\kappa+(q-1)\frac{\beta-\sqrt{\Delta_1}}{2}\right) & \mbox{with frequency $\frac{1}{2}\left((v-1)+\frac{2\kappa+\beta(v-1)}{\sqrt{\Delta_1}}\right)$,}\\
\end{array}\right.
\end{eqnarray*}
where $|D|=\kappa$. By Lemma \ref{Plessmoments}, we then have $A_1^\perp=A_2^\perp=0$. Hence, we have $d^\perp \geq3$.

Note that the generator matrix of $\overline{\mathcal{C}_D}$ is given as follows:
\begin{eqnarray*}
 \overline{G}=\left[
\begin{array}{cccc}
 \phi(\mathbf{d}_1)&  \phi(\mathbf{d}_2)& \cdots&  \phi(\mathbf{d}_n)\\
 \textbf{1}& \textbf{1}&\cdots& \textbf{1}
 \end{array}
 \right],
 \end{eqnarray*}
 where $\phi(\mathbf{d_i})=(\langle e_1, d_i\rangle_{n/e}, \langle e_2, d_i\rangle_{n/e},\cdots, \langle e_{r}, d_i\rangle_{n/e})^{T} $and $\{e_1,  e_2, \cdots, e_{r}\} $ is a $\gf_{q}$-basis of $V_{n}^{(p)}$. Since $D$ is $\gf_{q}^*$-invariant, then for any $d_i \in D$, we have $a_1d_i$ and $a_2d_i \in D$ for $q>3$, where $1,a_1, a_2 \in \gf_{q}^*$ are pairwise distinct. Let $\alpha, \beta, \gamma \in \gf_{q}$ such that $\alpha\phi(\mathbf{d_i})+\beta\phi(a_1\mathbf{d_i})+\gamma\phi(a_2\mathbf{d_i})=0$ and $\alpha+\beta+\gamma=0$. Then we have 
 \begin{eqnarray*}
\left\{
\begin{array}{ll}
\alpha+\beta+\gamma=0,\\
\alpha+a_1\beta+a_2\gamma=0.
\end{array}\right.
\end{eqnarray*}
 Let $\beta=1$. Then we have $\gamma=-\frac{a_1-1}{a_2-1}$ and $\alpha=-\frac{a_2-a_1}{a_2-1}$. Since $1,a_1, a_2 \in \gf_{q}^*$ are pairwise distinct, we deduce  $\gamma \neq0$ and $\alpha \neq0$. Thus  the minimum distance  of $\overline{\mathcal{C}_D}^{\perp}$ is $3$. When $q>3$ and $r\geq2$, by the sphere-packing bound, it is easy to deduce $d^{\perp}\leq 4$. Hence, $\overline{\mathcal{C}_D}^{\perp}$  is at least almost optimal  according to the sphere packing bound.

The self-orthogonality of $\overline{\mathcal{C}_D}$ follows from Lemma \ref{lem-31} and Theorems \ref{The-31}, \ref{The-32} and \ref{The-33}.
\end{IEEEproof}

\begin{example}
Let $p=2$, $e=1$ and $r=n=6$. Let $\gf_{2^6}=\langle \alpha\rangle$. Take $\ell_1=1$ and $N=p^{\ell_1}+1$, then $C_{i}=\{\alpha^{jN+i}:0 \leq j\leq\frac{q^r-1}{N}-1\}$ for $0\leq i\leq N-1$. Let $D=C_{0}\bigcup C_1$. Then $D$ is a $(64, 42, 26, 30)$-partial difference set by Lemma \ref{lem-23}, where $\chi_1(\alpha^iD)\in \{2,-6\}$.  The $\overline{\mathcal{C}_D}$ in Theorem \ref{theo-31} has parameters $[42, 7, 18]$ and weight enumerator $1+21z^{18}+42z^{20}+42z^{22}+21z^{24}+z^{42}$.
It is verified that $8 \mid \left(|D|-\sum_{x \in D}\chi_{b_1}(x)-\sum_{x \in D}\chi_{b_2}(x)+\sum_{x \in D}\chi_{b_1+b_2}\left(x\right)\right)$ by Lemma \ref{lem-23}. By Theorem \ref{theo-31}, we deduce that $\overline{\mathcal{C}_D}$  is a \textbf{singly-even self-orthogonal} code.
\end{example}

\begin{theorem}\label{theo-32}
Let $p$ be a prime, $q=p^e$, $n$ and $e$ are positive integers such that $r = \frac{n}{e}$. Let $D\subseteq V_{n}^{(p)}$ be an $\gf_{q}^*$-invariant partial difference set with parameters $(p^n, \kappa, \lambda, \mu)$ satisfying  $-D=D$, $\mathbf{e} \in D$, $\langle D\rangle=V_{n}^{(p)}$. Then we have the following.
\begin{enumerate}
\item $\overline{\mathcal{C}_D}$ is a $[\kappa, r+1]$ linear code  with its weight distribution as shown in Table \ref{tab-32}. Besides,  $\overline{\mathcal{C}_D}^{\perp}$  has parameters $[\kappa, \kappa-r-1,\geq3]$. In  particular, if $q>2$, then  $\overline{\mathcal{C}_D}^{\perp}$  has parameters $[\kappa, \kappa-r-1,3]$ and, for  $\lambda \geq 1, \mu \geq 1$, $r\geq2$ and $q>3$, $\overline{\mathcal{C}_D}^{\perp}$  is at least almost optimal  according to the sphere packing bound.
\item If $p\mid \kappa$ and $p>3$, then  $\overline{\mathcal{C}_D}$ is a self-orthogonal code. If $q=3$, $9| \left(|D|-\chi_{b}(D)\right)$ for all $b \in V_{n}^{(p)} \backslash \{0\}$ and $p\mid \kappa$, then $\overline{\mathcal{C}_D}$ is a self-orthogonal code. If $q=2$, $8| \left(|D|-\chi_{b_1}(D)-\chi_{b_2}(D)+\chi_{b_1+b_2}\left(D\right)\right)$ for all $b_1, b_2 \in V_{n}^{(p)}\backslash \{0\}$ and $p\mid \kappa$, then $\overline{\mathcal{C}_D}$ is a self-orthogonal code.
\end{enumerate}

\begin{table}[ht]
\begin{center}
\caption{The weight distribution of $\overline{\cC_{D}}$ in Theorem \ref{theo-32}}\label{tab-32}
\begin{tabular}{cc} \hline
Weight  &  Frequency   \\ \hline
$0$          &  $1$ \\
$\kappa$ & $q-1$ \\
$\kappa-\frac{1}{q}(\kappa-\frac{\beta+\sqrt{\Delta_2}}{2})$  &  $\frac{q-1}{2}\left((v-1)+\frac{2(v-\kappa)-\beta(v-1)}{\sqrt{\Delta_2}}\right)$ \\
$\kappa-\frac{1}{q}(\kappa-\frac{\beta-\sqrt{\Delta_2}}{2})$  & $\frac{q-1}{2}\left((v-1)-\frac{2(v-\kappa)-\beta(v-1)}{\sqrt{\Delta_2}}\right)$ \\
$\kappa-\frac{1}{q}\left(\kappa+(q-1)\frac{\beta+\sqrt{\Delta_2}}{2}\right)$& $\frac{1}{2}\left((v-1)+\frac{2(v-\kappa)-\beta(v-1)}{\sqrt{\Delta_2}}\right)$\\
$\kappa-\frac{1}{q}\left(\kappa+(q-1)\frac{\beta-\sqrt{\Delta_2}}{2}\right)$ & $\frac{1}{2}\left((v-1)-\frac{2(v-\kappa)-\beta(v-1)}{\sqrt{\Delta_2}}\right)$\\
\hline
\end{tabular}
\end{center}
\end{table}
\end{theorem}

\begin{IEEEproof}
 Let $$\textbf{c}(b,c)=\left(\langle b,d_1\rangle_{n/e}, \langle b,d_2\rangle_{n/e},\cdots, \langle b,d_n\rangle_{n/e}\right)+c\textbf{1}$$
   be a codeword in $\overline{\mathcal{C}_D}$, where $b \in V_{n}^{(p)}$ and $c\in \gf_q$.
   According to Lemmas  \ref{lem-22} and \ref{lem-34}, $\overline{\mathcal{C}_D}$ is an $[\kappa, r+1]$ linear code over $\gf_q$ and
\begin{eqnarray*}
\wt(\textbf{c}(b,c))=\left\{
\begin{array}{ll}
0 & \mbox{with frequency $1$,}\\
\kappa  & \mbox{with frequency $q-1$,}\\
\kappa-\frac{1}{q}(\kappa-\frac{\beta+\sqrt{\Delta_2}}{2}) & \mbox{with frequency $\frac{q-1}{2}\left((v-1)+\frac{2(v-\kappa)-\beta(v-1)}{\sqrt{\Delta_2}}\right)$,}\\
\kappa-\frac{1}{q}(\kappa-\frac{\beta-\sqrt{\Delta_2}}{2}) & \mbox{with frequency $\frac{q-1}{2}\left((v-1)-\frac{2(v-\kappa)-\beta(v-1)}{\sqrt{\Delta_2}}\right)$,}\\
\kappa-\frac{1}{q}\left(\kappa+(q-1)\frac{\beta+\sqrt{\Delta_2}}{2}\right) & \mbox{with frequency $\frac{1}{2}\left((v-1)+\frac{2(v-\kappa)-\beta(v-1)}{\sqrt{\Delta_2}}\right)$,}\\
\kappa-\frac{1}{q}\left(\kappa+(q-1)\frac{\beta-\sqrt{\Delta_2}}{2}\right) & \mbox{with frequency $\frac{1}{2}\left((v-1)-\frac{2(v-\kappa)-\beta(v-1)}{\sqrt{\Delta_2}}\right)$.}\\
\end{array}\right.
 \end{eqnarray*}
 According to Lemma \ref{Plessmoments}, we deduce $A_1^\perp=0$ and $A_2^\perp=0$. Then we have $d^\perp \geq3$.
 The rest of the proof is analogous to that of Theorem  \ref{theo-31}.
\end{IEEEproof}

\begin{example}
Let $p=2$, $e=1$, and $r=n=6$. Let $\gf_{2^6}=\langle \alpha\rangle$. Take $\ell_1=3$ and $N=p^{\ell_1}+1$, then $C_{i}=\{\alpha^{jN+i}:0 \leq j\leq\frac{q^r-1}{N}-1\}$ for $0\leq i\leq N-1$. Let $D=\bigcup_{i=0}^{6}C_{i}\bigcup\{0\}$. Then $D$ is a $(64, 50, 38, 42)$ partial difference set by Lemma \ref{lem-23}, where $\chi_1(\alpha^i D)\in \{2,-6\}$. Then $\overline{\mathcal{C}_D}$ in Theorem \ref{theo-32} has parameters $[50, 7, 22]$ and weight enumerator $1+14z^{22}+49z^{24}+49z^{26}+14z^{28}+z^{50}$.
It is easy to deduce $8 \mid \left(|D|-\chi_{b_1}(D)-\chi_{b_2}(D)+\chi_{b_1+b_2}\left(D\right)\right)$. According to Theorem \ref{theo-32}, we know that $\overline{\mathcal{C}_D}$ is a \textbf{singly-even self-orthogonal} code. 
\end{example}

\begin{remark}We make the following remarks on the codes in Theorems \ref{theo-31} and \ref{theo-32}.
\begin{enumerate}
\item In Table \ref{tab:best_codes}, we list the parameters of some codes produced by Theorems \ref{theo-31} and \ref{theo-32}. They are optimal or have best known parameters according to the Code Tables at http://www.codetables.de/.
    
\begin{table}[ht] 
\begin{center}
\caption{Good codes produced by Theorems \ref{theo-31} and \ref{theo-32}}\label{tab:best_codes}
\begin{tabular}{lllllll}
\hline
$q$ & $r$ & $\mathbf{e} \in D$ or $\mathbf{e} \notin D$ & Code  & Parameters  & Optimality\\
\midrule
$2$ & $4$ & $\mathbf{e} \notin D$ &  $\overline{\mathcal{C}_D}$ & $[10, 5, 4]$ & Optimal \\
$2$ & $4$ & $\mathbf{e} \in D$ &  $\overline{\mathcal{C}_D}$& $[13, 5, 5]$ & Optimal \\
$2$ & $6$ & $\mathbf{e} \notin D$ &  $\overline{\mathcal{C}_D}^\perp$ & $[42, 35, 4]$ & Optimal \\
$2$ & $6$ & $\mathbf{e} \notin D$ &  $\overline{\mathcal{C}_D}$ & $[36, 7, 16]$ & Optimal \\
$2$ & $8$ & $\mathbf{e} \notin D$ &  $\overline{\mathcal{C}_D}$ & $[120, 9, 56]$ & Optimal \\
$2$ & $8$ & $\mathbf{e} \notin D$ &  $\overline{\mathcal{C}_D}$ & $[153, 9, 72]$ & Best known parameters\\
$2$ & $8$ & $\mathbf{e} \notin D$ &  $\overline{\mathcal{C}_D}$ & $[170, 9, 80]$ & Best known parameters \\
$3$ & $4$ & $\mathbf{e} \notin D$ &  $\overline{\mathcal{C}_D}^\perp$ & $[24, 19, 3]$ & Optimal \\
$3$ & $4$ & $\mathbf{e} \notin D$ &  $\overline{\mathcal{C}_D}$ & $[48, 5, 30]$ & Optimal \\
$3$ & $4$ & $\mathbf{e} \notin D$ &  $\overline{\mathcal{C}_D}$ & $[40, 5, 24]$ & Optimal \\
$3$ & $4$ & $\mathbf{e} \notin D$ &  $\overline{\mathcal{C}_D}^\perp$ & $[40, 35, 3]$ & Optimal \\
$3$ & $5$ & $\mathbf{e} \notin D$ &  $\overline{\mathcal{C}_D}$ & $[220, 6, 144]$ & Optimal \\
$3$ & $5$ & $\mathbf{e} \notin D$ &  $\overline{\mathcal{C}_D}^\perp$ & $[220, 214, 3]$ & Optimal \\
$4$ & $4$ & $\mathbf{e} \in D$ &  $\overline{\mathcal{C}_D}$ & $[52, 5, 36]$ & Optimal  \\
$4$ & $4$ & $\mathbf{e} \in D$ & $\overline{\mathcal{C}_D}^\perp$ & $[52, 47, 3]$ & Optimal \\
$4$ & $4$ & $\mathbf{e} \notin D$ & $\overline{\mathcal{C}_D}^\perp$ & $[102, 97, 3]$ & Optimal \\
$4$ & $4$ & $\mathbf{e} \notin D$ & $\overline{\mathcal{C}_D}^\perp$ & $[204, 199, 3]$ & Optimal \\
$5$ & $2$ & $\mathbf{e} \notin D$ &  $\overline{\mathcal{C}_D}$ & $[12, 3, 8]$ & Optimal  \\
$5$ & $2$ & $\mathbf{e} \notin D$ &  $\overline{\mathcal{C}_D}$ & $[16, 3, 12]$ & Optimal  \\
$5$ & $2$ & $\mathbf{e} \notin D$ &  $\overline{\mathcal{C}_D}$& $[20, 3, 15]$ & Optimal  \\
$5$ & $4$ & $\mathbf{e} \in D$ &  $\overline{\mathcal{C}_D}$ & $[105, 5, 80]$ & Best known parameters \\
$7$ & $2$ & $\mathbf{e} \notin D$ &  $\overline{\mathcal{C}_D}$ & $[36, 3, 30]$ & Optimal  \\
$7$ & $2$ & $\mathbf{e} \notin D$ &  $\overline{\mathcal{C}_D}$ & $[42, 3, 35]$ & Optimal  \\
$8$ & $2$ & $\mathbf{e} \in D$ &  $\overline{\mathcal{C}_D}$ & $[50, 3, 42]$ & Optimal \\
$8$ & $2$ & $\mathbf{e} \in D$ &  $\overline{\mathcal{C}_D}$ & $[57, 3, 49]$ & Optimal \\
$9$ & $2$ & $\mathbf{e} \notin D$ &  $\overline{\mathcal{C}_D}$ & $[64, 3, 56]$ & Optimal \\
$9$ & $2$ & $\mathbf{e} \notin D$ &  $\overline{\mathcal{C}_D}$ & $[73, 3, 63]$ & Optimal \\
$9$ & $3$ & $\mathbf{e} \notin D$ &  $\overline{\mathcal{C}_D}^\perp$ & $[104, 100, 3]$ & Optimal \\
\hline
\end{tabular}
\end{center}
\end{table}
\item Let $r\geq2$ and $\overline{\mathcal{C}_D}$ be a $q$-ary self-orthogonal code obtained from Theorem \ref{theo-31} or Theorem \ref{theo-32}. Let $\mathcal{C}_1=\overline{\mathcal{C}_D}^{\perp}$ and $\mathcal{C}_2$ be the dual of the code $\{c\mathbf{1}:c \in \gf_{q}\}\subseteq \overline{\mathcal{C}}$. Then it is easy to deduce that $\mathcal{C}_2$ has parameters $[\kappa, \kappa-1, 2]$ and $\mathcal{C}_1^{\perp}\subset\mathcal{C}_1\subset\mathcal{C}_2$.
According to Lemma \ref{lem-Steane}, then there exist pure quantum codes with parameters $[[\kappa, \kappa-r-2, 3]]_q$. Specially, when $\lambda, \mu\geq1 $ and $q>2$, these quantum codes are  at least almost optimal  according to the quantum Hamming bound for ﬁxed lengths and dimensions. In Table \ref{compare}, we  compare  our  pure  quantum codes with the known ones in \cite{Edel}. It is shown that our pure quantum codes have better parameters as they have higher rate than that of known ones in \cite{Edel}.

\begin{table*}[!t]
\begin{center}
\caption{Comparing our pure quantum codes with those in \cite{Edel}}\label{compare}
\begin{tabular}{cccc}\hline
  Condition & Our quantum codes  &   Quantum codes in \cite{Edel}    \\ \hline
%$q=3, r=6$&  $[[156, 148, 3]]_3$ &  $[[160, 148, 3]]_3$ \\  
$q=4, r=4$&  $[[102, 96, 3]]_4$ &  $[[105, 93, 3]]_4$ \\
$q=4, r=4$&  $[[154, 148, 3]]_4$ &  $[[156, 144, 3]]_4$ \\
$q=4, r=5$&  $[[186, 179, 3]]_4$ &  $[[190, 179, 3]]_4$ \\
%$q=5, r=2$&  $[[20, 16, 3]]_5$ &  $[[23, 16, 3]]_5$ \\
%$q=5, r=4$&  $[[520, 514, 3]]_5$ &  $[[521, 511, 3]]_5$ \\
$q=7, r=4$&  $[[385, 379, 3]]_7$ &  $[[387, 379, 3]]_7$ \\
$q=7, r=4$&  $[[672, 666, 3]]_7$ &  $[[680, 660, 3]]_7$ \\
$q=7, r=4$&  $[[721, 715, 3]]_7$ &  $[[721, 707, 3]]_7$ \\
$q=8, r=2$&  $[[28, 24, 3]]_8$ &  $[[33, 23, 3]]_8$ \\
$q=8, r=4$&  $[[126, 120, 3]]_8$ &  $[[128, 119, 3]]_8$ \\
$q=8, r=4$&  $[[630, 624, 3]]_8$ &  $[[645, 623, 3]]_8$ \\
$q=9, r=2$&  $[[24, 20, 3]]_9$ &  $[[26, 20, 3]]_9$ \\
$q=9, r=2$&  $[[72, 68, 3]]_9$ &  $[[73, 67, 3]]_9$ \\
$q=9, r=3$&  $[[624, 619, 3]]_9$ &  $[[637, 616, 3]]_9$ \\

\hline
\end{tabular}
\end{center}
\end{table*}

\end{enumerate}
\end{remark}

\section{Self-orthogonal codes from projective linear codes}

In this section, we present a new construction of self-orthogonal codes by the action of a multiplicative subgroup $G \subset \mathbb{F}_q^*$ $(|G|>2)$ on the columns of an arbitrary projective linear code. This construction yields self-orthogonal codes with flexible parameters. Then we derive quantum codes from their augmented codes. In particular, selecting appropriate projective codes leads to optimal quantum codes with a high degree of parametric flexibility.

\begin{theorem} \label{theorem-2}
Let $q=p^e$ be a prime power and $G=\{g_1,g_2,\ldots,g_\mathbf{s}\}$ be a multiplicative subgroup of $\gf_{q}^*$ with $|G|=\mathbf{s}>2$ and $\mathbf{s}\mid(q-1)$. Let $\mathcal{C}$ be an $[\mathfrak{n},k,d]$ projective linear code with generator matrix $\mathcal{G}=[\mathbf{v}_1,\mathbf{v}_2,\ldots, \mathbf{v}_\mathfrak{n}]$ over $\gf_{q}$. Let $$\mathcal{\widetilde{G}} = [g_1\mathbf{v}_1, \dots, g_1\mathbf{v}_\mathfrak{n}, \dots, g_\mathbf{s}\mathbf{v}_1, \dots, g_\mathbf{s}\mathbf{v}_\mathfrak{n}]$$ be the matrix obtained by the action of $G$ on the columns of $\mathcal{G}$. Then the linear code $\widetilde{\mathcal{C}}$, generated by the matrix $\widetilde{\mathcal{G}}$, is an $[\mathbf{s}\mathfrak{n}, k, \mathbf{s}d]$ self-orthogonal code with $\mathbf{1} \notin \mathcal{\widetilde{C}}$.
\end{theorem}

\begin{IEEEproof}
It is clear that the length of $\mathcal{\widetilde{G}}$ is $\mathbf{s}n$.
Let $\mathbf{r}_1, \mathbf{r}_2, \ldots, \mathbf{r}_k$ be the rows of $\mathcal{G}$ and $\mathbf{\widetilde{r}}_1, \mathbf{\widetilde{r}}_2, \ldots, \mathbf{\widetilde{r}}_k$ be the rows of $\mathcal{\widetilde{G}}$, where each 
$$\mathbf{\widetilde{r}}_i=(g_1\mathbf{r}_i, g_2\mathbf{r}_i, \ldots, g_\mathbf{s}\mathbf{r}_i).$$ 
Suppose that there exist $\alpha_1,\dots,\alpha_k \in \mathbb{F}_q$ such that
\[
\sum_{i=1}^{k}\alpha_{i}\widetilde{\mathbf{r}}_i = \mathbf{0}.
\]
Then for all $g \in G$, we have
\[
g\left(\sum_{i=1}^{k}\alpha_{i}\mathbf{r}_i\right) = \mathbf{0}.
\]
 Since $g \neq 0$, we have $\sum_{i=1}^{k}\alpha_{i}\mathbf{r}_i = 0$. Because $\mathbf{r}_1, \mathbf{r}_2, \dots, \mathbf{r}_k$ are linearly independent, it follows that $\alpha_1 = \cdots = \alpha_k = 0$. Hence $\mathbf{\widetilde{r}}_1, \mathbf{\widetilde{r}}_2, \dots, \mathbf{\widetilde{r}}_k$ are also linearly independent. Therefore, the dimension of $\mathcal{\widetilde{G}}$ is $k$. By the definition of $\mathcal{\widetilde{C}}$, its minimum distance equals $\mathbf{s}$ times that of $\mathcal{C}$.

Let $\mathbf{\widetilde{c}}_1$ and $\mathbf{\widetilde{c}}_2$ be two codewords in  $\mathcal{\widetilde{C}}$. Then we have 
$$\mathbf{\widetilde{c}}_1=(g_1\textbf{k}_1\mathbf{v}_1, \dots, g_1\textbf{k}_1\mathbf{v}_\mathfrak{n}, \dots, g_\mathbf{s}\textbf{k}_1\mathbf{v}_1, \dots, g_\mathbf{s}\textbf{k}_1\mathbf{v}_\mathfrak{n})
$$
with $\textbf{k}_1$ a $1\times k$ vector in $\gf_q^k$, and
$$\mathbf{\widetilde{c}}_2=(g_1\textbf{k}_2\mathbf{v}_1, \dots, g_1\textbf{k}_2\mathbf{v}_\mathfrak{n}, \dots, g_\mathbf{s}\textbf{k}_2\mathbf{v}_1, \dots, g_\mathbf{s}\textbf{k}_2\mathbf{v}_\mathfrak{n})
$$
with $\textbf{k}_2$ a $1\times k$ vector in $\gf_q^k$. It follows that
\begin{eqnarray*}
\nonumber &
&\mathbf{\widetilde{c}}_1\cdot\mathbf{\widetilde{c}}_2\\
&=&\sum_{g\in G}\sum_{i=1}^{\mathfrak{n}}(g\textbf{k}_1\mathbf{v}_i)(g\textbf{k}_2\mathbf{v}_i)\\
&=&\sum_{g\in G}g^2\sum_{i=1}^{\mathfrak{n}}(\textbf{k}_1\mathbf{v}_i)(\textbf{k}_2\mathbf{v}_i)\\
&=&(\mathbf{c}_1\cdot\mathbf{c}_2)\sum_{g\in G}g^2,
\end{eqnarray*}
where $\mathbf{c}_1:=(\textbf{k}_1\mathbf{v}_1,\textbf{k}_1\mathbf{v}_2,\ldots,\textbf{k}_1\mathbf{v}_\mathfrak{n})$ and $\mathbf{c}_2:=(\textbf{k}_2\mathbf{v}_1,\textbf{k}_2\mathbf{v}_2,\ldots,\textbf{k}_2\mathbf{v}_\mathfrak{n})$ are two codewords in $\mathcal{C}$.
Note that $$\sum_{g \in G}g^2=1+\beta^2+\cdots+\beta^{2(s-1)}=\frac{1-(\beta^2)^{\mathbf{s}}}{1-\beta^2}=0,$$ where $G=\langle\beta\rangle$ and $|G|>2$.
Hence, we deduce $\mathbf{\widetilde{c}}_1\cdot\mathbf{\widetilde{c}}_2=0$.
Then  $\mathcal{\widetilde{G}}$ is a self-orthogonal code.

Assume that $\mathbf{1} \in \mathcal{\widetilde{C}}$. Then there exists $\mathbf{c}=(c_1, \ldots, c_\mathfrak{n}) \in \mathcal{C}$ such that $(g_1c_1,\ldots,g_\mathbf{s}c_1,\ldots,g_1c_\mathfrak{n}, \ldots, g_\mathbf{s}c_\mathfrak{n})$ is the all-1 vector. This implies $g_1=g_2=\cdots=g_\mathbf{s}$, which is impossible. 
Therefore, we have $\mathbf{1} \notin \mathcal{\widetilde{C}}.$
\end{IEEEproof}

In then following, we investigate the self-orthogonality of the augmented code of $\mathcal{\widetilde{C}}$.

\begin{theorem}\label{coroll-32}
 Let $\mathcal{\widetilde{C}}$ be the linear code defined in Theorem \ref{theorem-2}. 
Then the augmented code $\overline{\mathcal{\widetilde{C}}}$ is an $[\mathbf{s}\mathfrak{n}, k+1]$ self-orthogonal code if and only if $p \mid \mathfrak{n}$.
Its dual $\overline{\mathcal{\widetilde{C}}}^{\perp}$ has parameters $[\mathbf{s}\mathfrak{n}, \mathbf{s}\mathfrak{n}-k-1,\geq 3]$. 
\end{theorem}

\begin{IEEEproof}
For any $\widetilde{\mathbf{c}}\in\widetilde{\mathcal{C}}$, there exists $\mathbf{c}=(c_1,\dots,c_\mathfrak{n})\in\mathcal{C}$ such that $\widetilde{\mathbf{c}}=(g_1c_1,\dots,g_\mathbf{s}c_1,\dots,g_1c_\mathfrak{n},\dots,g_\mathbf{s}c_\mathfrak{n})$. Then the sum of all entries of $\widetilde{\mathbf{c}}$ equals
\[
 \sum_{g\in G}\sum_{i=1}^{\mathfrak{n}} g c_i = \left(\sum_{g\in G} g\right)\left(\sum_{i=1}^{\mathfrak{n}} c_i\right).
\]
Since $G=\langle\beta\rangle$ with $|G|>2$, we have $\sum_{g\in G} g=0$ and $\sum_{g\in G}\sum_{i=1}^{\mathfrak{n}} g c_i =0$.
Thus $\widetilde{\mathcal{C}}$ is even-like. Note that $\gcd(p,\mathbf{s})=1$ as $\mathbf{s} \mid (q-1)$.
The self-orthogonality of $\overline{\mathcal{\widetilde{C}}}$ follows from Theorem \ref{theorem-2} and Lemma \ref{lem-31}. 
 According to Theorem \ref{theorem-2}, we know  $\mathbf{1} \notin \mathcal{\widetilde{C}}$. Then the dimension of $\overline{\mathcal{\widetilde{C}}}$ is $k+1$.

The generator matrix of $\overline{\mathcal{\widetilde{C}}}$ is given by
\begin{eqnarray*}
 \overline{\widetilde{G}}=\left[
\begin{array}{ccccccc}
 g_1\mathbf{v}_1 & \cdots & g_1\mathbf{v}_\mathfrak{n} & \cdots & g_\mathbf{s}\mathbf{v}_1 & \cdots & g_\mathbf{s}\mathbf{v}_\mathfrak{n}\\
 1& \cdots & 1& \cdots & 1 & \cdots & 1
 \end{array}
 \right].
 \end{eqnarray*}
Since $\mathcal{C}$ is a projective linear code, the vectors $\mathbf{v}_1, \mathbf{v}_2, \ldots, \mathbf{v}_\mathfrak{n}$ are pairwise linearly independent over $\gf_q$. 
Moreover, $g_1, g_2, \ldots, g_{\mathbf{s}}$ are pairwise distinct. 
Consequently, the columns of $\overline{\widetilde{G}}$ are pairwise linearly independent over $\gf_q$. 
It follows that the minimum distance $d^{\perp}$ of $\overline{\mathcal{\widetilde{C}}}^{\perp}$ satisfies $d^{\perp} \geq 3$.
\end{IEEEproof}

Now we use  $\overline{\mathcal{\widetilde{C}}}^{\perp}$ to construct quantum codes. 

\begin{corollary}\label{coroll-33}
With the same notation as in Theorem  \ref{theorem-2}. Let $p \mid \mathfrak{n}$, $\mathbf{s}\mid (q-1)$ and $k\geq2$. Then there exist  pure quantum codes with parameters $[[\mathbf{s}\mathfrak{n}, \mathbf{s}\mathfrak{n}-k-2, 3]]_q$.
\end{corollary}

\begin{IEEEproof}
Let $\mathcal{C}_1=\overline{\mathcal{\widetilde{C}}}^{\perp}$ and $\mathcal{C}_2$ be the dual of the code $\{c\mathbf{1}:c \in \gf_{q}\}\subseteq \overline{\mathcal{\widetilde{C}}}$. Then it is easy to deduce that $\mathcal{C}_2$ has parameters $[\mathbf{s}\mathfrak{n}, \mathbf{s}\mathfrak{n}-1, 2]$ and $\mathcal{C}_1^{\perp}\subset\mathcal{C}_1\subset\mathcal{C}_2$ by Theorem \ref{coroll-32}. 
The desired conclusion follows from Lemma \ref{lem-Steane}. 
\end{IEEEproof}

In Corollary \ref{coroll-33}, the quantum codes may be optimal if we select suitable projective codes. 
In the following, we select two families of projective codes to derive optimal quantum codes.

\begin{corollary}\label{coroll-33-1}
Let $q=p^e$ be a prime power, $\mathbf{s}\mid (q-1)$, $\mathbf{s}>2$ and $1\leq u\leq r-1$ with $r$ a positive integer. Then there exist dimension optimal pure quantum codes with parameters $\left[\left[\mathbf{s}\frac{q^r-q^{u}}{q-1}, \mathbf{s}\frac{q^r-q^{u}}{q-1}-r-2, 3\right]\right]_q$ with  respect  to  the  quantum Hamming bound. 
\end{corollary}

\begin{IEEEproof}
It is known that the MacDonald codes are projective linear codes obtained by
 puncturing the simplex codes and have parameters $[\frac{q^r-q^{u}}{q-1}, r, q^{r-1}-q^{u-1}]$, where $1\leq u\leq r-1$.
 By Theorem \ref{theorem-2}, a family of
self-orthogonal codes with parameters $[\mathbf{s}\frac{q^r-q^{u}}{q-1}, r, \mathbf{s}(q^{r-1}-q^{u-1})]$ can be obtained from MacDonald codes with  $\mathbf{s}\mid (q-1)$ and $\mathbf{s}>2$. Its augmented code is also a self-orthogonal code whose dual code  has parameters $[\mathbf{s}\frac{q^r-q^{u}}{q-1},\mathbf{s}\frac{q^r-q^{u}}{q-1}-r-1,\geq 3]$, according to Theorem \ref{coroll-32}. The desired conclusion can be obtained from Lemma \ref{lem-Steane} by an argument similar to that used in the proof of Corollary \ref{coroll-33}.
It is easy to verify these quantum codes are dimension optimal with  respect  to  the  quantum Hamming bound.
\end{IEEEproof}

\begin{lemma}\label{lemma-31} \cite[Theorem 6 and Proposition 5]{YWei}
Let $q=p$ be a prime and $n, s$ be integers with $0\leq s\leq n-4$ for even $n+s$, and $0\leq s\leq n-3$, $(p,n)\neq (3,3)$ for odd $n+s$. Let $f(x):V_{n}^{(p)}\rightarrow \gf_{p}$ be an $s$-plateaued function belonging to $\mathscr{F}$ and $\sharp B_{+}(f)=k$. When $n+s$ is even, $\widehat{\mathcal{C}}_{D_{f,sq}}$ and $\widehat{\mathcal{C}}_{D_{f,nsq}}$ are $p$-ary $\left[\frac{p^{n-1}-\epsilon_{0}p^{\frac{n+s}{2}-1}}{2}, n\right]$ projective linear codes. When $n+s$ is odd, $\widehat{\mathcal{C}}_{D_{f,sq}}$ is a $p$-ary $\left[\frac{p^{n-1}+\epsilon_{0}p^{\frac{n+s-1}{2}}}{2}, n\right]$ projective linear codes, and $\widehat{\mathcal{C}}_{D_{f,nsq}}$ is a $p$-ary $\left[\frac{p^{n-1}-\epsilon_{0}p^{\frac{n+s-1}{2}}}{2}, n\right]$ projective linear codes.
\end{lemma}

Now we will use the projective code in Lemma \ref{lemma-31} to construct optimal pure quantum codes. According to Theorems \ref{theorem-2} and \ref{coroll-32}, together with Lemma \ref{lem-Steane}, the following corollary follows from a proof analogous to that of Corollary \ref{coroll-33}.

\begin{corollary}
Let $q=p$ be a prime, $\mathbf{s}\mid (p-1)$ and $\mathbf{s}>2$. Let $n, s$ be integers with $0\leq s\leq n-4$ for even $n+s$, and $0\leq s\leq n-3$, $(p,n)\neq (3,3)$  for odd $n+s$. 
Then we have the following. 
\begin{enumerate}
\item When $n+s$ is even, there exists a family of dimension optimal pure quantum codes with parameters $$\left[\left[\mathbf{s}\frac{p^{n-1}-\epsilon_{0}p^{\frac{n+s}{2}-1}}{2}, \mathbf{s}\frac{p^{n-1}-\epsilon_{0}p^{\frac{n+s}{2}-1}}{2}-n-2,3\right]\right]_p$$ with respect to the quantum Hamming bound.
    
 \item When $n+s$ is odd, there exist two families of dimension optimal pure quantum codes with parameters $$\left[\left[\mathbf{s}\frac{p^{n-1}+\epsilon_{0}p^{\frac{n+s-1}{2}}}{2}, \mathbf{s}\frac{p^{n-1}+\epsilon_{0}p^{\frac{n+s-1}{2}}}{2}-n-2,3\right]\right]_p$$  and  $$\left[\left[\mathbf{s}\frac{p^{n-1}-\epsilon_{0}p^{\frac{n+s-1}{2}}}{2}, \mathbf{s}\frac{p^{n-1}-\epsilon_{0}p^{\frac{n+s-1}{2}}}{2}-n-2,3\right]\right]_p$$ with respect to the quantum Hamming bound.
\end{enumerate}
\end{corollary}

\section{Self-Orthogonal Minimal Codes  via the characteristic function method}
As observed in \cite{Jin}, constructing self-orthogonal minimal codes that violate the Ashikhmin-Barg condition is an interesting problem. In this section, we present new families of such codes. To this end, we use the characteristic function method to construct linear codes, establish criteria for their self-orthogonality, and then construct several classes of self-orthogonal minimal codes that violate the Ashikhmin-Barg condition.

\subsection{Self-orthogonality criteria for linear codes via the characteristic function method}
Let $q=p^e$ for a prime $p$ and a positive integer $e\mid n$. Take a proper subset $D$ of  $V_{n}^{(p)}$. Here we view $V_n^{(p)}$ as a vector space over $\mathbb{F}_{q}$ of dimension $\frac{n}{e}$. Define a characteristic function of $D$ as 
 \begin{eqnarray*}
 f_{D}(x)=\left\{
\begin{array}{lll}
1  &   \mbox{if $x \in D$},\\
0   &   \mbox{if $x \in V_{n}^{(p)}\setminus D$}.
\end{array} \right.
 \end{eqnarray*}
 Define a $q$-ary linear code by
 \begin{eqnarray}\label{eq-3}
 \mathcal{C}_{f_{D}}=\left\{\left(a f_{D}(x)+\langle b,x\rangle_{n/e}\right)_{x\in V_{n}^{(p)}}: (a,b)\in \gf_{q}\times V_{n}^{(p)}\right\}. 
 \end{eqnarray}

\begin{lemma}\label{lemm-41}
Let \( p \) be a prime, \( q = p^e \), and let \( e, n \) be positive integers such that \( e \mid n \) and \( (q, n) \neq (2,1), (2,2), (3,1) \).  
Set \( D = V_n^{(p)} \). Then the code \( \mathcal{C}_D \) defined in Equation~\eqref{eq-1} is self-orthogonal.
\end{lemma}

\begin{IEEEproof}
Since $D=V_{n}^{(p)}$, then $D$ is $\gf_{q}^*$-invariant. According to Theorem \ref{The-31}, we deduce that $\mathcal{C}_D$ is self-orthogonal code for $q>3$.
In the following, we prove that $\mathcal{C}_D$ is a self-orthogonal code for $q=2,3$.
By the orthogonal relation of additive characters, we have
\begin{eqnarray*}
\sum_{x\in V_{n}^{(p)}}\chi_b(x)=\left\{
\begin{array}{ll}
p^n  &   \mbox{if $b=0$},\\
0    &   \mbox{if $b\in V_{n}^{(p)}\setminus \{0\}$}.
\end{array} \right.
\end{eqnarray*}
When $q=3$ and $n \neq 1$, we have $9 \mid |D|$ and $9 \mid \sum_{x\in V_{n}^{(p)}}\chi_b(x)$.
When $q=2$ and $n>2 $, we have $8 \mid |D|$ and $8 \mid \sum_{x\in V_{n}^{(p)}}\chi_b(x)$.
The proof is completed by Theorems~\ref{The-32} and~\ref{The-33}.
\end{IEEEproof}

In the following, we establish a criterion for $\mathcal{C}_{f_{D}}$ to be self-orthogonal. 

\begin{theorem} \label{The-41}
Let $q=p^e$, where $p$ is a prime and $e$ is a positive integer such that $e\mid n$ and  $(q,n)\neq(2,1),(2,2),(3,1)$. Let $\mathcal{C}_D$ be defined in Equation (\ref{eq-1}) and $ \mathcal{C}_{f_{D}}$ be defined in Equation (\ref{eq-3}). Then $ \mathcal{C}_{f_{D}}$ is self-orthogonal if and only if $\mathcal{C}_D$ is an even-like code and $p \mid |D|$.
\end{theorem}

\begin{IEEEproof}
Assume that $\mathcal{C}_D$ is an even-like code and $p \mid |D|$. 
Let $\mathbf{c}_{a_1,b_1}$ and $\mathbf{c}_{a_2,b_2}$ be any two codewords in $\mathcal{C}_{f_{D}}$ given by
$$\mathbf{c}_{a_1,b_1}=\left(a_1 f_{D}(x)+\langle b_1,x\rangle_{n/e}\right)_{x\in V_{n}^{(p)}}$$
and 
$$\mathbf{c}_{a_2,b_2}=\left(a_2 f_{D}(x)+\langle b_2,x\rangle_{n/e}\right)_{x\in V_{n}^{(p)}}$$
with $a_1,a_2\in \gf_q$ and $b_1,b_2\in  V_{n}^{(p)}$. 
Then we have
\begin{eqnarray*}
&&\mathbf{c}_{a_1,b_1}\cdot\mathbf{c}_{a_2,b_2}\\
&=&\sum_{x \in V_{n}^{(p)}}\left(a_1 f_{D}(x)+\langle b_1,x\rangle_{n/e}\right)\left(a_2 f_{D}(x)+\langle b_2,x\rangle_{n/e}\right)\\
&=&\sum_{x \in V_{n}^{(p)}}a_1a_2f_{D}(x)^2+\sum_{x \in V_{n}^{(p)}}a_1f_{D}(x)\langle b_2,x\rangle_{n/e}+\sum_{x \in V_{n}^{(p)}}a_2f_{D}(x)\langle b_1,x\rangle_{n/e}+\sum_{x \in V_{n}^{(p)}}\langle b_1,x\rangle_{n/e}\langle b_2,x\rangle_{n/e}\\
&=&\sum_{x \in D}a_1a_2+\sum_{x \in D}a_1\langle b_2,x\rangle_{n/e}+\sum_{x \in D}a_2\langle b_1,x\rangle_{n/e}+\sum_{x \in V_{n}^{(p)}}\langle b_1,x\rangle_{n/e}\langle b_2,x\rangle_{n/e}.
\end{eqnarray*}
Note that $\left(\langle b_1,x\rangle_{n/e}\right)_{x\in V_{n}^{(p)}}$ and $\left(\langle b_2,x\rangle_{n/e}\right)_{x\in V_{n}^{(p)}}$ are two codewords in $\mathcal{C}_D$ for $D=V_{n}^{(p)}$. By Lemma \ref{lemm-41}, we have $\sum_{x \in V_{n}^{(p)}}\langle b_1,x\rangle_{n/e}\langle b_2,x\rangle_{n/e}=0$ for $(q,n)\neq(2,1),(2,2),(3,1)$.
Since  $\mathcal{C}_D$ is  even-like code and $p \mid |D|$, we have $\sum_{x \in D}a_1a_2=0$, $\sum_{x \in D}a_1\langle b_2,x\rangle_{n/e}=0$ and $\sum_{x \in D}a_2\langle b_1,x\rangle_{n/e}=0$. Thus $\mathbf{c}_{a_1,b_1}\cdot\mathbf{c}_{a_2,b_2}=0$, which implies that $ \mathcal{C}_{f_{D}}$ is self-orthogonal code.

 Conversely, we assume that $\mathcal{C}_{f_D}$ is self-orthogonal. 
 When $a_1\neq 0, a_2 \neq 0$ and $b_1= b_2 =0$, we have $$\mathbf{c}_{a_1,b_1}\cdot\mathbf{c}_{a_2,b_2}=\sum_{x \in D}a_1a_2=a_1a_2|D|=0,$$ which implies $p \mid |D|$. 
When $a_1 \neq 0, a_2=0$ and $b_1=0, b_2 \neq0$, we have $$\mathbf{c}_{a_1,b_1}\cdot\mathbf{c}_{a_2,b_2}=\sum_{x \in D}a_1\langle b_2,x\rangle_{n/e}=0,$$ which yields $\sum_{x \in D}\langle b_2,x\rangle_{n/e}=0$ for $b_2 \in V_{n}^{(p)}$. Hence $\mathcal{C}_D$ is an even-like code. This completes the proof. 
\end{IEEEproof}

\begin{corollary}\label{cor-41}
Let $q=p^e>3$. If $D$ is $\gf_{q}^*$-invariant and $p \mid |D|$, then $ \mathcal{C}_{f_{D}}$ is a self-orthogonal code.
\end{corollary}

\begin{IEEEproof}
By Corollary \ref{corol-31}, since $D$ is $\gf_{q}^*$-invariant, we deduce that $\mathcal{C}_{D}$ is even-like. 
Then the desired conclusion follows from Theorem \ref{The-41}. 
\end{IEEEproof}

\subsection{New families of self-orthogonal minimal codes that violate the Ashikhmin-Barg condition}
To construct self-orthogonal minimal codes, we require several lemmas. We begin by computing a class of hybrid character sums in the following lemma.

\begin{lemma}\label{lem-41}
Let $p$ be an odd prime. Let $S=\sum_{x \in V_{n}^{(p)}}\eta_1\left(f(x)\right)\chi_b\left( x\right),$
where $b \in V_{n}^{(p)}\backslash \{0\}$ and $f(x): V_{n}^{(p)}\rightarrow \gf_{p}$ be an $s$-plateaued function belonging to $\mathscr{F}$ with Walsh support $S_f$. 
Let $\epsilon_b$ be defined in Subsection \ref{subsectionF}.
Then we have 
the following results. 
\begin{enumerate}
\item If $n+s$ is even, then
 \begin{eqnarray*}
 S=\left\{
\begin{array}{lll}
0  &   \mbox{if $b \notin S_{f}$ or $b \in S_{f}, f^*(b)=0$},\\
\epsilon_{b} p^{\frac{{n+s}}{2}}   &   \mbox{if $b \in S_{f}, f^*(b) \in SQ$},\\
-\epsilon_{b} p^{\frac{{n+s}}{2}}    &  \mbox{if $b \in S_{f}, f^*(b)\in NSQ$}.\\
\end{array} \right.
 \end{eqnarray*}
 \item If $n+s$ is odd, then
 \begin{eqnarray*}
 S=\left\{
\begin{array}{lll}
0  &   \mbox{if $b \notin S_{f}$},\\
\epsilon_{b}(p-1) p^{\frac{n+s-1}{2}}   &   \mbox{if $b \in S_{f}, f^*(b)=0$},\\
- \epsilon_{b}p^{\frac{n+s-1}{2}}   &  \mbox{if $b \in S_{f}, f^*(b)\neq0$}.\\
\end{array} \right.
 \end{eqnarray*}
 \end{enumerate}
\end{lemma}

\begin{IEEEproof}
From the Fourier expansion of the quadratic character $\eta_1$ in Lemma \ref{Fourier1} and the properties of Gaussian sums in Lemma \ref{quadGuasssum1}, it follows that
 \begin{eqnarray}\label{eqn1}
 \nonumber 
S&=&\frac{1}{p}\sum_{x\in V_{n}^{(p)}}\left(\sum_{\varphi \in \widehat{\mathbb{F}}_{p}}G(\eta_1,\overline{\varphi})\varphi\left(f(x)\right)\right)\chi_b(x)\\
 \nonumber
 &=&\frac{1}{p}\sum_{x\in V_{n}^{(p)}}\left(\sum_{y \in \gf_{p}}G(\eta_1,\overline{\varphi}_y)\varphi_y\left(f(x)\right)\right)\chi_b(x)\\
  \nonumber
 &=&\frac{1}{p}G(\eta_1,\overline{\varphi}_1)\sum_{x\in V_{n}^{(p)}}\chi_b(x)+\frac{1}{p}\sum_{y\in \gf_{p}^*}G(\eta_1,\overline{\varphi}_y)\sum_{x\in V_{n}^{(p)}}\varphi_y(f(x))\chi_b(x)\\
 \nonumber
&=&0+\frac{1}{p}\sum_{y\in \gf_{p}^*}G(\eta_1,\overline{\varphi}_y)\sum_{x\in V_{n}^{(p)}}\varphi_1(yf(x))\chi_b(x)\\
 \nonumber
&=&\frac{1}{p}G(\eta_1,\varphi_1)\sum_{y\in \gf_{p}^*}\eta_1(-y)\sum_{x \in V_{n}^{(p)}}\zeta_p^{yf(x)+\langle b, x\rangle_{n}}\\
 \nonumber
&=&\frac{1}{p}G(\eta_1,\varphi_1)\sum_{y\in \gf_{p}^*}\eta_1(-y)\sum_{x \in V_{n}^{(p)}}\zeta_p^{y\left(f(x)-\langle -\frac{b}{y}, x\rangle_{n}\right)}\\
&=&\frac{1}{p}G(\eta_1,\varphi_1)\sum_{y\in \gf_{p}^*}\eta_1(-y)\sigma_{y}\left(W_{f}\left(\frac{-b}{y}\right)\right),
\end{eqnarray}
where $\sigma_y$ is the automorphism of $\mathbb{Q(}\zeta_p)$ defined by $ \sigma_y(\zeta_p)=\zeta_p^{y}$. We continue the proof by considering the following cases.

Case 1: If $b \notin S_{f}$, then  $W_{f}\left(\frac{-b}{y}\right)=0$. Thus we have 
$$S=0.$$

Case 2: If $b \in S_{f}$ and $n+s$ is even, by Lemma \ref{YWei}, we have 
\begin{eqnarray*}
S&=&\frac{1}{p}G(\eta_1,\varphi_1)\sum_{y\in \gf_{p}^*}\eta_1(-y)\sigma_{y}\left(\epsilon_{b}p^{\frac{n+s}{2}}\zeta_{p}^{f^*(\frac{-b}{y})}\right)\\
&=&\frac{1}{p}G(\eta_1,\varphi_1)\sum_{y\in \gf_{p}^*}\eta_1(-y)\epsilon_{b}p^{\frac{n+s}{2}}\zeta_{p}^{y^{1-h_b}f^*(b)}\\
&=&\frac{\epsilon_{b}p^{\frac{n+s}{2}}}{p}G(\eta_1,\varphi_1)\sum_{y\in \gf_{p}^*}\eta_1(-1)\eta_{1}(\widetilde{y})\zeta_{p}^{\widetilde{y}f^*(b)},
\end{eqnarray*}
where \(h_b = h\) if \(b \in B_{+}(f)\) and \(h_b = h'\) if \(b \in B_{-}(f)\), and since \(\gcd(h_b-1, p-1)=1\), we have \(\widetilde{y} := y^{1-h_b}\) runs through \(\gf_{p}^{*}\) when \(y\) runs through \(\gf_{p}^{*}\), with \(\eta_{1}(\widetilde{y}) = \eta_{1}(y)\).
 Thus we have 
\begin{eqnarray*}
S&=&\frac{\epsilon_{b}p^{\frac{n+s}{2}}}{p}G(\eta_1,\varphi_1)\eta_{1}(-1)G(\eta_1,\varphi_1)\eta_{1}(f^*(b))\\
&=&\left\{
\begin{array}{lll}
0  &   \mbox{if  $ f^*(b)=0$},\\
\epsilon_{b} p^{\frac{{n+s}}{2}}   &   \mbox{if $ f^*(b) \in SQ$},\\
-\epsilon_{b} p^{\frac{{n+s}}{2}}    &  \mbox{if $ f^*(b)\in NSQ$},\\
\end{array} \right.
\end{eqnarray*}
where $G(\eta_0,\varphi_1)G(\eta_0,\varphi_1)=\eta_{0}(-1)p.$

Case 3: If $b \in S_{f}$ and $n+s$ is odd, by Lemma \ref{YWei}, we have 
\begin{eqnarray*}
S&=&\frac{1}{p}G(\eta_1,\varphi_1)\sum_{y\in \gf_{p}^*}\eta_1(-y)\sigma_{y}(\epsilon_{b}p^{\frac{n+s-1}{2}}\sqrt{p^*}\zeta_{p}^{f^*(\frac{-b}{y})})\\
&=&\frac{1}{p}G(\eta_1,\varphi_1)\sum_{y\in \gf_{p}^*}\eta_1(-y)\epsilon_{b}p^{\frac{n+s-1}{2}}\sqrt{p^*}\eta_{1}(y)\zeta_{p}^{yf^*(\frac{-b}{y})}\\
&=&\frac{1}{p}\epsilon_{b}p^{\frac{n+s-1}{2}}\sqrt{p^*}G(\eta_1,\varphi_1)\sum_{y\in \gf_{p}^*}\eta_1(-1)\zeta_{p}^{y^{1-h_b}f^*(b)}\\
&=&\left\{
\begin{array}{lll}
\epsilon_{b} (p-1)p^{\frac{{n+s-1}}{2}}   &   \mbox{if $ f^*(b) = 0$},\\
-\epsilon_{b} p^{\frac{{n+s-1}}{2}}    &  \mbox{if $ f^*(b)\neq 0$},\\
\end{array} \right.
\end{eqnarray*}
where \(h_b = h\) if \(b \in B_{+}(f)\) and \(h_b = h'\) if \(b \in B_{-}(f)\),  $y^{1-h_b}$ runs through $\gf_{p}^*$ when $y$ runs through $\gf_{p}^*$ due to $\gcd(h_b-1,p-1)=1$ and $G(\eta_{1},\varphi_1)=\sqrt{p^*}$.
\end{IEEEproof}

\begin{lemma}\label{lem-42}
Let $T:=\sum\limits_{\substack{x \in V_n^{(p)}\\f(x) =0}}\chi_b(x)$, $b \in V_{n}^{(p)}\setminus\{0\}$, $p$ be an odd prime and $f(x): V_{n}^{(p)}\rightarrow \gf_{p}$ be an $s$-plateaued function belonging to $\mathscr{F}$. Let $\epsilon_b$ be defined in Subsection \ref{subsectionF}.  Then we have the following results.
 \begin{enumerate}
\item If $n+s$ is even, then
 \begin{eqnarray*}
 T=\left\{
\begin{array}{lll}
0  &   \mbox{if $b \notin S_{f}$ },\\
\epsilon_{b}\frac{p-1}{p} p^{\frac{{n+s}}{2}}   &   \mbox{if $b \in S_{f}, f^*(b)=0 $},\\
-\epsilon_{b} \frac{1}{p}p^{\frac{{n+s}}{2}}    &  \mbox{if $b \in S_{f}, f^*(b)\neq0$}.
\end{array} \right.
 \end{eqnarray*}
 \item If $n+s$ is odd, then
 \begin{eqnarray*}
 T=\left\{
\begin{array}{lll}
0  &   \mbox{if $b \notin S_{f}$ or $b \in S_f, f^*(b)=0$},\\
\epsilon_{b}\frac{1}{p} p^{\frac{n+s-1}{2}}p^*   &   \mbox{if $b \in S_{f}, f^*(b)\in SQ$},\\
- \epsilon_{b}\frac{1}{p}p^{\frac{n+s-1}{2}}p^*    &  \mbox{if $b \in S_{f}, f^*(b)\in NSQ$},\\
\end{array} \right.
 \end{eqnarray*}
where $p^*=\eta_{0}(-1)p.$
\end{enumerate}
\end{lemma}

\begin{IEEEproof}
By the orthogonal relation of additive characters, we have
\begin{eqnarray*}
T&=&\frac{1}{p}\sum_{x \in V_n^{(p)}}\chi_b(x)\sum_{y \in \gf_{p}}\varphi_{1}(yf(x))\\
&=&0+\frac{1}{p}\sum_{y \in \gf_{p}^{*}}\sum_{x \in V_n^{(p)}}\zeta_{p}^{yf(x)-\langle -b, x \rangle_n}\\
&=&\frac{1}{p}\sum_{y \in \gf_{p}^{*}}\sum_{x \in V_n^{(p)}}\sigma_{y}\left(\zeta_{p}^{f(x)-\langle -\frac{b}{y}, x \rangle_n}\right)\\
&=&\frac{1}{p}\sum_{y \in \gf_{p}^*}\sigma_{y}\left(W_{f}(-\frac{b}{y})\right).
\end{eqnarray*}
We now consider the following cases.

Case 1: If $b \notin S_{f}$, then $W_{f}(-\frac{b}{y})=0$ and $T=0$. 

Case 2: If $b \in S_{f}$ and $n+s$ is even, by Lemma \ref{YWei}, we have
\begin{eqnarray*}
T&=&\frac{1}{p}\sum_{y \in \gf_{p}^*}\sigma_{y}\left(\epsilon_b p^{\frac{n+s}{2}}\zeta_{p}^{f^*(-\frac{b}{y})}\right)\\
&=&\frac{1}{p}\sum_{y \in \gf_{p}^*}\epsilon_b p^{\frac{n+s}{2}}\zeta_{p}^{y^{1-h_b}f^*(b)}\\
&=&\left\{
\begin{array}{lll}
\epsilon_{b} \frac{p-1}{p}p^{\frac{{n+s}}{2}}   &   \mbox{if $ f^*(b) = 0$},\\
-\epsilon_{b}\frac{1}{p} p^{\frac{{n+s}}{2}}    &  \mbox{if $ f^*(b)\neq 0$},
\end{array} \right.
\end{eqnarray*}
where \(h_b = h\) if \(b \in B_{+}(f)\) and \(h_b = h'\) if \(b \in B_{-}(f)\).

Case 3: If $b \in S_{f}$ and $n+s$ is odd, by Lemma \ref{YWei},  we have
\begin{eqnarray*}
T&=&\frac{1}{p}\sum_{y \in \gf_{p}^*}\sigma_{y}\left(\epsilon_b p^{\frac{n+s-1}{2}}\sqrt{p^*}\zeta_{p}^{f^*(-\frac{b}{y})}\right)\\
&=&\frac{1}{p}\sum_{y \in \gf_{p}^*}\epsilon_b p^{\frac{n+s-1}{2}}\sqrt{p^*}\eta_{1}(y)\zeta_{p}^{y^{1-h_b}f^*(b)}\\
&=&\left\{
\begin{array}{lll}
0   &   \mbox{if $ f^*(b) = 0$}\\
\frac{1}{p}\epsilon_b p^{\frac{n+s-1}{2}}\sqrt{p^*}G(\eta_{1},\varphi_1)\eta_{1}(f^*(b))    &  \mbox{if $ f^*(b)\neq 0$}\\
\end{array} \right.\\
&=&\left\{
\begin{array}{lll}
0   &   \mbox{if $ f^*(b) = 0$},\\
\frac{1}{p}\epsilon_b p^{\frac{n+s-1}{2}}p^*  &  \mbox{if $ f^*(b)\in SQ$},\\
-\frac{1}{p}\epsilon_b p^{\frac{n+s-1}{2}}p^*  &  \mbox{if $ f^*(b)\in NSQ$},\\
\end{array} \right.\\
\end{eqnarray*}
where  \(h_b = h\) if \(b \in B_{+}(f)\) and \(h_b = h'\) if \(b \in B_{-}(f)\), and $G(\eta_{1},\varphi_1)=\sqrt{p^*}$ by Lemma \ref{quadGuasssum3}.
\end{IEEEproof}

\begin{lemma}\label{lem-43}
Let $S_1:=\sum\limits_{\substack{x \in V_n^{(p)}\\f(x) \in SQ\bigcup\{0\}}}\chi_b(x)$, $b \in V_{n}^{(p)}\setminus\{0\}$, $p$ be an odd prime and $f(x): V_{n}^{(p)}\rightarrow \gf_{p}$ be an $s$-plateaued function belonging to $\mathscr{F}$. Let $\epsilon_b$ be defined in Subsection \ref{subsectionF}. Then we have the following results. 
\begin{enumerate}
\item If $n+s$ is even, then we have
\begin{eqnarray*}
S_1=\left\{
\begin{array}{ll}
0  &   \mbox{if $b \notin S_f$},\\
\frac{p-1}{2p}\epsilon_b p^{\frac{n+s}{2}}  &   \mbox{if $b \in S_f, f^*(b)=0$ or $b \in S_f, f^*(b) \in SQ$},\\
-\frac{p+1}{2p}\epsilon_b p^{\frac{n+s}{2}}  &   \mbox{if $b \in S_f, f^*(b) \in NSQ$}.
\end{array} \right.
\end{eqnarray*}

\item If $n+s$ is odd, then we have
\begin{eqnarray*}
S_1=\left\{
\begin{array}{ll}
0  &   \mbox{if $b \notin S_f$},\\
\frac{p-1}{2}\epsilon_b p^{\frac{n+s-1}{2}}  &   \mbox{if $b \in S_f, f^*(b) =0$},\\
-\frac{1}{2}\epsilon_b p^{\frac{n+s-1}{2}}+\frac{1}{2p}\epsilon_{b}p^{\frac{n+s-1}{2}}p^*  &   \mbox{if $b \in S_f, f^*(b) \in SQ$},\\
-\frac{1}{2}\epsilon_b p^{\frac{n+s-1}{2}}-\frac{1}{2p}\epsilon_{b}p^{\frac{n+s-1}{2}}p^*  &   \mbox{if $b \in S_f, f^*(b) \in NSQ$}.
\end{array} \right.
\end{eqnarray*}
\end{enumerate}
\end{lemma}

\begin{IEEEproof} Note that 
 \begin{eqnarray*}
 S_{1}&=&\sum\limits_{\substack{x \in V_{n}^{(p)}\\F(x) \in SQ\bigcup\{0\}}}\chi_b(x)\\
&=&\sum_{x \in V_{n}^{(p)}}\frac{1+\eta_1(f(x))}{2}\chi_{b}(x)-\frac{1}{2}\sum\limits_{\substack{x \in V_{n}^{(p)}\\f(x)=0}}\chi_b(x)+\sum\limits_{\substack{x \in V_{n}^{(p)}\\f(x)=0}}\chi_b(x)\\
&=&\frac{1}{2}\sum_{x \in V_{n}^{(p)}}\eta_1(f(x))\chi_{b}(x)+\frac{1}{2}\sum\limits_{\substack{x \in V_{n}^{(p)}\\f(x)=0}}\chi_b(x).
 \end{eqnarray*}
  Then the desired conclusion follows from Lemmas \ref{lem-41} and \ref{lem-42}.
 \end{IEEEproof}

\begin{lemma}\label{lem-44}
Let $S_2:=\sum\limits_{\substack{x \in V_n^{(p)}\\f(x) \in NSQ\bigcup\{0\}}}\chi_b(x)$, $b \in V_{n}^{(p)}\setminus\{0\}$, $p$ be an odd prime and $f(x): V_{n}^{(p)}\rightarrow \gf_{p}$ be an $s$-plateaued function belonging to $\mathscr{F}$. Let $\epsilon_b$ be defined in Subsection \ref{subsectionF}. 

%if $n+s$ is even
%\begin{eqnarray*}
%S_1=\left\{
%\begin{array}{ll}
%0  &   \mbox{if $b \notin S_f$},\\
%\frac{p-1}{2p}\epsilon_b p^{\frac{n+s}{2}}  &   \mbox{if $b \in S_f, f^*(b)=0$ or $b \in S_f, f^*(b) \in NSQ$},\\
%-\frac{p+1}{2p}\epsilon_b p^{\frac{n+s}{2}}  &   \mbox{if $b \in S_f, f^*(b) \in SQ$},\\
%\end{array} \right.
%\end{eqnarray*}

If $n+s$ is odd, then we have
\begin{eqnarray*}
S_2=\left\{
\begin{array}{ll}
0  &   \mbox{if $b \notin S_f$},\\
-\frac{p-1}{2}\epsilon_b p^{\frac{n+s-1}{2}}  &   \mbox{if $b \in S_f, f^*(b) =0$},\\
\frac{1}{2}\epsilon_b p^{\frac{n+s-1}{2}}+\frac{1}{2p}\epsilon_{b}p^{\frac{n+s-1}{2}}p^*  &   \mbox{if $b \in S_f, f^*(b) \in SQ$},\\
\frac{1}{2}\epsilon_b p^{\frac{n+s-1}{2}}-\frac{1}{2p}\epsilon_{b}p^{\frac{n+s-1}{2}}p^*  &   \mbox{if $b \in S_f, f^*(b) \in NSQ$},\\
\end{array} \right.
\end{eqnarray*}
\end{lemma}

\begin{IEEEproof}
With a proof similar to that of Lemma  \ref{lem-43}, 
the desired conclusion directly follows from  Lemmas \ref{lem-41} and \ref{lem-42}. 
\end{IEEEproof}

%\begin{lemma}\label{lem-31}
%Let $S_1:=\sum\limits_{\substack{x \in V_n^{(p)}\\F(x) \in SQ\bigcup\{0\}}}\chi_b(x)$, $b \in V_{n}^{(p)}\setminus\{0\}$, $p$ be an odd prime and $F(x)$ be a vectorial dual-bent function satisfying . Then we have
%\begin{eqnarray*}
%S_1=\left\{
%\begin{array}{ll}
%\frac{(p^m-1)}{2}\varepsilon p^{\frac{n}{2}-m}  &   \mbox{if $F^*(b)=0$ or $F^*(b) \in SQ$},\\
%-\frac{p^m+1}{2}\varepsilon p^{\frac{n}{2}-m}  &   \mbox{if $F^*(b) \in NSQ$},\\
%\end{array} \right.
%\end{eqnarray*}
% where $\varepsilon\in \{\pm 1\}$ was defined in Condition .
%\end{lemma}
%\begin{IEEEproof}
% \begin{eqnarray*}
% S_{1}&=&\sum\limits_{\substack{x \in V_{n}^{(p)}\\F(x) \in SQ\bigcup\{0\}}}\chi_b(x)\\
%&=&\sum_{x \in V_{n}^{(p)}}\frac{1+\eta_m(F(x))}{2}\chi_{b}(x)-\frac{1}{2}\sum\limits_{\substack{x \in V_{n}^{(p)}\\f(x)=0}}\chi_b(x)+\sum\limits_{\substack{x \in V_{n}^{(p)}\\f(x)=0}}\chi_b(x)\\
%&=&\frac{1}{2}\sum_{x \in V_{n}^{(p)}}\eta_m(F(x))\chi_{b}(x)+\frac{1}{2}\sum\limits_{\substack{x \in V_{n}^{(p)}\\f(x)=0}}\chi_b(x)\\
% \end{eqnarray*}
 %According to Lemma () and (), we can obtain the desired conclusion.
%\end{IEEEproof}

\begin{lemma}\label{333}\cite{Heng1}
Let $p$ be an odd prime. Let $S':=\sum_{x\in V_n^{(p)}}\eta_{m}(F(x))\chi_b(x)$  for $b \in V_n^{(p)}\setminus \{0\}$ and $F(x)$ be a vectorial dual-bent function satisfying \textbf{Condition A}. Then we have 
\begin{eqnarray*}
 S'=\left\{
\begin{array}{ll}
\frac{p^m-1}{p^m}\upsilon p^{\frac{n}{2}} (-1)^{m-1}\epsilon^m\sqrt{p^m}\eta_m(-1)   &   \mbox{if $F^*(b)=0$},\\
-\frac{1}{p^m}\upsilon p^{\frac{n}{2}} (-1)^{m-1}\epsilon^m\sqrt{p^m}\eta_m(-1)    &   \mbox{if $F^*(b)\neq0$},
\end{array} \right.
\end{eqnarray*}
where $\upsilon \in \{\pm\epsilon^{m}\}$ is a constant  for $\epsilon:=\sqrt{(-1)^{\frac{p-1}{2}}}$.
\end{lemma}

\begin{lemma}\label{lem15}\cite{Heng1}
Let $T':=\sum\limits_{\substack{x \in V_n^{(p)}\\F(x)=0}}\chi_b(x)$, $b \in V_{n}^{(p)}\setminus\{0\}$, $p$ be an odd prime and $F(x)$ be a vectorial dual-bent function satisfying \textbf{Condition A}. Then we have
\begin{eqnarray*}
T'&=&\left\{
\begin{array}{lll}
0  &   \mbox{if $F^*(b)=0$},\\
\frac{1}{p^m}\upsilon p^{\frac{n}{2}}(-1)^{m-1}\epsilon^m\sqrt{p^m}  &   \mbox{if $\eta_m(F^*(b))=1$},\\
-\frac{1}{p^m}\upsilon p^{\frac{n}{2}}(-1)^{m-1}\epsilon^m\sqrt{p^m}   &   \mbox{if $\eta_m(F^*(b))=-1$},\\
\end{array} \right.
\end{eqnarray*}
where $\upsilon \in \{\pm\epsilon^{m}\}$ is a constant  for $\epsilon:=\sqrt{(-1)^{\frac{p-1}{2}}}$.
\end{lemma}

\begin{lemma}\label{lem-45}
Let $S_3:=\sum\limits_{\substack{x \in V_n^{(p)}\\F(x) \in SQ\bigcup\{0\}}}\chi_b(x)$, $b \in V_{n}^{(p)}\setminus\{0\}$, $p$ be an odd prime and $F(x)$ be a vectorial dual-bent function satisfying \textbf{Condition A}. Then we have
\begin{eqnarray*}
S_3=\left\{
\begin{array}{ll}
\frac{p^m-1}{2p^m}\upsilon (-1)^{m-1}\epsilon^m\eta_{m}(-1)p^{\frac{n+m}{2}}  &   \mbox{if $F^*(b)=0$},\\
(1-\eta_{m}(-1))\frac{1}{2p^m}\upsilon (-1)^{m-1}\epsilon^m\eta_{m}(-1)p^{\frac{n+m}{2}}& \mbox{if$F^*(b) \in SQ$}\\
(-1-\eta_{m}(-1))\frac{1}{2p^m}\upsilon (-1)^{m-1}\epsilon^m\eta_{m}(-1)p^{\frac{n+m}{2}}  &   \mbox{if $F^*(b) \in NSQ$},\\
\end{array} \right.
\end{eqnarray*}
where $\upsilon \in \{\pm\epsilon^{m}\}$ is a constant  for $\epsilon:=\sqrt{(-1)^{\frac{p-1}{2}}}$.
\end{lemma}
\begin{IEEEproof}
With a proof similar to that of Lemma \ref{lem-43}, the desired conclusion follows from Lemmas \ref{333} and \ref{lem15}.
\end{IEEEproof}

\begin{lemma}\label{lem-46}
Let $S_4:=\sum\limits_{\substack{x \in V_n^{(p)}\\F(x) \in NSQ\bigcup\{0\}}}\chi_b(x)$, $b \in V_{n}^{(p)}\setminus\{0\}$, $p$ be an odd prime and $F(x)$ be a vectorial dual-bent function satisfying \textbf{Condition A}. Then we have
\begin{eqnarray*}
S_4=\left\{
\begin{array}{ll}
-\frac{p^m-1}{2p^m}\upsilon (-1)^{m-1}\epsilon^m\eta_{m}(-1)p^{\frac{n+m}{2}}  &   \mbox{if $F^*(b)=0$},\\
(1+\eta_{m}(-1))\frac{1}{2p^m}\upsilon (-1)^{m-1}\epsilon^m\eta_{m}(-1)p^{\frac{n+m}{2}}& \mbox{if$F^*(b) \in SQ$}\\
(\eta_{m}(-1)-1)\frac{1}{2p^m}\upsilon (-1)^{m-1}\epsilon^m\eta_{m}(-1)p^{\frac{n+m}{2}}  &   \mbox{if $F^*(b) \in NSQ$},\\
\end{array} \right.
\end{eqnarray*}
where $\upsilon \in \{\pm\epsilon^{m}\}$ is a constant  for $\epsilon:=\sqrt{(-1)^{\frac{p-1}{2}}}$.
\end{lemma}
\begin{IEEEproof}
With a proof similar to that of Lemma \ref{lem-43}, the desired conclusion follows from Lemmas \ref{333} and \ref{lem15}.
\end{IEEEproof}

\begin{lemma}\label{lem-dim}
Let $p$ be a prime, $q=p^e$, $e \mid n$, $r= \frac{n}{e}$ and $0 \in D  \subseteq V_{n}^{(p)}$. Then the dimension of the code $\mathcal{C}_{f_{D}}$ defined in Equation (\ref{eq-3}) is $r+1$.
\end{lemma}

\begin{IEEEproof}
It suffices to prove that, as $(a,b)$ runs over $\gf_q \times V_{n}^{(p)}$, the zero codeword appears exactly once in $\mathcal{C}_{f_{D}}$.
Let $a f_{D}(x)+\langle b,x\rangle_{n/e}=0$ for all $x \in V_{n}^{(p)}$. Since $0 \in D$, we have $f_{D}(0)=1$. Then we have $a f_{D}(0)+\langle b,0\rangle_{n/e}=0$, which implies $a=0$. 
We then have $\langle b,x\rangle_{n/e}=0$ for all $x \in V_{n}^{(p)}$, which implies $b=0$. Then we deduce that the zero codeword appears only when $a=b=0$.
The desired conclusion follows. 
\end{IEEEproof}

In what follows, we study the Hamming weights of codewords in $\mathcal{C}_{f_D}$, provided that $0 \in D$ and $D$ is $\gf_{q}^*$-invariant.

\begin{lemma} \label{the-41}
let $p$ be a prime, $q=p^e$, $e \mid n$, $r= \frac{n}{e}$ and $0 \in D \subseteq V_{n}^{(p)}$ such that $D$ is  $\gf_{q}^*$-invariant.
Let $\mathcal{C}_{f_D}$ be the linear code defined by Equation (\ref{eq-3}).
Then $\mathcal{C}_{f_D}$  is a $[p^n, r+1]$ code over $\gf_q$. For a codeword $\textbf{c}(a,b):=\left(a f_{D}(x)+\langle b,x\rangle_{n/e}\right)_{x\in V_{n}^{(p)}}$ with $(a,b)\in \gf_{q}\times V_{n}^{(p)}$, its Hamming weight is given by
\begin{eqnarray*}
\wt(\textbf{c}(a,b))
=\left\{
\begin{array}{ll}
0 & \mbox{if $a=b=0$,}\\
|D| & \mbox{if $a\neq0, b=0$,}\\
p^n-p^{n-e} & \mbox{if $a=0, b\neq0$,}\\
p^n-p^{n-e}+\chi_{b}(D) & \mbox{if $a\neq0, b\neq0$.}\\
\end{array}\right.
 \end{eqnarray*}
\end{lemma}

\begin{IEEEproof}
For a codeword $\textbf{c}(a,b):=\left(a f_{D}(x)+\langle b,x\rangle_{n/e}\right)_{x\in V_{n}^{(p)}}$ with $(a,b)\in \gf_{q}\times V_{n}^{(p)}$, we consider the following cases. 

Case 1: If $a \neq 0$ and $b=0$, then $\wt(\textbf{c}(a,b))=|D|$.

Case 2: If $a = 0$ and $b\neq0$, then $\wt(\textbf{c}(a,b))=p^n-p^{n-e}$.

Case 3: If $a \neq 0$ and $b\neq0$, by the orthogonal relation of additive characters, we have 
\begin{eqnarray*}
\wt(\textbf{c}(a,b))&=&p^n-\sharp\{x \in V_{n}^{(p)}: af_{D}(x)+\langle b,x\rangle_{n/e}=0\}\\
&=&p^n-\frac{1}{q}\sum_{y \in \gf_{q}}\sum_{x\in V_{n}^{(p)}}\zeta_{p}^{\tr_{q/p}\left(y\left(af_{D}(x)+\langle b,x\rangle_{n/e}\right)\right)}\\
&=&p^n-p^{n-e}-\frac{1}{q}\sum_{y \in \gf_{q}^*}\sum_{x\in V_{n}^{(p)}}\zeta_{p}^{\tr_{q/p}\left(y\left(af_{D}(x)+\langle b,x\rangle_{n/e}\right)\right)}\\
&=&p^n-p^{n-e}-\frac{1}{q}\sum_{y \in \gf_{q}^{*}}\sum_{x\in D}\zeta_{p}^{\tr_{q/p}\left(ya+\langle b,yx\rangle_{n/e}\right)}-\frac{1}{q}\sum_{y \in \gf_{q}^*}\sum_{x\in V_{n}^{(p)}\backslash D}\zeta_{p}^{\langle yb,x\rangle_n}\\
&=&p^n-p^{n-e}-\frac{1}{q}\sum_{y \in \gf_{q}^*}\varphi_{1}(ya)\sum_{x\in D}\chi_{b}(yx)-\frac{1}{q}\sum_{y \in \gf_{q}^*}\sum_{x \in V_{n}^{(p)}}\chi_{b}(yx)+\frac{1}{q}\sum_{y \in \gf_{q}^*}\sum_{x\in D}\chi_{b}(yx)\\
&=&p^n-p^{n-e}+\frac{1}{q}\sum_{x\in D}\chi_{b}(x)+\frac{q-1}{q}\sum_{x \in D}\chi_{b}(x)\\
&=&p^n-p^{n-e}+\sum_{x \in D}\chi_{b}(x),
\end{eqnarray*}
where the fifth equation holds due to the condition that $D$ is $\gf_{q}^*$-invariant.
By Lemma \ref{lem-dim}, the dimension is $r+1$. Then the desired conclusions follow. 
\end{IEEEproof}

\begin{table}[!h]
\begin{center}
\caption{The weight distribution of $\mathcal{C}_{f_D}$ in Theorem \ref{theorem-41}.}\label{tab-41}
\begin{tabular}{cc} \hline
Weight  &  Frequency   \\ \hline
0 & 1\\

$p^{n-m}+\frac{p^m-1}{2}(p^{n-m}+\upsilon (-1)^{m-1}\epsilon^m\eta_{m}(-1)p^{\frac{n-m}{2}})$          &  $p^e-1$\\

$p^n-p^{n-e}$  & $p^n-1$ \\

$p^n-p^{n-e}+\frac{p^m-1}{2p^m}\upsilon (-1)^{m-1}\epsilon^m\eta_{m}(-1)p^{\frac{n+m}{2}}$     & $(p^e-1)(p^{n-m}-1)$ \\

$p^n-p^{n-e}+(1-\eta_{m}(-1))\frac{1}{2p^m}\upsilon (-1)^{m-1}\epsilon^m\eta_{m}(-1)p^{\frac{n+m}{2}}$  & $\frac{(p^e-1)(p^m-1)}{2}\left(p^{n-m}+\upsilon^{-1} (-1)^{m-1}\epsilon^m\eta_{m}(-1)p^{\frac{n-m}{2}}\right)$\\
$p^n-p^{n-e}+(-1-\eta_{m}(-1))\frac{1}{2p^m}\upsilon (-1)^{m-1}\epsilon^m\eta_{m}(-1)p^{\frac{n+m}{2}}$  & $\frac{(p^e-1)(p^m-1)}{2}\left(p^{n-m}-\upsilon^{-1} (-1)^{m-1}\epsilon^m\eta_{m}(-1)p^{\frac{n-m}{2}}\right)$\\
\hline
\end{tabular}
\end{center}
\end{table}

In the following two theorems, we use vectorial dual-bent functions to construct self-orthogonal minimal codes that violate the Ashikhmin–Barg condition.

 \begin{theorem}\label{theorem-41}
Let $q = p^e > 3$ be an odd prime power, and let $n \geq 3m$, $e \mid m \mid n$. Suppose $F(x): V_n^{(p)} \rightarrow \mathbb{F}_{p^m}$ is a vectorial dual-bent function satisfying \textbf{Condition A}. Define $D = \left\{ x \in V_n^{(p)} : F(x) \in SQ \cup \{0\} \right\}$. Then the code $\mathcal{C}_{f_D}$ defined in Equation (\ref{eq-3}) is a self-orthogonal minimal four-weight code over $\gf_q$ with parameters
\[
\left[ p^n,\ r+1,\ p^{n-m} + \frac{p^m-1}{2} \left( p^{n-m} + \upsilon (-1)^{m-1} \epsilon^m \eta_{m}(-1) p^{\frac{n-m}{2}} \right) \right]
\]
and it satisfies $\dfrac{w_{\min}}{w_{\max}} < \dfrac{q-1}{q}$. The weight distribution of $\mathcal{C}_{f_D}$ is given in Table \ref{tab-41}, where $\upsilon \in \{\pm\epsilon^{m}\}$ is a constant and $\epsilon := \sqrt{(-1)^{\frac{p-1}{2}}}$.
\end{theorem}

\begin{IEEEproof}
Let $F(x): V_n^{(p)} \rightarrow \mathbb{F}_{p^m}$ be a vectorial dual-bent function satisfying \textbf{Condition A}. If $\gcd(l-1, p^m-1)=1$ with $p$ an odd prime, then $l$ must be even. If $F(x) \in SQ \cup \{0\}$, then $F(ax)=a^lF(x) \in SQ \cup \{0\}$ for any $a\in \gf_q^*$. Consequently, $D$ is $\gf_q^*$-invariant. Since $D$ is $\gf_q^*$-invariant and $p \mid |D|$ for $n \geq 3m$ by Lemma \ref{lemm-21}, it follows from Corollary~\ref{cor-41} that the code $\mathcal{C}_{f_D}$ is self-orthogonal.

 For a codeword $\textbf{c}(a,b):=\left(a f_{D}(x)+\langle b,x\rangle_{n/e}\right)_{x\in V_{n}^{(p)}}$ with $(a,b)\in \gf_{q}\times V_{n}^{(p)}$, by Lemmas \ref{lem-45} and  \ref{the-41}, we have 
 \begin{eqnarray*}
&& \wt\left(\textbf{c}(a,b)\right)\\
&=& \left\{
\begin{array}{lll}
0&\mbox{if $a=0$, $b= 0$ },\\
p^{n-m}+\frac{p^m-1}{2}(p^{n-m}+\upsilon (-1)^{m-1}\epsilon^m\eta_{m}(-1)p^{\frac{n-m}{2}})  &   \mbox{if $a\neq0$, $b= 0$ },\\
p^n-p^{n-e} &  \mbox{if $a=0$, $b \neq 0$},\\
p^n-p^{n-e}+\frac{p^m-1}{2p^m}\upsilon (-1)^{m-1}\epsilon^m\eta_{m}(-1)p^{\frac{n+m}{2}} & \mbox{if $a\neq0,  b \neq 0$, $F^*(b)=0$.}\\
p^n-p^{n-e}+(1-\eta_{m}(-1))\frac{1}{2p^m}\upsilon (-1)^{m-1}\epsilon^m\eta_{m}(-1)p^{\frac{n+m}{2}} &  \mbox{if $a \neq 0, b \neq 0, F^*(b)\in SQ$.}\\
p^n-p^{n-e}+(-1-\eta_{m}(-1))\frac{1}{2p^m}\upsilon (-1)^{m-1}\epsilon^m\eta_{m}(-1)p^{\frac{n+m}{2}} &  \mbox{if $a \neq 0, b \neq 0, F^*(b)\in NSQ$.}\\
\end{array} \right.
 \end{eqnarray*}
 Denote these Hamming weights by $\wt_i$ for $0\leq i\leq 5$ in order. 
  According to Lemma \ref{lemm-21}, it is easy to determine the frequency $A_{\wt_i}$ of each weight. Then the weight distribution of $\mathcal{C}_{f_D}$ is given by the Table \ref{tab-41}.

In what follows, we prove that $\mathcal{C}_{f_D}$ is a minimal code violating the Ashikhmin-Barg condition by Lemma \ref{minimal1}. 
Firstly, we consider the case where $\eta_{m}(-1)=1$ and $\upsilon (-1)^{m-1}\epsilon^m\eta_{m}(-1)=1$. 
 It is easy to see that
  \begin{eqnarray*}
 \frac{\wt_{\min}}{\wt_{\max}}=\frac{p^{n-m}+\frac{p^m-1}{2}(p^{n-m}+p^{\frac{n-m}{2}})}{p^n-p^{n-e}+\frac{p^m-1}{2}p^{\frac{n-m}{2}}}<\frac{q-1}{q}.
 \end{eqnarray*}
For $i = 1,2,3,4$, define $E_i$ as the set of all codewords of weight $\wt_i$. Then
\begin{enumerate}
\item $E_1=\{ \mathbf{c}_{a, 0} \in \mathcal{C}_{f_D} | a\neq0, b=0 \}$,
\item $E_2=\{ \mathbf{c}_{0, b} \in \mathcal{C}_{f_D} | a=0, b \neq0 ~\mbox{or}~ a \neq 0, b \neq 0, F^*(b)\in SQ \}$,
\item $E_3=\{ \mathbf{c}_{a, b} \in \mathcal{C}_{f_D}| a\neq0, b\neq 0, F^*(b)=0 \}$,
\item $E_4=\{ \mathbf{c}_{a, b} \in \mathcal{C}_{f_D}| a \neq 0, b \neq 0, F^*(b))\in NSQ\}$.    
\end{enumerate}
Note that $\wt_{3}>\wt_{2}>\wt_{4}>\wt_{1}$. 
Let $(a_i, b_i) \in \gf_q \times V_n^{(p)} \setminus \{(0,0)\}$ for $i = 1, 2$. Then for any two linearly independent codewords $\mathbf{a} = \mathbf{c}(a_1, b_1)$ and $\mathbf{b} = \mathbf{c}(a_2, b_2)$ in $\mathcal{C}_{f_D}$, their coverage can be categorized into the following four cases.

Case 1: For $\mathbf{a}=\mathbf{c}(a_1,b_1) \in E_1$ and  $\mathbf{b}=\mathbf{c}(a_2,b_2) \in \bigcup_{i=1}^{4}E_{i}$, we have $$\sum_{z \in \gf_{q}^*}\wt(\mathbf{a}+z\mathbf{b})\geq(q-1)\wt_1>(q-1)\wt_1-\wt_1\geq(q-1)\wt(\mathbf{a})-\wt(\mathbf{b}),$$
which satisfies the inequality in Lemma \ref{minimal1}.

Case 2: Let $\mathbf{a}=\mathbf{c}(a_1,b_1) \in E_2$. Then $b_1\neq 0$ by the definition of $E_2$. 

Subcase 2.1: If $\mathbf{b}=\mathbf{c}(a_2,b_2) \in E_1$, then $b_2=0$ and $\mathbf{a}+z\mathbf{b}=\mathbf{c}(a_1+za_2,b_1)\not\in  E_1$ for $z\in \gf_q^*$.
This yields that
$$\sum_{z \in \gf_{q}^*}\wt(\mathbf{a}+z\mathbf{b})\geq(q-1)\wt_4=(q-1)(p^n-p^{n-e}-p^{\frac{n-m}{2}}).$$
Note that
$$(q-1)\wt(\mathbf{a})-\wt(\mathbf{b})=(q-1)\wt_2-\wt_1=(q-1)(p^n-p^{n-e})-\left( p^{n-m}+\frac{p^m-1}{2}(p^{n-m}+p^{\frac{n-m}{2}})\right).$$
It is verified that
$$\sum_{z \in \gf_{q}^*}\wt(\mathbf{a}+z\mathbf{b})>(q-1)\wt(\mathbf{a})-\wt(\mathbf{b}),$$
which satisfies the inequality in Lemma \ref{minimal1}.

Subcase 2.2: If $\mathbf{b}=\mathbf{c}(a_2,b_2) \in \bigcup_{i=2}^{4}E_i$, then $\mathbf{a}+z\mathbf{b}=\mathbf{c}(a_1+za_2,b_1+zb_2)$ for $z\in \gf_q^*$. 
There exists at most one $z\in \gf_q^*$ such that $\mathbf{a}+z\mathbf{b}\in E_1$ as $b_1\neq 0$. 
Hence, we have
$$\sum_{z \in \gf_{q}^*}\wt(\mathbf{a}+z\mathbf{b})\geq(q-2)\wt_4+\wt_1=(q-2)(p^n-p^{n-e}-p^{\frac{n-m}{2}})+p^{n-m}+\frac{p^m-1}{2}(p^{n-m}+p^{\frac{n-m}{2}}).$$
Note that
$$(q-1)\wt(\mathbf{a})-w(\mathbf{b})\leq (q-1)\wt_2-\wt_4= (q-1)(p^n-p^{n-e})-(p^n-p^{n-e}-p^{\frac{n-m}{2}}).$$
It is verified that
$$\sum_{z \in \gf_{q}^*}\wt(\mathbf{a}+z\mathbf{b})>(q-1)\wt(\mathbf{a})-\wt(\mathbf{b}),$$
which satisfies the inequality in Lemma \ref{minimal1}.

Case 3: Let $\mathbf{a}=\mathbf{c}(a_1,b_1) \in E_3$. Then $a_1\neq 0$, $b_1\neq 0$ and  $F^*(b_1)=0$.

Subcase 3.1: If $\mathbf{b}=\mathbf{c}(a_2,b_2) \in E_1$, then $a_2\neq 0$, $b_2=0$ and $\mathbf{a}+z\mathbf{b}=\mathbf{c}(a_1+za_2,b_1)$ for $z\in \gf_q^*$. Note that $\mathbf{a}+z\mathbf{b}\not \in E_1$ for all $z\in \gf_q^*$. Besides, there exists exactly one $z\in \gf_q^*$ such that $\mathbf{a}+z\mathbf{b}\in E_2$.
Hence, we have 
$$\sum_{z \in \gf_{q}^*}\wt(\mathbf{a}+z\mathbf{b})\geq(q-2)\wt_4+\wt_2=(q-2)(p^n-p^{n-e}-p^{\frac{n-m}{2}})+(p^n-p^{n-e}).$$
Note that
$$(q-1)\wt(\mathbf{a})-\wt(\mathbf{b})=(q-1)\left(p^n-p^{n-e}+\frac{p^m-1}{2}p^{\frac{n-m}{2}}\right)-\left(p^{n-m}+\frac{p^m-1}{2}(p^{n-m}+p^{\frac{n-m}{2}})\right).$$
We then have 
$$\sum_{z \in \gf_{q}^*}\wt(\mathbf{a}+z\mathbf{b})>(q-1)\wt(\mathbf{a})-\wt(\mathbf{b}),$$
which satisfies the inequality in Lemma \ref{minimal1}.

Subcase 3.2: If $\mathbf{b} \in \bigcup_{i=2}^{4}E_i$, then $a_2\neq 0$ and $\mathbf{a}+z\mathbf{b}=\mathbf{c}(a_1+za_2,b_1+zb_2)$ for $z\in \gf_q^*$.
Then there exists at most one $z\in \gf_q^*$ such that $\mathbf{a}+z\mathbf{b}\in E_1$ as $b_1\neq 0$. 
Hence, we have
$$\sum_{z \in \gf_{q}^*}\wt(\mathbf{a}+z\mathbf{b})\geq(q-2)\wt_4+\wt_1=(q-2)(p^n-p^{n-e}-p^{\frac{n-m}{2}})+p^{n-m}+\frac{p^m-1}{2}(p^{n-m}+p^{\frac{n-m}{2}}).$$
Note that
$$(q-1)\wt(\mathbf{a})-\wt(\mathbf{b})\leq (q-1)\left(p^n-p^{n-e}+\frac{p^m-1}{2}p^{\frac{n-m}{2}}\right)-(p^n-p^{n-e}-p^{\frac{n-m}{2}}).$$
Then we can easily verify that
$$\sum_{z \in \gf_{q}^*}\wt(\mathbf{a}+z\mathbf{b})>(q-1)\wt(\mathbf{a})-\wt(\mathbf{b}),$$
which satisfies the inequality in Lemma \ref{minimal1}.

Case 4: Let $\mathbf{a}=\mathbf{c}(a_1,b_1) \in E_4$. Then we have $a_1 \neq 0, b_1 \neq 0, F^*(b_1)\in NSQ$.

Subcase 4.1: If $\mathbf{b} \in E_1$, then $a_2\neq 0$, $b_2=0$ and $\mathbf{a}+z\mathbf{b}=\mathbf{c}(a_1+za_2,b_1)$ for $z\in \gf_q^*$. 
There exists exactly one $z\in \gf_q^*$ such that  $\mathbf{a}+z\mathbf{b}\in E_2$ as $a_1\neq 0$ and $a_2\neq 0$. 
Note that $\mathbf{a}+z\mathbf{b}\not \in E_1$  for all $z\in \gf_q^*$ as $b_1\neq 0$. 
Then we have
$$\sum_{z \in \gf_{q}^*}\wt(\mathbf{a}+z\mathbf{b})\geq(q-2)\wt_4+\wt_2=(q-2)(p^n-p^{n-e}-p^{\frac{n-m}{2}})+(p^n-p^{n-e}).$$
Besides, it is obvious that
$$(q-1)\wt(\mathbf{a})-\wt(\mathbf{b})=(q-1)(p^n-p^{n-e}-p^{\frac{n-m}{2}})-\left(p^{n-m}+\frac{p^m-1}{2}(p^{n-m}+p^{\frac{n-m}{2}})\right).$$
Then we obtain
$$\sum_{z \in \gf_{q}^*}\wt(\mathbf{a}+z\mathbf{b})>(q-1)\wt(\mathbf{a})-\wt(\mathbf{b}),$$
which satisfies the inequality in Lemma \ref{minimal1}.

Subcase 4.2: If $\mathbf{b} \in \bigcup_{i=2}^{4}E_i$, then $b_2\neq 0$ and $\mathbf{a}+z\mathbf{b}=\mathbf{c}(a_1+za_2,b_1+zb_2)$ for $z\in \gf_q^*$.
There exists at most one $z\in \gf_q^*$ such that $\mathbf{a}+z\mathbf{b}\in E_1$. 
Then we have
$$\sum_{z \in \gf_{q}^*}\wt(\mathbf{a}+z\mathbf{b})\geq(q-2)\wt_4+\wt_1=(q-2)(p^n-p^{n-e}-p^{\frac{n-m}{2}})+p^{n-m}+\frac{p^m-1}{2}(p^{n-m}+p^{\frac{n-m}{2}}).$$
Note that
$$(q-1)\wt(\mathbf{a})-\wt(\mathbf{b})\leq (q-1)(p^n-p^{n-e}-p^{\frac{n-m}{2}})-(p^n-p^{n-e}-p^{\frac{n-m}{2}}).$$
Then we deduce 
$$\sum_{z \in \gf_{q}^*}\wt(\mathbf{a}+z\mathbf{b})>(q-1)\wt(\mathbf{a})-\wt(\mathbf{b}),$$
which satisfies the inequality in Lemma \ref{minimal1}.

Based on the above discussions, we conclude that $\mathcal{C}_{f_D}$ is minimal when $\eta_{m}(-1)=1$ and $\upsilon (-1)^{m-1}\epsilon^m\eta_{m}(-1)=1$, by Lemma~\ref{minimal1}. The minimality in the remaining cases can be proved similarly. Hence, the desired conclusions follow.
\end{IEEEproof}

\begin{example}
Let $p=5$, $n=3$, $e=m=1$ and $F(x)=\tr_{p^n/p}(x^2)$. By Equation~\ref{Tr(x^2)1}, we have $\upsilon=1$. Then $\mathcal{C}_{f_D}$ is a self-orthogonal minimal $[125,4,85]$ code over $\gf_5$ with weight enumerator $1+4z^{85}+160z^{95}+364z^{100}+96z^{110}$ and $\frac{w_{\min}}{w_{\max}}=\frac{85}{110}<\frac{4}{5}$. This result is confirmed by a Magma program.
\end{example}

\begin{table}[!h]
\begin{center}
\caption{The weight distribution of $\mathcal{C}_{f_D}$ in Theorem \ref{theorem-42}.}\label{tab-42}
\begin{tabular}{cc} \hline
Weight   &  Frequency   \\ \hline
0 & 1\\

$p^{n-m}+\frac{p^m-1}{2}(p^{n-m}-\upsilon (-1)^{m-1}\epsilon^m\eta_{m}(-1)p^{\frac{n-m}{2}})$          &  $p^e-1$\\

$p^n-p^{n-e}$  & $p^n-1$ \\

$p^n-p^{n-e}-\frac{p^m-1}{2p^m}\upsilon (-1)^{m-1}\epsilon^m\eta_{m}(-1)p^{\frac{n+m}{2}}$     & $(p^e-1)(p^{n-m}-1)$ \\

$p^n-p^{n-e}+(1+\eta_{m}(-1))\frac{1}{2p^m}\upsilon (-1)^{m-1}\epsilon^m\eta_{m}(-1)p^{\frac{n+m}{2}}$  & $\frac{(p^e-1)(p^m-1)}{2}\left(p^{n-m}+\upsilon^{-1} (-1)^{m-1}\epsilon^m\eta_{m}(-1)p^{\frac{n-m}{2}}\right)$\\
$p^n-p^{n-e}+(\eta_{m}(-1)-1)\frac{1}{2p^m}\upsilon (-1)^{m-1}\epsilon^m\eta_{m}(-1)p^{\frac{n+m}{2}}$  & $\frac{(p^e-1)(p^m-1)}{2}\left(p^{n-m}-\upsilon^{-1} (-1)^{m-1}\epsilon^m\eta_{m}(-1)p^{\frac{n-m}{2}}\right)$\\
\hline
\end{tabular}
\end{center}
\end{table}

 \begin{theorem}\label{theorem-42}
Let $q=p^e>3$ be an odd prime power, $n\geq 3m$, $e\mid m \mid n$, and let $F(x): V_n^{(p)} \rightarrow \gf_{p^m}$ be a vectorial dual-bent function satisfying \textbf{Condition A}. Define $D = \{ x \in V_n^{(p)} : F(x) \in NSQ \cup \{0\} \}$. Then the code $\mathcal{C}_{f_D}$ defined in Equation~(\ref{eq-3}) is a self-orthogonal minimal four-weight code with parameters
\[
\left[ p^n,\, r+1,\, p^{n-m} + \frac{p^m-1}{2}\left( p^{n-m} - \upsilon (-1)^{m-1}\epsilon^m \eta_{m}(-1) p^{\frac{n-m}{2}} \right) \right],
\]
satisfying $\dfrac{w_{\min}}{w_{\max}} < \dfrac{q-1}{q}$. The weight distribution of $\mathcal{C}_{f_D}$ is given in Table~\ref{tab-42}, where $\upsilon \in \{\pm\epsilon^{m}\}$ is a constant and $\epsilon := \sqrt{(-1)^{\frac{p-1}{2}}}$.
\end{theorem}

\begin{IEEEproof}
According to Lemmas~\ref{lem-46} and~\ref{the-41}, its weight distribution can be easily obtained. The rest of the proof is similar to that of Theorem~\ref{theorem-41} and is therefore omitted.
\end{IEEEproof}

In what follows, we use $s$-plateaued functions to construct self-orthogonal minimal codes that violate the Ashikhmin–Barg condition. In the remainder of this subsection, we set $e=1$ and regard $V_n^{(p)}$ as an $n$-dimensional vector space over $\mathbb{F}_{p}$. Then $\mathcal{C}_{f_D}$, defined in Equation (\ref{eq-3}), is a $p$-ary linear code.

\begin{table}[!h]
\begin{center}
\caption{The weight distribution of $\mathcal{C}_{f_D}$ in Theorem \ref{theorem-43} if $n+s$ is even.}\label{tab-43}
\begin{tabular}{ccc} \hline
Weight   &  Frequency (if $0 \in B_{+}(f)$) & Frequency (if $0 \in B_{-}(f)$)  \\ \hline
$0$ & $1$ & $1$\\

$\frac{p^n+p^{n-1}}{2}+\epsilon_0\frac{p-1}{2}p^{\frac{n+s}{2}-1}$        &  $p-1$ &  $p-1$\\

$p^n-p^{n-1}$  & $p^{n+1}-(p-1)p^{n-s}-1$ & $p^{n+1}-(p-1)p^{n-s}-1$\\

$p^n-p^{n-1}+\frac{p-1}{2p}p^{\frac{n+s}{2}}$     & $(p-1)(\frac{k}{p}+(p-1)p^{\frac{n-s}{2}}-1)+\frac{(p-1)^2}{2}(\frac{k}{p}-p^{\frac{n-s}{2}-1})$ &$(p-1)\frac{k}{p}+\frac{(p-1)^2}{2}\frac{k}{p}$\\

$p^n-p^{n-1}-\frac{p-1}{2p}p^{\frac{n+s}{2}}$     & $(p-1)(p^{n-s-1}-\frac{k}{p})+\frac{(p-1)^2}{2}(p^{n-s-1}-\frac{k}{p})$ &$(p-1)(\frac{p+1}{2}p^{n-s-1}-\frac{p+1}{2}\frac{k}{p}-\frac{p-1}{2}p^{\frac{n-s-1}{2}}-1)$\\
$p^n-p^{n-1}-\frac{p+1}{2p}p^{\frac{n+s}{2}}$& $\frac{(p-1)^2}{2}(\frac{k}{p}-p^{\frac{n-s}{2}-1})$&$\frac{(p-1)^2}{2}\frac{k}{p}$ \\
$p^n-p^{n-1}+\frac{p+1}{2p}p^{\frac{n+s}{2}}$& $\frac{(p-1)^2}{2}(p^{n-s-1}-\frac{k}{p})$&$\frac{(p-1)^2}{2}(p^{n-s-1}-\frac{k}{p}+p^{\frac{n-s}{2}-1})$ \\
\hline
\end{tabular}
\end{center}
\end{table}

\begin{table}[!h]
\begin{center}
\caption{The weight distribution of $\mathcal{C}_{f_D}$ in Theorem \ref{theorem-43} if $n+s$ is odd and $p\equiv 1 \mod 4$.}\label{tab-44}
\begin{tabular}{ccc} \hline
Weight   &  Frequency (if $0 \in B_{+}(f)$) & Frequency (if $0 \in B_{-}(f)$)  \\ \hline
$0$ & $1$ & $1$\\

$\frac{p^n+p^{n-1}}{2}+\frac{p-1}{2}\epsilon_{0}p^{\frac{n+s-1}{2}}$        &  $p-1$ &  $p-1$\\

$p^n-p^{n-1}$  & $p^{n+1}-(p-1)p^{n-s}-1+\frac{(p-1)^2}{2}(p^{n-s-1}+p^{\frac{n-s-1}{2}})$ & $p^{n+1}-(p-1)p^{n-s}-1+\frac{(p-1)^2}{2}(p^{n-s-1}-p^{\frac{n-s-1}{2}})$\\

$p^n-p^{n-1}+\frac{p-1}{2}p^{\frac{n+s-1}{2}}$     & $(p-1)(\frac{k}{p}-1)$ &$(p-1)\frac{k}{p}$\\

$p^n-p^{n-1}-\frac{p-1}{2}p^{\frac{n+s-1}{2}}$     & $(p-1)(p^{n-s-1}-\frac{k}{p})$ &$(p-1)(p^{n-s-1}-\frac{k}{p}-1)$\\

$p^n-p^{n-1}-p^{\frac{n+s-1}{2}}$& $\frac{(p-1)^2}{2}(\frac{k}{p}-p^{\frac{n-s-1}{2}})$&$\frac{(p-1)^2}{2}\frac{k}{p}$ \\

$p^n-p^{n-1}+p^{\frac{n+s-1}{2}}$& $\frac{(p-1)^2}{2}(p^{n-s-1}-\frac{k}{p})$&$\frac{(p-1)^2}{2}(p^{n-s-1}-\frac{k}{p}+p^{\frac{n-s-1}{2}})$ \\
\hline
\end{tabular}
\end{center}
\end{table} 

\begin{table}[!h]
\begin{center}
\caption{The weight distribution of $\mathcal{C}_{f_D}$ in Theorem \ref{theorem-43} if $n+s$ is odd and $p\equiv 3 \mod 4$.}\label{tab-45}
\begin{tabular}{ccc} \hline
Weight   &  Frequency (if $0 \in B_{+}(f)$) & Frequency (if $0 \in B_{-}(f)$)  \\ \hline
$0$ & $1$ & $1$\\

$\frac{p^n+p^{n-1}}{2}+\frac{p-1}{2}\epsilon_{0}p^{\frac{n+s-1}{2}}$        &  $p-1$ &  $p-1$\\

$p^n-p^{n-1}$  & $p^{n+1}-(p-1)p^{n-s}-1+\frac{(p-1)^2}{2}(p^{n-s-1}+p^{\frac{n-s-1}{2}})$ & $p^{n+1}-(p-1)p^{n-s}-1+\frac{(p-1)^2}{2}(p^{n-s-1}-p^{\frac{n-s-1}{2}})$\\

$p^n-p^{n-1}+\frac{p-1}{2}p^{\frac{n+s-1}{2}}$      & $(p-1)(\frac{k}{p}-1)$ &$(p-1)(\frac{k}{p})$\\

$p^n-p^{n-1}-\frac{p-1}{2}p^{\frac{n+s-1}{2}}$     & $(p-1)(p^{n-s-1}-\frac{k}{p})$ &$(p-1)(p^{n-s-1}-\frac{k}{p}-1)$\\

$p^n-p^{n-1}-p^{\frac{n+s-1}{2}}$& $\frac{(p-1)^2}{2}(\frac{k}{p}-p^{\frac{n-s-1}{2}})$&$\frac{(p-1)^2}{2}\frac{k}{p}$ \\

$p^n-p^{n-1}+p^{\frac{n+s-1}{2}}$& $\frac{(p-1)^2}{2}(p^{n-s-1}-\frac{k}{p})$&$\frac{(p-1)^2}{2}(p^{n-s-1}-\frac{k}{p}+p^{\frac{n-s-1}{2}})$ \\
\hline
\end{tabular}
\end{center}
\end{table}

 \begin{theorem}\label{theorem-43}
Let $p>3$ be an odd prime, $e=1$, and let $n$ and $s$ be positive integers with $n-s>2$. Let $f(x): V_{n}^{(p)}\rightarrow \gf_{p}$ be an $s$-plateaued function belonging to $\mathscr{F}$, and define $D=\{x\in V_{n}^{(p)}: f(x) \in SQ\bigcup \{0\}\}$. Let $\epsilon_0\in \{\pm 1\}$ be defined in Subsection \ref{subsectionF}. Then we have the following results.
\begin{enumerate}
\item If $n+s$ is even, then the code $\mathcal{C}_{f_D}$ defined in Equation (\ref{eq-3}) is a $\left[p^n, n+1,\frac{p^n+p^{n-1}}{2}+\epsilon_0\frac{p-1}{2}p^{\frac{n+s}{2}-1}\right]$ self-orthogonal minimal six-weight code over $\gf_p$ satisfying $\frac{w_{\min}}{w_{\max}}<\frac{p-1}{p}$. The weight distribution of $\mathcal{C}_{f_D}$ is given in Table \ref{tab-43}.
\item If $n+s$ is odd, then the code $\mathcal{C}_{f_D}$ defined in Equation (\ref{eq-3}) is a $\left[p^n, n+1,p^{n-1}+\frac{p-1}{2}\bigl(p^{n-1}+\epsilon_{0}p^{\frac{n+s-1}{2}}\bigr)\right]$ self-orthogonal minimal six-weight code over $\gf_p$ satisfying $\frac{w_{\min}}{w_{\max}}<\frac{p-1}{p}$. The weight distribution of $\mathcal{C}_{f_D}$ is given in Table \ref{tab-44} for the case $p\equiv 1 \pmod{4}$ and in Table \ref{tab-45} for the case $p\equiv 3 \pmod{4}$, respectively.
\end{enumerate}
\end{theorem}

\begin{IEEEproof}
By the definition of $\mathscr{F}$, we have $f(ax)=a^{t_a}f(x)$, where $a \in \gf_{p}^*$ and $\gcd(t_a-1,p-1)=1$ ($t_a=t$ if $x \in B_{+}(f^*)$, and  $t_a=t'$ if $x \in B_{-}(f^*)$). Since $p$ is an odd prime, then $t_a$ is even. For $f(x) \in SQ\bigcup\{0\}$, we have $f(ax)=a^{t_a}f(x) \in SQ\bigcup\{0\}$. Thus $D$ is $\gf_{p}^*$-invariant. 
Since $p\mid |D|$ by Lemma \ref{lem-11}, it follows from Corollary~\ref{cor-41} that $\mathcal{C}_{f_D}$ is self-orthogonal.

For a codeword $\textbf{c}(a,b):=\left(a f_{D}(x)+\langle b,x\rangle_{n/1}\right)_{x\in V_{n}^{(p)}}$ with $(a,b)\in \gf_{p}\times V_{n}^{(p)}$, by Lemma \ref{lem-43} and Theorem \ref{the-41}, we have the following results:
\begin{enumerate}
\item If $n+s$ is even, then we have
 \begin{eqnarray*}
&& \wt\left(\mathcal{C}_{f_{D}}\right)\\
&=& \left\{
\begin{array}{lll}
0&\mbox{if $a=0$, $b= 0$ },\\
\frac{p^n+p^{n-1}}{2}+\epsilon_0\frac{p-1}{2}p^{\frac{n+s}{2}-1}  &   \mbox{if $a\neq0$, $b= 0$ },\\
p^n-p^{n-1} &  \mbox{if $a=0$, $b \neq 0$ or $a\neq0$, $b \neq 0$ $b \notin S_f$},\\
p^n-p^{n-1}+\frac{p-1}{2p}\epsilon_{b}p^{\frac{n+s}{2}} & \mbox{if $a\neq0,  b \neq 0$, $b \in S_f$, $f^*(b)=0$}\\
 &\mbox{or $a\neq0,  b \neq 0$, $b \in S_f$, $f^*(b)\in SQ$,}\\
p^n-p^{n-1}-\frac{p+1}{2p}\epsilon_{b}p^{\frac{n+s}{2}} &  \mbox{if $a\neq0,  b \neq 0$, $b \in S_f$, $f^*(b)\in NSQ$,}\\
\end{array} \right.
 \end{eqnarray*}
 where $\epsilon_b=\pm1 $ be defined in Subsection \ref{subsectionF}.
 
\item If $n+s$ is odd, then we have
  \begin{eqnarray*}
&& \wt\left(\mathcal{C}_{f_{D}}\right)\\
&=& \left\{
\begin{array}{lll}
0&\mbox{if $a=0$, $b= 0$ },\\
p^{n-1}+\frac{p-1}{2}(p^{n-1}+\epsilon_{0}p^{\frac{n+s-1}{2}})  &   \mbox{if $a\neq0$, $b= 0$ },\\
p^n-p^{n-1} &  \mbox{if $a=0$, $b \neq 0$ or $a\neq0$, $b \neq 0$ $b \notin S_f$},\\
p^n-p^{n-1}+\frac{p-1}{2}\epsilon_{b}p^{\frac{n+s-1}{2}} & \mbox{if $a\neq0,  b \neq 0$, $b \in S_f$, $f^*(b)=0$,}\\
p^n-p^{n-1}+\frac{1}{2}(\eta_{0}(-1)-1)\epsilon_{b}p^{\frac{n+s-1}{2}} &  \mbox{if $a\neq0,  b \neq 0$, $b \in S_f$, $f^*(b)\in SQ$,}\\
p^n-p^{n-1}+\frac{1}{2}(-\eta_{0}(-1)-1)\epsilon_{b}p^{\frac{n+s-1}{2}} &  \mbox{if $a\neq0,  b \neq 0$, $b \in S_f$, $f^*(b)\in NSQ$,}\\
\end{array} \right.
 \end{eqnarray*}
  where $\epsilon_b=\pm1 $ be defined in Subsection \ref{subsectionF}.
 \end{enumerate}

 In the following, we consider the case where  $n+s$ is even. Firstly, we determine the frequency of each weight of $\mathcal{C}_{f_D}$. 
 Denote by $$\wt_0=0,$$ 
 $$\wt_1=\frac{p^n+p^{n-1}}{2}+\epsilon_0\frac{p-1}{2}p^{\frac{n+s}{2}-1},$$
 $$\wt_2=p^n-p^{n-1},$$ 
 $$\wt_3=p^n-p^{n-1}+\frac{p-1}{2p}p^{\frac{n+s}{2}},$$ 
 $$\wt_4=p^n-p^{n-1}-\frac{p-1}{2p}p^{\frac{n+s}{2}},$$
  $$\wt_5=p^n-p^{n-1}-\frac{p+1}{2p}p^{\frac{n+s}{2}},$$
   and $$\wt_6=p^n-p^{n-1}+\frac{p+1}{2p}p^{\frac{n+s}{2}}.$$
  It is obvious that $A_{\wt_0}=1$, $A_{\wt_1}=p-1$ and $A_{\wt_2}=p^{n+1}-(p-1)p^{n-s}-1$. Note that
 \begin{eqnarray*}
 & A_{w_3}=\left|\left\{(a,b)\in \gf_{p}^*\times V_{n}^{(p)}:a \neq 0, b \neq 0, b\in B_{+}(f), b\in S_{f}, f^*(b)=0\right\}\right|\\
 & +\left|\left\{(a,b)\in \gf_{p}^*\times V_{n}^{(p)}:a \neq 0, b \neq 0, b\in B_{+}(f), b\in S_{f}, f^*(b)\in SQ\right\}\right|.
\end{eqnarray*}
According to Lemma \ref{lem-13}, when $0\in B_{+}(f)$, then we have $A_{\wt_3}=(p-1)(\frac{k}{p}+(p-1)p^{\frac{n-s}{2}-1}-1)+\frac{(p-1)^2}{2}(\frac{k}{p}-p^{\frac{n-s}{2}-1})$;
 when $0\in B_{-}(f)$, then we have $A_{\wt_3}=(p-1)\frac{k}{p}+\frac{(p-1)^2}{2}\frac{k}{p}$. Besides, it is clear that
\begin{eqnarray*}&A_{\wt_4}=\left|\left\{(a,b)\in \gf_{p}^*\times V_{n}^{(p)}:a \neq 0, b \neq 0, b\in B_{-}(f), b\in S_{f}, f^*(b)=0\right\}\right|\\
&+\left|\left\{(a,b)\in \gf_{p}^*\times V_{n}^{(p)}:a \neq 0, b \neq 0, b\in B_{-}(f), b\in S_{f}, f^*(b)\in SQ\right\}\right|. 
\end{eqnarray*}
 According to Lemma \ref{lem-13}, when $0\in B_{+}(f)$, then we have $A_{\wt_4}=(p-1)(p^{n-s-1}-\frac{k}{p})+\frac{(p-1)^2}{2}(p^{n-s-1}-\frac{k}{p})$;
 when $0\in B_{-}(f)$, then we have $A_{\wt_4}=(p-1)(p^{n-s-1}-\frac{k}{p}-(p-1)p^{\frac{n-s}{2}-1}-1)+(p-1)\frac{p-1}{2}(p^{n-s-1}-\frac{k}{p}+p^{\frac{n-s}{2}-1})$.
 For
 $$A_{\wt_5}=\left|\left\{(a,b)\in \gf_{p}^*\times V_{n}^{(p)}:a \neq 0, b \neq 0, b\in B_{+}(f), b\in S_{f}, f^*(b)\in NSQ\right\}\right|,$$  by Lemma \ref{lem-13}, we have $A_{\wt_5}=\frac{(p-1)^2}{2}(\frac{k}{p}-p^{\frac{n-s}{2}-1})$ if $0\in B_{+}(f)$, and $A_{w_5}=\frac{(p-1)^2}{2}\frac{k}{p}$ if $0\in B_{-}(f)$.
 For
 $$A_{\wt_6}=\left|\left\{(a,b)\in \gf_{p}^*\times V_{n}^{(p)}:a \neq 0, b \neq 0, b\in B_{-}(f), b\in S_{f}, f^*(b)\in NSQ\right\}\right|,$$  
 by Lemma \ref{lem-13}, we have $A_{\wt_6}=\frac{(p-1)^2}{2}(p^{n-s-1}-\frac{k}{p})$ if $0\in B_{+}(f)$, and $A_{\wt_6}=\frac{(p-1)^2}{2}(p^{n-s-1}-\frac{k}{p}+p^{\frac{n-s}{2}-1})$ if $0\in B_{-}(f)$.
 Now we prove that $\mathcal{C}_{f_D}$ is a minimal code violating the Ashikhmin-Barg condition. It is obvious that
 \begin{eqnarray*}
 \frac{\wt_{\min}}{\wt_{\max}}=\frac{\frac{p^n+p^{n-1}}{2}+\epsilon_0\frac{p-1}{2}p^{\frac{n+s}{2}-1}}{p^n-p^{n-1}+\frac{p+1}{2p}p^{\frac{n+s}{2}}}<\frac{p-1}{p}.
 \end{eqnarray*}
For $i = 1,2,3,4,5,6$, define $E_i$ as the set of all codewords of weight $\wt_i$.
\begin{enumerate}
\item $E_1=\{ \mathbf{c}_{a, 0} \in \mathcal{C}_{f_D} | a\neq0, b=0 \}$,
\item $E_2=\{ \mathbf{c}_{a, b} \in \mathcal{C}_{f_D} | a=0, b \neq0 ~\mbox{or} ~ a\neq0, b \neq 0, b \notin S_f\}$,
\item $E_3\cup E_4=\{ \mathbf{c}_{a, b} \in \mathcal{C}_{f_D}| a\neq0, b\neq 0,b \in S_f, f^*(b)=0 ~\mbox{or}~ a \neq 0, b \in S_f, b \neq 0, f^*(b)\in SQ\}$,
\item $E_5\cup E_6=\{ \mathbf{c}_{a, b} \in \mathcal{C}_{f_D}| a \neq 0, b \neq 0, b \in S_f, f^*(b)\in NSQ\}$.    
\end{enumerate}
Note that $$\wt_{6}>\wt_{3}>\wt_{2}>\wt_{4}>\wt_{5}>\wt_{1}.$$
Let $(a_i, b_i) \in \mathbb{F}_p \times V_n^{(p)} \setminus \{(0,0)\}$ for $1 \leq i \leq 2$. Then, for any two linearly independent codewords $\mathbf{a} = \mathbf{c}(a_1, b_1)$ and $\mathbf{b} = \mathbf{c}(a_2, b_2)$, their coverage can be divided into the following four cases.

Case 1: Let $\mathbf{a}=\mathbf{c}(a_1,b_1) \in E_1$. Then $b_1=0$. For $\mathbf{b}=\mathbf{c}(a_2,b_2)$, we have $b_2\neq 0$ as $\mathbf{a}$ and $\mathbf{b}$
are $\gf_q$-linearly independent. Note that $\mathbf{a}+z\mathbf{b}=\mathbf{c}(a_1+za_2,zb_2)$ for $z\in \gf_p^*$.
Then we have $$\sum_{z \in \gf_{p}^*}\wt(\mathbf{a}+z\mathbf{b})>(p-1)\wt_1>(p-1)\wt_1-\wt_1>(p-1)\wt(\mathbf{a})-\wt(\mathbf{b}),$$
which satisfies the inequality in Lemma \ref{minimal1}.

Case 2: Let $\mathbf{a}=\mathbf{c}(a_1,b_1) \in E_2$. Then $b_1\neq 0$. We now consider the following subcases.

Subcase 2.1: If $\mathbf{b}=\mathbf{c}(a_2,b_2) \in E_1$, then $a_2\neq 0$ and $b_2=0$. Note that $\mathbf{a}+z\mathbf{b}=\mathbf{c}(a_1+za_2,b_1)  \notin E_{1}$ due to $b_1\neq 0$. Then we have
$$\sum_{z \in \gf_{p}^*}\wt(\mathbf{a}+z\mathbf{b})\geq(p-1)\wt_5=(p-1)\left(p^n-p^{n-1}-\frac{p+1}{2p}p^{\frac{n+s}{2}}\right).$$
Since
$$(p-1)\wt(\mathbf{a})-\wt(\mathbf{b})=(p-1)\wt_2-\wt_1=(p-1)(p^n-p^{n-1})-\left(\frac{p^n+p^{n-1}}{2}+\epsilon_0\frac{p-1}{2}p^{\frac{n+s}{2}-1}\right),$$
we deduce
$$\sum_{z \in \gf_{p}^*}\wt(\mathbf{a}+z\mathbf{b})>(p-1)\wt(\mathbf{a})-\wt(\mathbf{b}),$$
which satisfies the inequality in Lemma \ref{minimal1}.

Subcase 2.2: Let $\mathbf{b}=\mathbf{c}(a_2,b_2)\in \bigcup_{i=2}^{6}E_i$. Then $b_2\neq 0$. Note that $\mathbf{a}+z\mathbf{b}=\mathbf{c}(a_1+za_2,b_1+zb_2)$.
There exists at most one $z\in \gf_p^*$ such that $\mathbf{a}+z\mathbf{b}\in E_1$. 
Then we have 
$$\sum_{z \in \gf_{p}^*}\wt(\mathbf{a}+z\mathbf{b})\geq(p-2)\wt_5+\wt_1=(p-2)\left(p^n-p^{n-1}-\frac{p+1}{2p}p^{\frac{n+s}{2}}\right)+\frac{p^n+p^{n-1}}{2}+\epsilon_0\frac{p-1}{2}p^{\frac{n+s}{2}-1}.$$
Since
$$(p-1)\wt(\mathbf{a})-\wt(\mathbf{b})\leq (p-1)\wt_2-\wt_5= (p-1)(p^n-p^{n-1})-\left(p^n-p^{n-1}-\frac{p+1}{2p}p^{\frac{n+s}{2}}\right),$$
we have
$$\sum_{z \in \gf_{p}^*}\wt(\mathbf{a}+z\mathbf{b})>(p-1)\wt(\mathbf{a})-\wt(\mathbf{b}),$$
which satisfies the inequality in Lemma \ref{minimal1}.

Case 3: Let $\mathbf{a}=\mathbf{c}(a_1,b_1) \in E_3\cup E_4$. Then we have $a_1\neq 0$ and $b_1\neq 0$.

Subcase 3.1: Let $\mathbf{b}=\mathbf{c}(a_2,b_2) \in E_1$. Then we have $a_2\neq 0$ and $b_2=0$. 
Note that $\mathbf{a}+z\mathbf{b}=\mathbf{c}(a_1+za_2,b_1)$ for $z\in \gf_p^*$.
There exists exactly one $z\in \gf_p^*$ such that $\mathbf{a}+z\mathbf{b}\in E_2$. 
If $a_1+za_2\neq 0$, then it is easy to see $\mathbf{a}+z\mathbf{b}\in E_3\cup E_4$ and $\wt(\mathbf{a}+z\mathbf{b})=\wt(\mathbf{a})$. 
Then we have
$$\sum_{z \in \gf_{p}^*}\wt(\mathbf{a}+z\mathbf{b})=(p-2)\wt(\mathbf{a})+\wt_2=(p-2)\wt(\mathbf{a})+(p^n-p^{n-1}).$$
Since
$$(p-1)\wt(\mathbf{a})-\wt(\mathbf{b})=(p-1)\wt(\mathbf{a})-\left(\frac{p^n+p^{n-1}}{2}+\epsilon_0\frac{p-1}{2}p^{\frac{n+s}{2}-1}\right),$$
we derive
$$\sum_{z \in \gf_{p}^*}\wt(\mathbf{a}+z\mathbf{b})>(p-1)\wt(\mathbf{a})-\wt(\mathbf{b}),$$
which satisfies the inequality in Lemma \ref{minimal1}.

Subcase 3.2:  Let $\mathbf{b}=\mathbf{c}(a_2,b_2) \in \bigcup_{i=2}^{6}E_i$. Then we have $b_2\neq 0$.  Note that $\mathbf{a}+z\mathbf{b}=\mathbf{c}(a_1+za_2,b_1+zb_2)$ for $z\in \gf_p^*$.
Next, we will split the discussion depending on whether there exists some $z$ such that $\mathbf{a} + z\mathbf{b} \in E_1$.

Subcase 3.2.1: If there exists some $z_0\in \gf_p^*$ such that $\mathbf{a}+z_0\mathbf{b}=\mathbf{c}(a_1+z_0a_2,b_1+z_1b_2)\in E_1$, then $b_1+z_0b_2=0$ and $b_1+zb_2=-z_0b_2+zb_2=(z-z_0)b_2$. We also obtain that  $b_1+zb_2=-{z_0}^{-1}(z-z_0)b_1$ and $b_1+zb_2$, $b_1$ are linearly dependent for $z \neq z_0$. According to Lemma~\ref{YWei}, since $\mathbf{a}=\mathbf{c}(a_1,b_1) \in E_3\cup E_4$ and $b_1+zb_2=-{z_0}^{-1}(z-z_0)b_1$, we have the following: if $b_1 \in S_f$ and $f^*(b_1)=0$, then $b_1+zb_2 \in S_f$ and $f^*(b_1+zb_2)=0$; if $b_1 \in S_f$ and $f^*(b_1)\in SQ$, then $b_1+zb_2 \in S_f$ and $f^*(b_1+zb_2)\in SQ$.
Specially, we have $\epsilon_{b_1+zb_2}=\epsilon_{b_1}$. When $z$ runs over $\gf_{p}^*$, there exist at most one $z_1 \in \gf_{p}^*$ such that $a_1+z_1a_2=0$ where $z_1 \neq z_0$ . When $z\neq z_0, z_1$, we have $\mathbf{a}+z\mathbf{b} \in E_3\cup E_4$ and $\wt(\mathbf{a}+z\mathbf{b})=\wt(\mathbf{a})$. 
The following hold:
\begin{enumerate}
\item If $z_1$ exists, we have 
$$\sum_{z \in \gf_{p}^*}\wt(\mathbf{a}+z\mathbf{b})=(p-3)\wt(\mathbf{a})+\wt_1+\wt_2=(p-3)\wt(\mathbf{a})+\frac{p^n+p^{n-1}}{2}+\epsilon_0\frac{p-1}{2}p^{\frac{n+s}{2}-1}+p^n-p^{n-1}.$$

From Lemma~\ref{YWei} and the relation $b_1 = -z_0 b_2$, we obtain the following: if $b_1 \in S_f$ and $f^*(b_1) = 0$, then $b_2 \in S_f$ and $f^*(b_2) = 0$; and if $b_1 \in S_f$ and $f^*(b_1) \in SQ$, then $b_2 \in S_f$ and $f^*(b_2) \in SQ$. In particular, $\epsilon_{b_1} = \epsilon_{b_2}$.
Consequently, for $a_2 \neq 0$, we have $\mathbf{b} = \mathbf{c}(a_2, b_2) \in E_3 \cup E_4$ and $\wt(\mathbf{b}) = \wt(\mathbf{a})$. If $a_2 = 0$, then $\mathbf{b} \in E_2$.
If $\mathbf{b}=\mathbf{c}(a_2,b_2) \in E_2 $, then we have 
$$(p-1)\wt(\mathbf{a})-\wt(\mathbf{b})= (p-1)\wt(\mathbf{a})-(p^n-p^{n-1}).$$
If $\mathbf{b}=\mathbf{c}(a_2,b_2) \in E_3\cup E_4 $, then we have
$$(p-1)\wt(\mathbf{a})-\wt(\mathbf{b})= (p-1)\wt(\mathbf{a})-\wt(\mathbf{a}).$$
We then derive
$$\sum_{z \in 
\gf_{p}^*}\wt(\mathbf{a}+z\mathbf{b})>(p-1)\wt(\mathbf{a})-\wt(\mathbf{b}),$$
which satisfies the inequality in Lemma \ref{minimal1}.

\item If $z_1$ does not exist, we have
$$\sum_{z \in \gf_{p}^*}\wt(\mathbf{a}+z\mathbf{b})=(p-2)\wt(\mathbf{a})+\wt_1=(p-2)\wt(\mathbf{a})+\frac{p^n+p^{n-1}}{2}+\epsilon_0\frac{p-1}{2}p^{\frac{n+s}{2}-1}.$$
If $\mathbf{b}=\mathbf{c}(a_2,b_2) \in E_2 $, then we have 
$$(p-1)\wt(\mathbf{a})-\wt(\mathbf{b})= (p-1)\wt(\mathbf{a})-(p^n-p^{n-1}).$$
If $\mathbf{b}=\mathbf{c}(a_2,b_2) \in E_3\cup E_4 $, then we have
$$(p-1)\wt(\mathbf{a})-\wt(\mathbf{b})= (p-1)\wt(\mathbf{a})-\wt(\mathbf{a}).$$
We then derive
$$\sum_{z \in 
\gf_{p}^*}\wt(\mathbf{a}+z\mathbf{b})>(p-1)\wt(\mathbf{a})-\wt(\mathbf{b}),$$
which satisfies the inequality in Lemma \ref{minimal1}.
\end{enumerate}

Subcase 3.2.2: If there exists no $z\in \gf_p^*$ such that $\mathbf{a}+z\mathbf{b}=\mathbf{c}(a_1+za_2,b_1+zb_2)\in E_1$,
then we have
$$\sum_{z \in \gf_{p}^*}\wt(\mathbf{a}+z\mathbf{b})\geq (p-1)\wt_5=(p-1)\left(p^n-p^{n-1}-\frac{p+1}{2p}p^{\frac{n+s}{2}}\right).$$
Since
$$(p-1)\wt(\mathbf{a})-\wt(\mathbf{b})\leq (p-1)\wt(\mathbf{a})-(p^n-p^{n-1}-\frac{p+1}{2p}p^{\frac{n+s}{2}})\leq(p-1)(p^n-p^{n-1}+\frac{p-1}{2p}p^{\frac{n+s}{2}})-(p^n-p^{n-1}-\frac{p+1}{2p}p^{\frac{n+s}{2}}),$$
we derive
$$\sum_{z \in \gf_{p}^*}\wt(\mathbf{a}+z\mathbf{b})>(p-1)\wt(\mathbf{a})-\wt(\mathbf{b}),$$
which satisfies the inequality in Lemma \ref{minimal1}.

Case 4: Let $\mathbf{a}=\mathbf{c}(a_1,b_1) \in E_5\cup E_6$. Then we have $a_1\neq 0$ and $b_1\neq 0$. 

Subcase 4.1: Let $\mathbf{b}=\mathbf{c}(a_2,b_2) \in E_1$. Then we have $a_2\neq 0$ and $b_2=0$. Note that $\mathbf{a}+z\mathbf{b}=\mathbf{c}(a_1+za_2,b_1)$ for $z\in \gf_p^*$.
There exists exactly one $z\in \gf_p^*$ such that $\mathbf{a}+z\mathbf{b}\in E_2$. If $a_1+za_2\neq 0$, then it is easy to see $\mathbf{a}+z\mathbf{b}\in E_5\cup E_6$ and $\wt(\mathbf{a}+z\mathbf{b})=\wt(\mathbf{a})$. 
Then we have
$$\sum_{z \in \gf_{p}^*}\wt(\mathbf{a}+z\mathbf{b})=(p-2)\wt(\mathbf{a})+\wt_2=(p-2)\wt(\mathbf{a})+(p^n-p^{n-1}).$$
Since
$$(p-1)\wt(\mathbf{a})-\wt(\mathbf{b})=(p-1)\wt(\mathbf{a})-(\frac{p^n+p^{n-1}}{2}+\epsilon_0\frac{p-1}{2}p^{\frac{n+s}{2}-1}).$$
we derive
$$\sum_{z \in \gf_{p}^*}\wt(\mathbf{a}+z\mathbf{b})>(p-1)\wt(\mathbf{a})-\wt(\mathbf{b}),$$
which satisfies the inequality in Lemma \ref{minimal1}.

Subcase 4.2: Let $\mathbf{b}=\mathbf{c}(a_2,b_2) \in \bigcup_{i=2}^{6}E_i$. Then we have $b_2\neq 0$.  Note that $\mathbf{a}+z\mathbf{b}=\mathbf{c}(a_1+za_2,b_1+zb_2)$ for $z\in \gf_p^*$.
Next, we divide the discussion according to whether there exists $z$ such that $\mathbf{a}+z\mathbf{b}\in E_1$.

Subcase 4.2.1: If there exists some $z_0 \in \gf_{p}^*$ such that $\mathbf{a}+z\mathbf{b}=\mathbf{c}(a_1+z_0a_2,b_1+z_0b_2)\in E_1$, then  $b_1+z_0b_2=0$ and $b_1+zb_2=-z_0b_2+zb_2=(z-z_0)b_2$. We also have $b_1 + z b_2 = -z_0^{-1}(z - z_0) b_1$, and $b_1 + z b_2$ is linearly dependent on $b_1$ for $z \neq z_0$. By Lemma \ref{YWei}, since $\mathbf{a} = \mathbf{c}(a_1, b_1) \in E_5 \cup E_6$ and $b_1 + z b_2 = -z_0^{-1}(z - z_0) b_1$, it follows that if $b_1 \in S_f$ and $f^*(b_1) \in NSQ$, then $b_1 + z b_2 \in S_f$ and $f^*(b_1 + z b_2) \in NSQ$. In particular, $\epsilon_{b_1 + z b_2} = \epsilon_{b_1}$. As $z$ runs over $\mathbb{F}_p^*$, there is at most one $z_1 \in \mathbb{F}_p^*$ with $z_1 \neq z_0$ such that $a_1 + z_1 a_2 = 0$. For $z \neq z_0, z_1$, we have $\mathbf{a} + z\mathbf{b} \in E_5 \cup E_6$ and $\wt(\mathbf{a} + z\mathbf{b}) = \wt(\mathbf{a})$. The following hold:
\begin{enumerate}
\item If $z_1$ exists, we have 

$$\sum_{z \in \gf_{p}^*}\wt(\mathbf{a}+z\mathbf{b})=(p-3)\wt(\mathbf{a})+\wt_1+\wt_2=(p-3)\wt(\mathbf{a})+\frac{p^n+p^{n-1}}{2}+\epsilon_0\frac{p-1}{2}p^{\frac{n+s}{2}-1}+p^n-p^{n-1}.$$

From Lemma \ref{YWei} and the relation $b_1 = -z_0 b_2$, we obtain that if $b_1 \in S_f$ and $f^*(b_1) \in \operatorname{NSQ}$, then $b_2 \in S_f$ and $f^*(b_2) \in \operatorname{NSQ}$. In particular, $\epsilon_{b_1} = \epsilon_{b_2}$. Consequently, for $a_2 \neq 0$, we have $\mathbf{b} = \mathbf{c}(a_2, b_2) \in E_5 \cup E_6$ and $\wt(\mathbf{b}) = \wt(\mathbf{a})$. If $a_2 = 0$, then $\mathbf{b} \in E_2$.
If $\mathbf{b}=\mathbf{c}(a_2,b_2)  \in E_2 $, then 
$$(p-1)\wt(\mathbf{a})-\wt(\mathbf{b})= (p-1)\wt(\mathbf{a})-(p^n-p^{n-1}).$$
If $\mathbf{b}=\mathbf{c}(a_2,b_2)  \in E_5\cup E_6 $, then $\wt(\mathbf{a})=\wt(\mathbf{b})$, thus
$$(p-1)\wt(\mathbf{a})-\wt(\mathbf{b})= (p-1)\wt(\mathbf{a})-\wt(\mathbf{a}).$$
We then derive
$$\sum_{z \in \gf_{p}^*}\wt(\mathbf{a}+z\mathbf{b})>(p-1)\wt(\mathbf{a})-\wt(\mathbf{b}),$$
which satisfies the inequality in Lemma \ref{minimal1}.
\item If $z_1$ does not exist, we have 
$$\sum_{z \in \gf_{p}^*}\wt(\mathbf{a}+z\mathbf{b})=(p-2)\wt(\mathbf{a})+\wt_1=(p-2)\wt(\mathbf{a})+\frac{p^n+p^{n-1}}{2}+\epsilon_0\frac{p-1}{2}p^{\frac{n+s}{2}-1}.$$
If $\mathbf{b}=\mathbf{c}(a_2,b_2)  \in E_2 $, then 
$$(p-1)\wt(\mathbf{a})-\wt(\mathbf{b})= (p-1)\wt(\mathbf{a})-(p^n-p^{n-1}).$$
If $\mathbf{b}=\mathbf{c}(a_2,b_2)  \in E_5\cup E_6 $, then $\wt(\mathbf{a})=\wt(\mathbf{b})$ and
$$(p-1)\wt(\mathbf{a})-\wt(\mathbf{b})= (p-1)\wt(\mathbf{a})-\wt(\mathbf{a}).$$
We then derive
$$\sum_{z \in \gf_{p}^*}\wt(\mathbf{a}+z\mathbf{b})>(p-1)\wt(\mathbf{a})-\wt(\mathbf{b}),$$
which satisfies the inequality in Lemma \ref{minimal1}.
\end{enumerate}

Subcase 4.2.2: If there exits no $z\in \gf_p^*$ such that $\mathbf{a}+z\mathbf{b}=\mathbf{c}(a_1+za_2,b_1+zb_2)\in E_1$ , then we have
$$\sum_{z \in \gf_{p}^*}\wt(\mathbf{a}+z\mathbf{b})\geq (p-1)\wt_5=(p-1)\left(p^n-p^{n-1}-\frac{p+1}{2p}p^{\frac{n+s}{2}}\right).$$
Since
\begin{eqnarray*}(p-1)\wt(\mathbf{a})-\wt(\mathbf{b})\leq (p-1)\wt(\mathbf{a})-\left(p^n-p^{n-1}-\frac{p+1}{2p}p^{\frac{n+s}{2}}\right)\\
\leq(p-1)\left(p^n-p^{n-1}+\frac{p+1}{2p}p^{\frac{n+s}{2}}\right)
-\left(p^n-p^{n-1}-\frac{p+1}{2p}p^{\frac{n+s}{2}}\right),\end{eqnarray*}
we derive
$$\sum_{z \in \gf_{p}^*}\wt(\mathbf{a}+z\mathbf{b})>(p-1)\wt(\mathbf{a})-\wt(\mathbf{b}),$$
which satisfies the inequality in Lemma \ref{minimal1}.

Based on the above discussion and Lemma \ref{minimal1}, we conclude that if $n+s$ is even, then $\mathcal{C}_{f_D}$ is a minimal linear code satisfying $\frac{w_{\min}}{w_{\max}} < \frac{p-1}{p}$. The case where $n+s$ is odd follows similarly to the even case, and is therefore omitted.
 \end{IEEEproof}

\begin{example}
Let $p=5$, $n=4$ and define $f(x): \mathbb{F}_{5}^{4} \rightarrow \mathbb{F}_{5}$ by $f(x_1,x_2,x_3,x_4)=4x_{1}^2 x_{3}^4 + 2x_{1}^2 + x_2 x_3$. Then $f(x)$ is a non-weakly regular $1$-plateaued function of type $(-)$, with $\sharp B_{+}(f)=100$ and $t=t'=2$. Using Magma, one verifies that the linear code $\mathcal{C}_{f_D}$ in Theorem \ref{theorem-43} has parameters $[625,5,325]$ and weight enumerator
\[
1 + 4z^{325} + 16z^{450} + 160z^{475} + 2784z^{500} + 80z^{525} + 80z^{550},
\]
satisfying $\frac{w_{\min}}{w_{\max}} = \frac{325}{550} < \frac{4}{5}$. Moreover, Magma also confirms that $\mathcal{C}_{f_D}$ is self-orthogonal.
\end{example}

\begin{table}[!h]
\begin{center}
\caption{The weight distribution of $\mathcal{C}_{f_D}$ in Theorem \ref{theorem44} if $n+s$ is odd and $p\equiv 1 \mod 4$.}\label{tab-46}
\begin{tabular}{ccc} \hline
Weight   &  Frequency (if $0 \in B_{+}(f)$) & Frequency (if $0 \in B_{-}(f)$)  \\ \hline
$0$ & $1$ & $1$\\

$\frac{p^n+p^{n-1}}{2}-\frac{p-1}{2}\epsilon_{0}p^{\frac{n+s-1}{2}}$        &  $p-1$ &  $p-1$\\

$p^n-p^{n-1}$  & $p^{n+1}-(p-1)p^{n-s}-1+\frac{(p-1)^2}{2}(p^{n-s-1}-p^{\frac{n-s-1}{2}})$ & $p^{n+1}-(p-1)p^{n-s}-1+\frac{(p-1)^2}{2}(p^{n-s-1}+p^{\frac{n-s-1}{2}})$\\

$p^n-p^{n-1}-\frac{p-1}{2}p^{\frac{n+s-1}{2}}$     & $(p-1)(\frac{k}{p}-1)$ &$(p-1)\frac{k}{p}$\\

$p^n-p^{n-1}+\frac{p-1}{2}p^{\frac{n+s-1}{2}}$     & $(p-1)(p^{n-s-1}-\frac{k}{p})$ &$(p-1)(p^{n-s-1}-\frac{k}{p}-1)$\\

$p^n-p^{n-1}+p^{\frac{n+s-1}{2}}$& $\frac{(p-1)^2}{2}(\frac{k}{p}+p^{\frac{n-s-1}{2}})$&$\frac{(p-1)^2}{2}\frac{k}{p}$ \\

$p^n-p^{n-1}-p^{\frac{n+s-1}{2}}$& $\frac{(p-1)^2}{2}(p^{n-s-1}-\frac{k}{p})$&$\frac{(p-1)^2}{2}(p^{n-s-1}-\frac{k}{p}-p^{\frac{n-s-1}{2}})$ \\
\hline
\end{tabular}
\end{center}
\end{table} 

\begin{table}[!h]
\begin{center}
\caption{The weight distribution of $\mathcal{C}_{f_D}$ in Theorem \ref{theorem44} if $n+s$ is odd and $p\equiv 3 \mod 4$.}\label{tab-47}
\begin{tabular}{ccc} \hline
Weight   &  Frequency (if $0 \in B_{+}(f)$) & Frequency (if $0 \in B_{-}(f)$)  \\ \hline
$0$ & $1$ & $1$\\

$\frac{p^n+p^{n-1}}{2}-\frac{p-1}{2}\epsilon_{0}p^{\frac{n+s-1}{2}}$        &  $p-1$ &  $p-1$\\

$p^n-p^{n-1}$  & $p^{n+1}-(p-1)p^{n-s}-1+\frac{(p-1)^2}{2}(p^{n-s-1}-p^{\frac{n-s-1}{2}})$ & $p^{n+1}-(p-1)p^{n-s}-1+\frac{(p-1)^2}{2}(p^{n-s-1}+p^{\frac{n-s-1}{2}})$\\

$p^n-p^{n-1}-\frac{p-1}{2}p^{\frac{n+s-1}{2}}$      & $(p-1)(\frac{k}{p}-1)$ &$(p-1)\frac{k}{p}$\\

$p^n-p^{n-1}+\frac{p-1}{2}p^{\frac{n+s-1}{2}}$     & $(p-1)(p^{n-s-1}-\frac{k}{p})$ &$(p-1)(p^{n-s-1}-\frac{k}{p}-1)$\\

$p^n-p^{n-1}+p^{\frac{n+s-1}{2}}$& $\frac{(p-1)^2}{2}(\frac{k}{p}+p^{\frac{n-s-1}{2}})$&$\frac{(p-1)^2}{2}\frac{k}{p}$ \\

$p^n-p^{n-1}-p^{\frac{n+s-1}{2}}$& $\frac{(p-1)^2}{2}(p^{n-s-1}-\frac{k}{p})$&$\frac{(p-1)^2}{2}(p^{n-s-1}-\frac{k}{p}-p^{\frac{n-s-1}{2}})$ \\
\hline
\end{tabular}
\end{center}
\end{table} 

 \begin{theorem}\label{theorem44}
Let $p>3$ be an odd prime, $e=1$, $n$ and $s$ be positive integers such that $n-s>2$ is odd, and let $f(x): V_n^{(p)} \rightarrow \mathbb{F}_p$ be an $s$-plateaued function belonging to $\mathscr{F}$. Define $D = \{ x \in V_n^{(p)} : f(x) \in NSQ \cup \{0\} \}$. Let $\epsilon_0 \in \{\pm 1\}$ be as defined in Subsection \ref{subsectionF}. Then the code $\mathcal{C}_{f_D}$ defined in Equation (\ref{eq-3}) is a $[p^n, n+1, p^{n-1} + \frac{p-1}{2}(p^{n-1} - \epsilon_0 p^{\frac{n+s-1}{2}})]$ self-orthogonal minimal six-weight code over $\gf_p$ satisfying $\frac{\wt_{\min}}{\wt_{\max}} < \frac{p-1}{p}$, and its weight distribution is given in Tables \ref{tab-46} and \ref{tab-47} for the two respective cases.
\end{theorem}

\begin{IEEEproof}
According to Lemma \ref{lem-44} and Theorem \ref{the-41}, the weight distribution follows directly. The rest of the proof is similar to that of Theorem \ref{theorem-43} and is therefore omitted.
\end{IEEEproof}

\begin{table}[!h]
\begin{center}
\caption{The weight distribution of $\mathcal{C}_{f_D}$ in Theorem \ref{theorem45} if $n+s$ is even.}\label{tab-48}
\begin{tabular}{ccc} \hline
Weight   &  Frequency (if $0 \in B_{+}(f)$) & Frequency (if $0 \in B_{-}(f)$)  \\ \hline
$0$ & $1$ & $1$\\

$p^{n-1}+\epsilon_{0}(p-1)p^{\frac{n+s}{2}-1}$        &  $p-1$ &  $p-1$\\

$p^n-p^{n-1}$  & $p^{n+1}-(p-1)p^{n-s}-1$ & $p^{n+1}-(p-1)p^{n-s}-1$\\

$p^n-p^{n-1}+\frac{p-1}{p}p^{\frac{n+s}{2}}$     & $(p-1)(\frac{k}{p}+(p-1)p^{\frac{n-s}{2}-1}-1)$ &$(p-1)\frac{k}{p}$\\

$p^n-p^{n-1}-\frac{p-1}{p}p^{\frac{n+s}{2}}$     & $(p-1)(p^{n-s-1}-\frac{k}{p})$ &$(p-1)(p^{n-s-1}-\frac{k}{p}-(p-1)p^{\frac{n-s}{2}-1}-1)$\\

$p^n-p^{n-1}-\frac{1}{p}p^{\frac{n+s}{2}}$& $(p-1)^2(\frac{k}{p}-p^{\frac{n-s}{2}-1})$&$(p-1)^2\frac{k}{p}$ \\

$p^n-p^{n-1}+\frac{1}{p}p^{\frac{n+s}{2}}$& $(p-1)^2(p^{n-s-1}-\frac{k}{p})$&$(p-1)^2(p^{n-s-1}-\frac{k}{p}+p^{\frac{n-s}{2}-1})$ \\
\hline
\end{tabular}
\end{center}
\end{table}

\begin{table}[!h]
\begin{center}
\caption{The weight distribution of $\mathcal{C}_{f_D}$ in Theorem \ref{theorem45} if $n+s$ is odd.}\label{tab-49}
\begin{tabular}{cc} \hline
Weight   &  Frequency   \\ \hline
0 & 1\\

$p^{n-1}$          &  $p-1$\\

$p^n-p^{n-1}$  & $p^n-1+(p-1)(p^n-p^{n-1})+(p-1)(p^{n-s-1}-1)$ \\

$p^n-p^{n-1}+ p^{\frac{n+s-1}{2}}$     & $\frac{(p-1)^2}{2}(p^{n-s-1}+p^{\frac{n-s-1}{2}})$ \\

$p^n-p^{n-1}- p^{\frac{n+s-1}{2}}$  & $\frac{(p-1)^2}{2}(p^{n-s-1}-p^{\frac{n-s-1}{2}})$\\
\hline
\end{tabular}
\end{center}
\end{table}

%sn-s>4n-s>2
\begin{theorem}\label{theorem45}
Let $p=3$, $e=1$ and $f(x): V_{n}^{(p)}\rightarrow \gf_{p}$ be an $s$-plateaued function belonging to $\mathscr{F}$.  Let $D=\{x\in V_{n}^{(p)}: f(x)=0\}$. Let $\epsilon_0 \in \{\pm 1\}$ be as defined in Subsection \ref{subsectionF}. Then the following results hold. 
\begin{enumerate}
\item If $n+s$ is even and  $0\leq s< n-2$, then the code $\mathcal{C}_{f_D}$ defined in Equation (\ref{eq-3}) is a $[p^n, n+1,p^{n-1}+\epsilon_{0}(p-1)p^{\frac{n+s}{2}-1}]$ self-orthogonal minimal six-weight code over $\gf_p$ with $\frac{\wt{_{\min}}}{\wt{_{\max}}}<\frac{p-1}{p}$ and its weight distribution  is given in Table \ref{tab-48}.
\item If $n+s$ is odd and $0\leq s < n-2$, then the code $\mathcal{C}_{f_D}$ defined in Equation (\ref{eq-3}) is a $[p^n, n+1,p^{n-1}]$ self-orthogonal minimal four-weight code over $\gf_p$ with $\frac{\wt{_{\min}}}{\wt{_{\max}}}<\frac{p-1}{p}$ and its weight distribution is given in Table \ref{tab-49}.
\end{enumerate}
\end{theorem}

\begin{IEEEproof}
 If $f(x) = 0$, then $f(ax) = a^{t_a} f(x) = 0$ as $f(x)\in \mathscr{F}$. Hence $D$ is $\mathbb{F}_q^*$-invariant.
 Since $p\in |D|$,  it follows from Corollary~\ref{cor-41} that $\mathcal{C}_{f_D}$ is self-orthogonal.
For a codeword $\textbf{c}(a,b):=\left(a f_{D}(x)+\langle b,x\rangle_{n/1}\right)_{x\in V_{n}^{(p)}}$ with $(a,b)\in \gf_{3}\times V_{n}^{(3)}$, by Lemmas \ref{lem-42} and \ref{the-41}, we have the following:
\begin{enumerate}
\item If $n+s$ is even, then we have
 \begin{eqnarray*}
\wt\left(\textbf{c}(a,b)\right)&=& \left\{
\begin{array}{lll}
0&\mbox{if $a=0$, $b= 0$ },\\
p^{n-1}+\epsilon_{0}(p-1)p^{\frac{n+s}{2}-1},  &   \mbox{if $a\neq0$, $b= 0$ },\\
p^n-p^{n-1}, &  \mbox{if $a=0$, $b \neq 0$ or $a\neq0$, $b \neq 0$ $b \notin S_f$},\\
p^n-p^{n-1}+\frac{p-1}{p}\epsilon_{b}p^{\frac{n+s}{2}}, & \mbox{if $a\neq0,  b \neq 0$, $b \in S_f$, $f^*(b)=0$ }\\
 &\mbox{or $a\neq0, b \neq 0$, $b \in S_f$, $f^*(b)\in SQ$,}\\
p^n-p^{n-1}-\frac{1}{p}\epsilon_{b}p^{\frac{n+s}{2}}, &  \mbox{if $a\neq0, b \neq 0$, $b \in S_f$, $f^*(b)\in NSQ$.}\\
\end{array} \right.
 \end{eqnarray*}
 If $n+s$ is odd, then we have
  \begin{eqnarray*}
&& \wt\left(\textbf{c}(a,b)\right)\\
&=& \left\{
\begin{array}{lll}
0&\mbox{if $a=0$, $b= 0$ },\\
p^{n-1},  &   \mbox{if $a\neq0$, $b= 0$ },\\
p^n-p^{n-1}, &  \mbox{if $a=0$, $b \neq 0$ or $a\neq0$, $b \neq 0$ $b \notin S_f$ or $a\neq0,  b \neq 0$, $b \in S_f$, $f^*(b)=0$ },\\
p^n-p^{n-1}+\eta_{0}(-1)\epsilon_{b}p^{\frac{n+s-1}{2}}, &  \mbox{if $a\neq0,  b \neq 0$, $b \in S_f$, $f^*(b)\in SQ$,}\\
p^n-p^{n-1}-\eta_{0}(-1)\epsilon_{b}p^{\frac{n+s-1}{2}}, &  \mbox{if $a\neq0,  b \neq 0$, $b \in S_f$, $f^*(b)\in NSQ$.}\\
\end{array} \right.
 \end{eqnarray*}
\end{enumerate}
By Lemma \ref{lem-12}, the determination of the frequencies $A_{w_i}$ is similar. The weight distribution of $\mathcal{C}_{f_D}$ is given in Tables \ref{tab-48} and \ref{tab-49} for the two cases. The proof of its minimality is analogous to that of Theorem \ref{theorem-43}. The details are omitted here.
  \end{IEEEproof}

\begin{example}
Let $p=3$, $n=6$ and define $f(x): \mathbb{F}_3^{6} \rightarrow \mathbb{F}_3$ by $f(x_1,x_2,x_3,x_4,x_5,x_6)=x_{1}^2 x_{5}^2 + x_{1}^2 + x_{2}^2 + x_{3}^2 + x_4 x_5$. Then $f(x)$ is a non-weakly regular $1$-plateaued function of type $(-)$, with $\sharp B_{+}(f)=162$ and $t=t'=2$. Using Magma, one verifies that the linear code $\mathcal{C}_{f_D}$ in Theorem \ref{theorem45} has parameters $[729,7,243]$ and weight enumerator
\[
1 + 2z^{243} + 144z^{495} + 1860z^{486} + 180z^{513},
\]
satisfying $\frac{w_{\min}}{w_{\max}} = \frac{243}{513} < \frac{2}{3}$. Moreover, Magma also confirms that $\mathcal{C}_{f_D}$ is self-orthogonal.
\end{example}

%\begin{remark}
%According to Lemma \ref{lemma-21}, for the set $-D=D$, it is a partial difference set if and only if it satisfies the Lmma \ref{lemma-21}. However, the character sum $\chi(D)$ of the
 %set defined in Lemma \ref{lem-42}, \ref{lem-43}, \ref{lem-44}, \ref{lem-45} and \ref{lem-46} has three‑valued eigenvalues where $\chi$ is nontrivial, they does not satisfy the conditions of the Lemma \ref{lemma-21}. Hence, the codes constructed in this section are different from those in \cite{Tao}. According to Corollary \ref{cor-41}, we can verify that minimal code in \cite{Tao} are self-orthogonal for $q>3$ and $p \mid |D|$. When $q=2,3$, if the character sum $\Psi(D)$ in \cite{Tao} satisfying Theorem \ref{The-32} and  \ref{The-33} and $p \mid |D|$, then the minimal code in \cite{Tao} are self-orthogonal. 
%\end{remark}

\section{Summary and concluding remarks}\label{sec7}
Determining whether a linear code is self-orthogonal is generally nontrivial. In this paper, we have established several criteria for the self-orthogonality of certain classes of linear codes. These criteria have enabled us to verify the self-orthogonality of some known codes and to construct new families of self-orthogonal codes and self-orthogonal minimal codes. Our results have provided an extended framework for constructing self-orthogonal codes and have highlighted their applications in quantum codes and minimal codes. Future work will extend these criteria to more general families of linear codes.
 
%\section{Simple References}

\end{document}